\documentclass[aip,jmp,nofootinbib]{revtex4-1}

\pdfoutput=1

\usepackage[utf8]{inputenc}
\usepackage{amsfonts} 
\usepackage{amsmath,amsthm,amssymb}
\usepackage{thmtools,thm-restate} 
\usepackage{hyperref} 
\usepackage[capitalize]{cleveref} 
\usepackage{tikz}
\usetikzlibrary{patterns,calc,decorations.pathreplacing}

\crefname{appsec}{Appendix}{Appendices} 

\newcommand{\id}{\mathrm{Id}}
\newcommand{\bra}[1]{\langle #1 \vert}
\newcommand{\ket}[1]{\vert #1 \rangle}

\newcommand{\scalprod}[2]{\langle #1 \vert #2 \rangle}
\newcommand{\Scalprod}[2]{\left \langle #1 \middle\vert #2 \right\rangle}
\newcommand{\tr}{\mathop{\mathrm{Tr}}}
\newcommand{\norm}[1]{\left \|#1 \right \|}

\newtheorem{theorem}{Theorem}
\newtheorem{prop}[theorem]{Proposition}
\newtheorem{lemma}[theorem]{Lemma}
\newtheorem{corollary}{Corollary}[theorem]

\theoremstyle{definition}
\newtheorem{definition}{Definition}

\def\d{0.15}
\def\dd{0.07}
\def\ddd{0.1}
\def\rad{0.04}

\tikzset{%
  state/.pic={%
      \draw (\d,\d) rectangle (1-\d,1-\d);
      \foreach \x in {(\d,\d), (\d,1-\d), (1-\d,\d), (1-\d,1-\d)}{
        \filldraw[black] \x circle (\rad);
      }
  }
}

\tikzset{
  pepstensor/.pic={
    \draw (-0.5,0)--(0.5,0);
    \draw (0,-0.5)--(0,0.5);
    \draw[fill=white] (-0.2,-0.2) rectangle (0.2,0.2); 
    \draw (0,0,0)--++(0,0,-0.75);
  }
}

\tikzset{
  pics/mpotensor/.style args={#1}{ 
    code={%
      \draw (0,-0.5)--(0,0.5);
      \draw (-0.5,0)--(0.5,0);
      \draw[fill=#1] (0,0) circle (0.25);
      \node[anchor=south west] at (0.05,0.1) {\tikzpictext};
    }
  },
  pics/mpotensor/.default=white
}

\tikzset{
  boundarympo/.pic={
    \draw (0,-0.3)--(0,0.3);
    \draw[black] (-0.5,0)--(0.5,0);
    \draw[black,fill=white] (0,0) circle (0.1);
    \node[black] at (0.0,0.0) {\tikzpictext};
  }
}

\tikzset{
  reduction/.pic={
    \draw (0,0)--(0.25,0);
    \draw[rounded corners=3pt] (0.25,-0.75) rectangle (0.5,0.75);
    \draw (0.5,-0.5) -- (0.75,-0.5);
    \draw (0.5,0.5) --(0.75,0.5);
  }
}

\tikzset{
  longreduction/.pic={
    \draw (0,-0.5)--(0.25,-0.5);
    \draw[rounded corners=3pt] (0.25,-1.25) rectangle (0.5,0.75);
    \draw (0.5,-1) -- (0.75,-1);
    \draw (0.5,0.5) --(0.75,0.5);
  }
}

\tikzset{
  nilpotent/.pic={
    \draw (0,-1)--(0,1);
    \draw (-0.5,-0.5)--(0.5,-0.5);
    \draw (-0.5,0.5)--(0.5,0.5);
    \draw[fill=white,rounded corners=3pt] (-0.15,-0.75) rectangle (0.15,0.75);
  }
}

\tikzset{
  pics/mps/.style args={#1}{ 
    code={
      \begin{scope}[rotate=#1]
        \draw (-0.5,0)--(0.5,0);
        \draw (0,0)--(0,0.5);
        \draw[fill=white] (-0.2,-0.2) rectangle (0.2,0.2); 
      \end{scope}
      \node[anchor=south west, font=\scriptsize] at (-0.55,0.1) {\tikzpictext};
    }
  },
  pics/mps/.default=0
}

\tikzset{
  mpsgauge/.pic={%
    \draw (-0.5,0)--(0.5,0);
    \draw[fill=white] (0,0) circle (0.2);
    \node[anchor=south,font=\scriptsize] at (0,0.15) {\tikzpictext};
  }
}
\def\states[#1] (#2,#3) (#4,#5){%
  \foreach \x in {#2,...,#4}
    \foreach \y in {#3,...,#5}
      \pic[#1] at (\x,\y) {state};
}

\def\truncstates[#1] (#2,#3) (#4,#5){%
  \begin{scope}
      \clip (#2+0.5,#3+0.5) rectangle (#4+0.5,#5+0.5);
      \foreach \x in {#2,...,#4}
      \foreach \y in {#3,...,#5}
      \pic[#1] at (\x,\y) {state};
  \end{scope}
}

\def\truncstateswo[#1,#2] (#3,#4) (#5,#6){%
  \begin{scope}
    \clip (#3+0.5,#4+0.5) rectangle (#5+0.5,#6+0.5);
    \foreach \x in {#3,...,#5}
      \foreach \y in {#4,...,#6}{
        \pic[#1] at (\x,\y) {state};
        \draw[#2] (\x,\y) circle (0.35);
      }
  \end{scope}
}

\tikzset{
  halfstate/.pic={
    \draw (0,0)--(-0.15,0);
    \draw[rounded corners=2pt] (-0.3,-0.5+\d-\dd) rectangle (-0.15,0.5-\d+\dd);
    \draw (-0.3,-0.5+\d)--(-0.5,-0.5+\d)--(-0.5,0.5-\d)--(-0.3,0.5-\d);
    \filldraw[black] (-0.5,-0.5+\d) circle (\rad);
    \filldraw[black] (-0.5,0.5-\d) circle (\rad);    
  }
}

\tikzset{every picture/.style={baseline={([yshift=-.5ex]current bounding box.center)}, scale=1.2,transform shape}}

\begin{document}

\title{A generalization of the injectivity condition for Projected Entangled Pair States}
\author{Andras Molnar}
\author{Yimin Ge}
\author{Norbert Schuch}
\author{J. Ignacio Cirac}
\affiliation{Max-Planck-Institut für Quantenoptik, Hans-Kopfermann-Str. 1, D-85748 Garching, Germany}

\makeatletter
\hypersetup{pdftitle={\@title},pdfauthor={Andras Molnar, Yimin Ge, Norbert Schuch, J. Ignacio Cirac}}
\makeatother

\begin{abstract}
  We introduce a family of tensor network states that we term semi-injective Projected Entangled-Pair States (PEPS). They extend the class of injective PEPS and include other states, like the ground states of the AKLT and the CZX models in square lattices. We construct parent Hamiltonians for which semi-injective PEPS are unique ground states. We also determine the necessary and sufficient conditions for two tensors to generate the same family of such states in two spatial dimensions. Using this result, we show that the third cohomology labeling of Symmetry Protected Topological phases extends to semi-injective PEPS.
\end{abstract}

\maketitle

\section{Introduction}

Tensor Network States (TNS) are expected to approximate ground states of local Hamiltonians well~\cite{Hastings2006a,Arad2013,Molnar2015}. Their local description in terms of simple tensors makes them suitable for both numerical and analytical investigations. First, this local structure enables calculations in large or infinite systems. In fact, the Density Matrix Renormalization Group (DMRG) algorithm~\cite{White1992}, which proved successful in simulating one-dimensional systems, can be re-expressed in terms of Matrix Product States (MPS)~\cite{Ostlund1995,Verstraete2004d}, the simplest TNS. This connection also motivated a provably efficient algorithm to find the ground state of a gapped one-dimensional local Hamiltonian~\cite{Arad2016}. Algorithms based on higher-dimensional generalizations of MPS, Projected Entangled Pair States
(PEPS~\cite{Verstraete2004,Jordan2007}), are now also giving the best known numerical results for certain Hamiltonians in two dimensions (see, for example, Ref.~\onlinecite{Zheng2016,Corboz2014,Corboz2011b}). Second, TNS are useful for creating and analyzing exactly solvable
models with translation symmetry. Indeed, paradigmatic wavefunctions
appearing in different areas of research have a very simple PEPS
description, where one can generate a whole family of many-body states by
contracting, according to a given geometry, the so-called auxiliary indices
of $N$ copies of a single tensor. In two-dimensional systems this
includes, for instance, the Cluster
State~\cite{Briegel2001,Verstraete2003} underlying measurement based
computation, the rotationally invariant spin-liquid
AKLT~\cite{Affleck1988,Affleck1987} and RVB~\cite{Schuch2012} states, and
other chiral~\cite{Yang2014,Poilblanc2015} and non-chiral
states~\cite{Schuch2010,Buerschaper2014,Sahinoglu2014,Bultinck2015} embodying topological
order. In particular, PEPS encompass all known non-chiral topological
orders~\cite{Levin2005}. All these states allow for the construction of
local parent Hamiltonians and, by making use of their local description,
the ground space structure and the behavior of the low-energy excitations
can be fully analyzed.

In the last years, the theory of TNS has been considerably developed. In
particular, in one dimension a general structural theory of MPS is known.
First, the ``fundamental theorem of
MPS''~\cite{Perez-Garcia2007,Cirac2017} states that any two tensors
generating the same family of wavefunctions can always be related by a
``gauge'' transformation acting on the auxiliary indices of each tensor
independently. This result allowed for the classification of
symmetry-protected topological (SPT) phases in one
dimension~\cite{Chen2010a,Schuch2011}. Indeed, let us consider an MPS
invariant under a symmetry group $G$. Then, the local gauge
transformations which relate the tensor generating the MPS and that
obtained by a symmetry action form themselves a \emph{projective}
representation of $G$~\cite{Sanz2009,Pollmann2010}. The equivalence
classes formed by those projective representations, which are labelled by
the second group cohomology ${H}^2(G,{U}(1))$, thus
characterize the different SPT phases under the action of $G$. Second, for
any MPS there exists a systematic way of constructing parent Hamiltonians
with a controlled ground space~\cite{Fannes1992,Nachtergaele1996,Perez-Garcia2007}. All this
renders MPS a general framework for the study of solvable models and the
classification of phases in one dimension.

In two dimensions, considerable progress has been made in understanding
the description of topologically ordered phases in the PEPS
framework~\cite{Schuch2010,Buerschaper2014,Sahinoglu2014,Bultinck2015,Gu2008,Buerschaper2008}.
However, a general structural theory of PEPS and, in particular, a
``fundamental theorem of PEPS'', is only known for a specific class,
termed ``injective''~\cite{Perez-Garcia2010,Fannes1992}. Those are PEPS
whose generating tensor, when considered as a map from the auxiliary to
the physical indices, can be inverted. While formally, any random PEPS
(after blocking few spins) is injective, many relevant examples such as
the square lattice AKLT model~\cite{Affleck1987}, the 
RVB~\cite{Anderson1973,Verstraete2006a} state, wavefunctions obtained
from classical statistical mechanics models~\cite{Verstraete2006a}, 
as well as all topologically ordered states, do not have this property. Their
lack of injectivity prevents us from
understanding their behavior under symmetries and, at the same time, makes
the canonical construction of parent Hamiltonians with unique ground
states for injective PEPS inapplicable.

In order to analyze SPT phases in two dimensions, one might thus
tentatively restrict to injective PEPS and apply the fundamental theorem, which yet again yields a projective local symmetry action on the auxiliary indices. However, in two dimensions the classification of projective representations in terms of group cohomology is not stable under blocking, and thus cannot be used to label phases under symmetries unless translational invariance is imposed~\cite{Chen2011a}. This is remedied by the CZX model
proposed by Chen \emph{et al.}~\cite{Chen2011a}. It is made up from a
\emph{non-injective} PEPS with a corresponding symmetry action, in a way
that it exhibits non-trivial symmetry invariants which are insensitive to
blocking.  Specifically, the symmetry action on the auxiliary indices at the boundary of any region is given by Matrix Product Unitary Operators (MPUOs). Those can be labeled by elements of the third cohomology group $H^3(G,U(1))$ of the symmetry group $G$, and therefore those elements are expected to label the different SPT phases\cite{Chen2011a,Williamson2014}.

In this paper, we introduce \emph{semi-injective} PEPS and develop the
structure theory for them. They significantly generalize injective PEPS
and include, among others, the square-lattice AKLT model, all wavefunctions based on classical models, as well as CZX-like states. Our central result is a ``fundamental theorem of semi-injective PEPS'' which states that any two tensors generating the same semi-injective PEPS are related by an invertible Matrix Product Operator (MPO) acting on their auxiliary indices. For semi-injective PEPS with an on-site symmetry, the corresponding MPOs form a representation of the group. We give a general and fully rigorous proof, based on the arguments of Chen \emph{et al.}~\cite{Chen2011a}, that these MPO representations are labeled by the third cohomology group ${H}^3(G,\mathbb
C^*)={H}^3(G,{U}(1))$, suggesting that this labels SPT
phases for all semi-injective PEPS, including those away from fixed point
wavefunctions~\cite{Gu2009} or with non-unitary symmetries. We further
show that symmetries of semi-injective PEPS must have a special two-layer
structure. This, by itself, enables the definition of an MPO acting on the
physical indices alone, and thus allows one to assign another label
${H}^3(G,\mathbb C^*)$ to the symmetry action, which we show to coincide
with that of the MPO acting on the boundary. Therefore, the SPT labeling
is  just as much a property of the physical symmetry itself as it is of
the boundary, and can in fact be directly inferred by analyzing the
structure of the physical symmetry action, without needing to invoke the
underlying PEPS. As a corollary, this implies that the ${H}^3$
label of the MPO action is well-defined and, in particular, the same on
the horizontal and vertical boundaries, also in the absence of rotational
symmetry. Furthermore, we also provide two different construction for
parent Hamiltonians with unique ground states.

The paper is structured as follows.  In \cref{sec:formalism} we introduce
the formalism, and define semi-injective PEPS.  In \cref{sec:results}, we
give an overview of the main results of the paper.  Readers who are
interested only in the results rather than the proofs might thus restrict
to \cref{sec:formalism,sec:results}.  In
\cref{sec:hamiltonian},  we discuss parent Hamiltonian for semi-injective
PEPS.  In \cref{sec:mps}, we summarize central results from the structure
theory of MPS needed for the remainder of the paper.  In
\cref{sec:canonicalform}, we develop the structure theory of
semi-injective PEPS, i.e., we give a local characterization of when two
semi-injective PEPS generate the same state.  In \cref{sec:spt}, we apply
this to characterize symmetric semi-injective PEPS, yielding the
characterization in terms of the third cohomology group.

\section{Formalism}\label{sec:formalism}

In this section we introduce MPS, PEPS and the graphical calculus commonly used in the field of tensor networks (TNs). We modify the standard notations in order to be able to represent TNs that are concatenations of several layers of two-dimensional TNS. Using this notation, we define semi-injective PEPS. We show that this class of PEPS contains all injective PEPS as well as some examples that are not known to have injective PEPS description.

\subsection{Notation}

Translationally invariant MPS are defined in terms of a single rank-three tensor $A\in \mathbb{C}^D\otimes\mathbb{C}^D\otimes\mathbb{C}^d$, $A = \sum_{\alpha, \beta, i} A_{\alpha \beta}^i \ket{\alpha}\bra{\beta} \otimes \ket{i} = \sum_i A^i \otimes\ket{i}$. The corresponding state on $n$ particles is 
\begin{equation}
\ket{V_n(A)} =\sum_{i}\tr\{A^{i_1} \dots A^{i_n}\}\ket{i_1 \dots i_n} \ .
\end{equation} 
The tensor and the corresponding state can be represented graphically as follows. Each tensor is represented with a square with lines attached to it.  The number of lines connected to the rectangle is the rank of the tensor and each line represents one index. For example, the single MPS tensor is represented as:
\begin{equation}
  \begin{tikzpicture}
    \pic[pic text = $A$] at (0,0) {mps};
  \end{tikzpicture} \ .
\end{equation}
Tensor contractions are depicted by joining the lines corresponding to the indices contracted. For example, the contraction of two MPS tensors is
\begin{equation}
  \sum_{\alpha\beta\gamma i j} A^i_{\alpha \beta} A^j_{\beta\gamma} \ket{\alpha}\bra{\gamma} \otimes\ket{ij} = \ 
  \begin{tikzpicture}[baseline=-1mm]
    \pic[pic text = $A$] at (0,0) {mps};
    \pic[pic text = $A$] at (1,0) {mps};
  \end{tikzpicture} \ .
\end{equation}
Same way, the MPS can be depicted as:
\begin{equation}
  \ket{V_n(A)} = \sum_{i}\tr\{A^{i_1} \dots A^{i_n}\}\ket{i_1 \dots i_n} = \ 
  \begin{tikzpicture}[baseline=-0.1cm]
    \node at (2,0) {$\dots$};
    \foreach \x in {0,1,3}{
      \pic[pic text=$A$] at (\x, 0) {mps};
    }
    \draw (-0.5,0)--(-0.5,-0.5)--(3.5,-0.5)--(3.5,0);
  \end{tikzpicture} \ .
\end{equation} 
We refer to the contracted legs as \emph{bonds} or \emph{virtual indices}, $D$ as the \emph{bond dimension} of the MPS tensor $A$,  the uncontracted leg as \emph{physical index}, and $d$ as the \emph{physical dimension} of $A$.

We will define PEPS now as generalizations of MPS. We will consider a square lattice, although other geometries can also be used. Take a rank-five tensor $B$:
\begin{equation}
  B = \ 
  \begin{tikzpicture}
    \pic at (0,0) {pepstensor};
  \end{tikzpicture} \ .
\end{equation}
Consider an $n\times m $ rectangular grid with periodic boundary conditions (that is, on a torus). PEPS is defined then by placing the tensor $B$ at every lattice point and contracting the neighboring  tensors:
\begin{equation}
  V_{n,m}(B) = \ 
  \begin{tikzpicture}
    \foreach \i in {0,1,2,3} \foreach \j in {0,1,2} \pic at (\i,\j) {pepstensor};
    \foreach \i/\j in {-1/1,4/1,1.5/-0.75,1.5/2.75} \node at (\i,\j) {$\dots$};
  \end{tikzpicture} \ .
\end{equation}
An equivalent description is the following. Place maximally entangled pair states on the edges of the grid. These particles are referred to as virtual particles. At every lattice point, act with an operator on the four virtual particles closest to the lattice point:
\begin{equation}
  V_{n,m}(B) = \ 
  \begin{tikzpicture}
    \begin{scope}
      \clip (0.42,0.42) rectangle (3.58,3.58);
      \foreach \x in {0,1,2,3} \foreach \y in {0,1,2,3}{
        \draw (\x+0.2,\y)--++(0.6,0);
        \filldraw (\x+0.2,\y) circle (0.05);
        \filldraw (\x+0.8,\y) circle (0.05);
        \draw (\x,\y+0.2)--++(0,0.6);
        \filldraw (\x,\y+0.2) circle (0.05);
        \filldraw (\x,\y+0.8) circle (0.05);
        \draw[red] (\x,\y) circle (0.4);
      }
    \end{scope}
    \foreach \i/\j in {0/2,4/2,2/0,2/4} \node at (\i,\j) {$\dots$};
  \end{tikzpicture} \ .
\end{equation}
Here each dot represents a virtual particle, the lines connecting them represent that they are in an entangled (here the maximally entangled) state. The red circles depict the operators acting on the four virtual particles. We call a PEPS \emph{injective} if these operators are injective maps.

In the  following, we use this notation when drawing tensor networks in two dimension. For example, a four-partite state will be depicted as
\begin{equation}
\begin{tikzpicture}
  \pic[green] at (0,0) {state};
\end{tikzpicture} \ .
\end{equation}
This four-partite state can equally be thought of as a non-translationally invariant MPS on four sites. Then each corner depicts an MPS tensor. Operators are depicted as circles or rectangles.  For example,
\begin{equation}
  \begin{tikzpicture}
    \truncstateswo[green,red] (0,0) (1,1);
  \end{tikzpicture}
\end{equation}
depicts a four-partite operator acting on the physical indices (black points) of four MPS tensors. In certain cases, we need more than two layers of tensors. In this case the upper layer is drawn bigger. For example a product of two operators acting on four MPS tensors are depicted as
\begin{equation}
\begin{tikzpicture}
\truncstateswo[green,red] (0,0) (1,1);
\draw[red, dashed] (1,1) circle (0.4);
\end{tikzpicture}.
\end{equation}
In this case the operator depicted as a solid circle acts first on the four MPS tensors, the dashed one acts second. 

We will often use a minimal rank decomposition of tensors. For convenience, we will refer to the operators in the decomposition as Schmidt vectors, even though the minimal rank decomposition is not necessarily a Schmidt decomposition, as we don't require orthogonality. For example, a minimal rank decomposition of a four-partite operator acting on four MPS tensors will be depicted as:
\begin{equation}
\begin{tikzpicture}[scale=1.5, transform shape]
  \def\d{0.2}
  \truncstates[green] (0,0) (1,1);
  \draw[red] (1.1,0.65) rectangle (1.3, 1.35);
  \draw[red] (0.7,0.65) rectangle (0.9, 1.35);
  \draw[red] (0.9,1) -- (1.1,1);
  \end{tikzpicture}.
\end{equation}
The Schmidt vectors of a four-partite state $\ket{\phi}$ can be related to its MPS description. We will therefore depict the minimal rank decomposition as 
\begin{equation}\label{eq:schmidt}
  \begin{tikzpicture}
    \pic[green] at (0,0) {state};
  \end{tikzpicture} = \ 
  \begin{tikzpicture}
    \pic[green] at (0,0) {halfstate};
    \pic[green,rotate=180] at (0,0) {halfstate};
  \end{tikzpicture} \ ,
\end{equation}
where the Schmidt vectors are denoted as:
\begin{equation}
\begin{tikzpicture}
\pic[green,rotate=90] at (0,0) {halfstate};
\end{tikzpicture}\ .
\end{equation}

\subsection{Semi-injective PEPS}\label{sec:defn}

In this section we define semi-injective PEPS. We show that this class contains all injective PEPS. Moreover, we provide examples that are not known to admit an injective PEPS description, yet they can be written as semi-injective PEPS.

\begin{definition}[Semi-injective PEPS] Let $\ket{\phi}$ be a
four-partite state with full rank reduced densities at every site, and let $O$ be
an invertible operator. The \emph{semi-injective} PEPS
$\ket{\Psi_{n\times m}(\phi, O)}$ 
is defined as
  \begin{equation}
  \ket{\Psi_{n\times m}(\phi, O)} = 
  \begin{tikzpicture}
  \truncstateswo[green,red] (0,0) (4,3)
  \end{tikzpicture}\ ,
  \end{equation}
  where the green rectangle is $\ket{\phi}$ and the red circle is $O$, and the state is defined on a torus with $n\times m$ copies of $\ket\phi$. We will often drop  $\ket{\phi}$ and $O$ and the indices $n,m$  from the notation.
\end{definition}

Note that the full rank assumption does not affect which states can be described as semi-injective PEPS, it is only needed to avoid unnecessary degrees of freedom in the operator $O$. 

These states can be written as a PEPS, but that PEPS is  in general not  injective. For example, a PEPS tensor generating the same state is 
\begin{equation}\label{eq:pepstensor}
\begin{tikzpicture}
\truncstateswo[green,red] (0,0) (1,1);
\end{tikzpicture}
\end{equation}
with an arbitrary (non necessarily translational invariant) MPS
decomposition of $\ket{\phi}$.\footnote{TN states of this form have been
used in other contexts, see e.g.\ Ref.~\onlinecite{Xie2014}.}

In the following we show that these states include injective PEPS as well as all the examples mentioned above.

\paragraph{Injective PEPS.} In this case the invertible operator is the PEPS tensor, and the four-partite state consists of two maximally entangled pairs (and a one-dimensional particle):
\begin{equation}
\begin{tikzpicture}
\clip (0.42,0.42) rectangle (3.58,3.58);
\foreach \x in {0,1,2,3} \foreach \y in {0,1,2,3}{
  \draw (\x+0.2,\y)--++(0.6,0);
  \filldraw (\x+0.2,\y) circle (0.05);
  \filldraw (\x+0.8,\y) circle (0.05);
  \draw (\x,\y+0.2)--++(0,0.6);
  \filldraw (\x,\y+0.2) circle (0.05);
  \filldraw (\x,\y+0.8) circle (0.05);
  \draw[red] (\x,\y) circle (0.4);
}
\end{tikzpicture} = 
\begin{tikzpicture}
  \truncstateswo[green,red] (0,0) (3,3);
\end{tikzpicture} 
\end{equation}
where the four-partite states are (the fourth particle is a scalar):
\begin{equation}
\begin{tikzpicture}
  \pic[green] at (0,0) {state};
\end{tikzpicture} = 
\begin{tikzpicture}[baseline=0.4cm]
\draw (0,0.2) -- (0,0.8);
\draw (0.2,0) -- (0.8,0); 
\foreach \u in {(0,0.2),(0,0.8),(0.2,0),(0.8,0)}
\filldraw \u circle (0.05);
\end{tikzpicture}
\end{equation}

\paragraph{The CZX model} The CZX model readily admits semi-injective PEPS form: the four-partite states are  GHZ states $\ket{0000}+\ket{1111}$, whereas the invertible operators are unitary operators $U_{CZX} = X^{\otimes 4} \cdot \prod_{\langle ij \rangle} CZ_{ij}$.

\paragraph{Purification of thermal states.}

Consider a commuting nearest neighbor Hamiltonian on a square lattice and the following purification of its Gibbs state:
\begin{equation}
  \ket{\Psi} = \frac{1}{\sqrt{Z}} \left(e^{-\beta H/2}\otimes\id \right)\bigotimes_i \ket{\phi^+}_i \ ,
\end{equation}
where $\ket{\phi^+} = \sum_j \ket{jj}$ is a maximally entangled pair state, and $e^{-\beta H}\otimes\id$ acts non-trivially on one of the entangled pairs at every lattice site. This state 
admits a PEPS description: as the Hamiltonian terms are commuting, $e^{-\beta H/2}$ is a product of local operators. The PEPS tensors are then simply the product of the Schmidt vectors of these local operators acting on $\ket{\phi^+}$. This tensor does not become injective after blocking exactly because of the corners: after blocking, the operators lying entirely inside the blocked region can be inverted. Note that  applying invertible operators on the tensor doesn't change injectivity. This factorizes the tensor into a product of tensors on the boundary and tensors on the corners. The boundary is injective, while the corners are not: the Schmidt vectors of the Hamiltonian terms commute, therefore any antisymmetric state on the corner is mapped to zero.

Nevertheless, these states admit a semi-injective PEPS representation. Indeed, block $2\times2$ pairs of particles. The four-partite state in the semi-injective PEPS description consists of the four pairs of particles in the state
\begin{equation}
  \begin{tikzpicture}[baseline=0.3cm]
  \draw[green] (0,0) rectangle (0.8,0.8);
  \foreach \u in {0,0.8} \foreach \v in {0,0.8}
  \filldraw (\u,\v) circle (0.05);
  \end{tikzpicture} = 
  \begin{tikzpicture}[baseline=-0.1cm,scale=0.7]
  \draw[orange] (0.5,0.5) ellipse (0.7cm and 0.1cm); 
  \draw[orange] (0.5,-0.5) ellipse (0.7cm and 0.1cm); 
  \draw[red] (0,0) ellipse (0.2cm and 0.65 cm); 
  \draw[red] (1,0) ellipse (0.2cm and 0.65 cm); 
  \foreach \x in {(0,-0.5),(0,0.5),(1,-0.5),(1,0.5)}
  \draw[fill=black] \x circle (0.05cm);
  \end{tikzpicture} = 
  \left(\prod_{\langle ij\rangle} e^{-\beta h_{ij}}\right) \otimes \id^{\otimes 4} \cdot \bigotimes_{i=1}^4 \ket{\phi^+}_i\ ,
\end{equation}
where both $i$ and $j$ runs over the particles in one block and $\ket{\phi^+}$ is the maximally entangled state. The dots represent $\ket{\phi^+}$, while the ellipses $e^{-\beta h_{ij}}\otimes \id^{\otimes 2}$. The invertible operator is a product of $e^{-\beta h_{ij}}\otimes \id^{\otimes 2}$ on the $2\times 2$ block shifted by half a lattice constant in both directions:
\begin{equation}
  \begin{tikzpicture}[baseline=-0.1cm]
  \draw[red] (0,0) circle (0.4cm);
  \end{tikzpicture} = 
  \begin{tikzpicture}[baseline=-0.1cm,scale=0.7]
  \draw[orange] (0.5,0.5) ellipse (0.7cm and 0.1cm); 
  \draw[orange] (0.5,-0.5) ellipse (0.7cm and 0.1cm); 
  \draw[red] (0,0) ellipse (0.2cm and 0.65 cm); 
  \draw[red] (1,0) ellipse (0.2cm and 0.65 cm); 
  \end{tikzpicture} \ .
\end{equation}
With this convention, the state is written as an injective PEPS as follows:
\begin{equation}
\begin{tikzpicture}[scale=0.5]
\clip (0.42,0.42) rectangle (6.58,6.58);
\foreach \x in {0,1,2,3,4,5,6} \foreach \y in {0,1,2,3,4,5,6}{
  \draw[fill=black] (\x,\y) circle (0.05cm);
  \draw[orange] (\x+0.5,\y) ellipse (0.7cm and 0.1cm); 
  \draw[red] (\x,\y+0.5) ellipse (0.2cm and 0.65 cm); 
}
\end{tikzpicture} = 
\begin{tikzpicture}
  \truncstateswo[green,red] (0,0) (3,3);
\end{tikzpicture} 
\end{equation}

\paragraph{AKLT in two dimensions.}

The two-dimensional AKLT model is a spin-2 system which is constructed as follows: place a singlet $\ket{01}-\ket{10}$ on each edge of a square lattice. Then, at each vertex, project the four virtual qubits into the 5-dimensional symmetric  subspace:
\begin{equation}
\ket{\Psi_{AKLT}}=
 \begin{tikzpicture}
    \foreach \i in {0,...,3}{
      \draw[very thick, draw=blue] (-0.4,\i)--(3.4,\i);
      \draw[very thick, draw=blue] (\i,-0.4)--(\i,3.4);
    }
    \foreach \x in {-0.2,...,2.8}{
      \foreach \y in {0,...,3}{
        \draw[white,fill=white] (\x,\y)--+(0.2,0.2)--+(0.4,0)--+(0.2,-0.2)--(\x,\y);
      }
    }
    \foreach \x in {0,...,3}{
      \foreach \y in {0,...,3}{
        \node[circle, minimum height = 0.5mm, draw=black, fill, thick, inner sep=0pt] at (\x+0.2,\y) {};
        \node[circle, minimum height = 0.5mm, draw=black, fill, thick, inner sep=0pt] at (\x-0.2,\y) {};
        \node[circle, minimum height = 0.5mm, draw=black, fill, thick, inner sep=0pt] at (\x,\y+0.2) {};
        \node[circle, minimum height = 0.5mm, draw=black, fill, thick, inner sep=0pt] at (\x,\y-0.2) {};
      }
    }
    \foreach \x in {0,...,3}{
      \foreach \y in {0,...,3}{
        \draw[thick, draw=orange] (\x,\y) circle (0.3);
      }
    }
    
  \end{tikzpicture}
\end{equation}
Here the blue lines represent singlets, the orange circles projectors into the symmetric subspace. As any virtual boundary state which is anti-symmetric on the two qubits of any corner is in the kernel of the map after appropriate applications of single-qubit $Y$s,   the PEPS tensors describing this state cannot be injective, even after blocking. 

The two-dimensional AKLT admits a semi-injective PEPS description as follows:

\begin{equation}
\begin{tikzpicture}
    \foreach \i in {0,...,3}{
      \draw[draw=blue] (-0.4,\i)--(3.4,\i);
      \draw[draw=blue] (\i,-0.4)--(\i,3.4);
    }
     Decoration
    \foreach \x in {-0.2,...,2.8}{
      \foreach \y in {0,...,3}{
        \draw[white,fill=white] (\x,\y)--+(0.2,0.2)--+(0.4,0)--+(0.2,-0.2)--(\x,\y);
      }
    }
    \foreach \x in {0,...,3}{
      \foreach \y in {0,...,3}{
        \node[circle, minimum height = 0.5mm, draw=black, fill, thick, inner sep=0pt] at (\x+0.2,\y) {};
        \node[circle, minimum height = 0.5mm, draw=black, fill, thick, inner sep=0pt] at (\x-0.2,\y) {};
        \node[circle, minimum height = 0.5mm, draw=black, fill, thick, inner sep=0pt] at (\x,\y+0.2) {};
        \node[circle, minimum height = 0.5mm, draw=black, fill, thick, inner sep=0pt] at (\x,\y-0.2) {};
      }
    }
    \foreach \x in {0,...,3}{
      \foreach \y in {0,...,3}{
        \draw[thick, draw=orange] (\x,\y) circle (0.3);
      }
    }
    
  \end{tikzpicture}
  = 
  \begin{tikzpicture}
  	\truncstateswo[green,red] (0,0) (3,3);
  \end{tikzpicture},
\end{equation}
with the four-partite state
\begin{equation}
\begin{tikzpicture}
  \pic[green] at (0,0) {state};
\end{tikzpicture}
= 
\begin{tikzpicture}[scale=.8]
 \draw[thick,blue] (0,.2) -- (0,.8);
 \draw[thick,blue] (1,.2) -- (1,.8);
 \draw[thick,blue] (.2,0) -- (.8,0);
 \draw[thick,blue] (.2,1) -- (.8,1);
 
 \foreach \x in { .2,.8}
 {
	\draw[fill=black] (\x,0) circle (0.05cm);
	\draw[fill=black] (\x,1) circle (0.05cm);
	\draw[fill=black] (1,\x) circle (0.05cm);
	\draw[fill=black] (0,\x) circle (0.05cm);
 }
 \draw[orange,rotate around={-45:(.1,.1)}] (.1,.1) ellipse (.3 and 0.15);
 \draw[orange,rotate around={45:(.9,.1)}] (.9,.1) ellipse (.3 and 0.15);
 \draw[orange,rotate around={45:(.1,.9)}] (.1,.9) ellipse (.3 and 0.15);
 \draw[orange,rotate around={-45:(.9,.9)}] (.9,.9) ellipse (.3 and 0.15);
\end{tikzpicture},
\end{equation}
where the blue lines are singlets, and the orange ellipses are the projectors into the 3-dimensional symmetric subspace (the four-partite state can thus effectively be viewed as a state on four qutrits: the one-dimensional AKLT state on four particles), and
\begin{equation}\label{eq:AKLTmap}
\begin{tikzpicture}
\draw[red] (0,0) circle (0.4cm);
\end{tikzpicture}
=
\begin{tikzpicture}[scale=.8]
 \draw[thick,blue] (0,.2) -- (0,.8);
 \draw[thick,blue] (1,.2) -- (1,.8);
 \draw[thick,blue] (.2,0) -- (.8,0);
 \draw[thick,blue] (.2,1) -- (.8,1);
 
 \foreach \x in { .2,.8}
 {
	\draw[fill=black] (\x,0) circle (0.05cm);
	\draw[fill=black] (\x,1) circle (0.05cm);
	\draw[fill=black] (1,\x) circle (0.05cm);
	\draw[fill=black] (0,\x) circle (0.05cm);
 }
 
  \foreach \x in { -.2,1.2}
 {
	\draw[black] (\x,0) circle (0.05cm);
	\draw[black] (\x,1) circle (0.05cm);
	\draw[black] (1,\x) circle (0.05cm);
	\draw[black] (0,\x) circle (0.05cm);
 }

 \foreach \x in {0,1}
 {
 	\foreach \y in {0,1}
 	{
 		\draw[thick,orange] (\x,\y) circle (.4);
 	}
 }
\end{tikzpicture}
\end{equation}
which acts on the qubits represented by the hollow dots, restricted to the symmetric subspace at each corner. It can be verified that the rank of \eqref{eq:AKLTmap} is 81. Clearly, the image of the adjoint of \eqref{eq:AKLTmap} is contained in the subspace which is symmetric in the pairs of qubits at each corner. The dimension of the latter is also 81. Thus, \eqref{eq:AKLTmap} is invertible.

\section{Summary and results}\label{sec:results}

In this section, we give a summary of the results obtained in this work.
The detailed derivations of all these results are given in the subsequent
sections.  The results extend the properties known for injective PEPS to
semi-injective ones. First, we show how to construct local Hamiltonians
for which they are the unique ground states. Next, we answer the question
under which local conditions two semi-injective PEPS generate the same
state.  We then use this result to characterize symmetries in
semi-injective PEPS. We also find that the third cohomology
classification of SPT phases naturally extends to these states, and thus
these states might be suitable to capture the physics of SPT phases.   In
the following, we give a detailed description of the results.

Consider two semi-injective PEPS generated by $(\phi_A,O_A)$ and $(\phi_B,O_B)$. Suppose that on an $n\times m$ torus, they generate states that are proportional to each other:
\begin{equation}
  \begin{tikzpicture}
    \truncstateswo[blue,purple] (0,0) (4,3)
  \end{tikzpicture} = \mu_{n,m} \ 
  \begin{tikzpicture}
    \truncstateswo[green,orange, dashed] (0,0) (4,3)
  \end{tikzpicture},
\end{equation}
where the purple circle and the blue rectangle depicts $O_A$ and $\ket{\phi_A}$, while the orange dashed circle and the green rectangle depicts $O_B$ and $\ket{\phi_B}$ and $\mu_{n,m}\in\mathbb{C}$. Inverting $O_B$, we obtain
\begin{equation}\label{eq:setup0}
  \begin{tikzpicture}
    \truncstateswo[blue,red] (0,0) (4,3)
  \end{tikzpicture} = \mu_{n,m} \ 
  \begin{tikzpicture}
    \truncstates[green] (0,0) (4,3)
  \end{tikzpicture},
\end{equation}
where the red circle denotes the invertible operator $O = O_B^{-1} O_A$.

In this setup, we prove the following:
\begin{theorem} \label{thm:1summary}
  If \cref{eq:setup0} holds for some specific $n_0,m_0\geq 3$, then  for all $n,m\in\mathbb{N}$,
  \begin{enumerate}
  	\item \cref{eq:setup0} holds with $\mu_{n,m} = \mu^{nm}$. \label{item:propconst}
    \item The action of $O$ corresponds to an MPO acting on the boundary as follows:
Take a minimal rank decomposition of the four-partite states w.r.t. the vertical cut. That is, write 
      \begin{equation}
      \begin{tikzpicture}
        \pic[blue] at (0,0) {state};
      \end{tikzpicture} = \ 
      \begin{tikzpicture}
        \pic[blue] at (0,0) {halfstate};
        \pic[blue,rotate=180] at (0,0) {halfstate};
      \end{tikzpicture} \quad \text{and} \quad 
      \begin{tikzpicture}
        \pic[green] at (0,0) {state};
      \end{tikzpicture} = \ 
      \begin{tikzpicture}
        \pic[green] at (0,0) {halfstate};
        \pic[green,rotate=180] at (0,0) {halfstate};
      \end{tikzpicture} \ .  
      \end{equation}
      Then for the Schmidt vectors defined as above, the following holds: There are two MPO tensors $X$ and $Y$ such that 
        \begin{equation}
        \begin{tikzpicture}
        \scope
        \clip (0.5,-0.5-\d) rectangle (3.5,0.5+\d);
        \foreach \x in {0,...,3}{
          \pic[blue,rotate=90] at (\x,0.5+\d) {halfstate};
          \pic[blue,rotate=-90] at (\x,-0.5-\d) {halfstate};
          \draw[red] (\x+0.5,0) circle (0.3);
        }
        \endscope
        \node[font=\tiny,anchor=east] at (0.6,0) {$\dots$};
        \node[font=\tiny,anchor=west] at (3.4,0) {$\dots$};
        \end{tikzpicture}
        =
        \mu^n\
          \begin{tikzpicture}[green,font=\scriptsize]
          \foreach \x in {0,...,2}{
            \pic[rotate=90] at (\x,0.5+\d) {halfstate};
            \pic[rotate=-90] at (\x,-0.5-\d) {halfstate};
            \pic[pic text = $X$, anchor=south west] at (\x,0.8) {boundarympo};
            \pic[pic text = $Y$, anchor=north west] at (\x,-0.8) {boundarympo};
          }
          \node[black,font=\tiny,anchor=east] at (-0.5,0.8) {$\dots$};
          \node[black,font=\tiny,anchor=west] at (2.5,0.8) {$\dots$};
          \node[black,font=\tiny,anchor=east] at (-0.5,-0.8) {$\dots$};
          \node[black,font=\tiny,anchor=west] at (2.5,-0.8) {$\dots$};
          \end{tikzpicture} \ ,
        \end{equation}
        where $\mu\in\mathbb{C}$ is the proportionality constant from
Point~\ref{item:propconst} above,  $V_n(Y) = \left(V_n(X)\right)^{-1}$ for every size $n$, and both $X$ and $Y$ become injective after blocking two tensors.

    \item 
    The operator $O$ is a four-particle non-translationally invariant MPO 
%
%
%
 	with the property that cutting the MPO into two halves yields a minimal rank decomposition of $O$. \label{item:operatorMPO}
    \item The operator $O$ is a product of two-body invertible operators:
      \begin{equation}
        O = \left(O_{14}\otimes O_{23}\right) \cdot \left(O_{12} \otimes O_{34}\right) = \left(\tilde{O}_{12} \otimes \tilde{O}_{34}\right) \cdot \left(\tilde{O}_{14}\otimes \tilde{O}_{23}\right)\ ,
      \end{equation}
      where the particles are numbered clockwise from the upper left corner and $O_{ij}$ acts on particles $i$ and $j$.  Pictorially,
      \begin{equation}\label{eq:op_twolayer0}
        \begin{tikzpicture}[baseline=-0.1cm]
        \draw[red] (0,0) circle (0.5);
        \end{tikzpicture} =
        \begin{tikzpicture}[baseline=-0.1cm,rotate=90]
        \draw[red] (-0.5,-0.5) rectangle (-0.2,0.5);
        \draw[red] (0.2,-0.5) rectangle (0.5,0.5);
        \draw[red] (-0.6,-0.6) rectangle (0.6,-0.1);
        \draw[red] (-0.6,0.1) rectangle (0.6,0.6);      
        \end{tikzpicture} =
        \begin{tikzpicture}[baseline=-0.1cm]
        \draw[red] (-0.5,-0.5) rectangle (-0.2,0.5);
        \draw[red] (0.2,-0.5) rectangle (0.5,0.5);
        \draw[red] (-0.6,-0.6) rectangle (0.6,-0.1);
        \draw[red] (-0.6,0.1) rectangle (0.6,0.6);      
        \end{tikzpicture} \ .
      \end{equation}
  \end{enumerate}
\end{theorem}

In \cref{sec:spt} we use these results to rederive the third cohomology classification of SPT phases within the framework of semi-injective PEPS. The setup in this case is as follows:

Let $G$ be a group, $O_g$ a faithful (not necessarily unitary) representation of $G$. Let $\ket{\phi}$ be a four-partite state with full rank one-particle reduced densities. Suppose $\forall g\in G$, $O_g$ is a symmetry of the semi-injective PEPS defined by $\ket{\phi}$ and $\id$:
\begin{equation}\label{eq:spt_setup0}
\begin{tikzpicture}
\truncstateswo[blue,red] (0,0) (4,3)
\end{tikzpicture} = \mu^{nm}(g) \ 
\begin{tikzpicture}
\truncstates[blue] (0,0) (4,3)
\end{tikzpicture}\ ,
\end{equation}
where the blue squares represent $\ket{\phi}$, the red operators $O_g$.

Note that this setup can readily be applied for unitary symmetries of semi-injective PEPS: let the semi-injective PEPS be defined by the four-partite state $\ket{\phi}$ and an invertible operator $A$. Let the unitary representation of the symmetry group $G$ be $U_g$. Then, by inverting $A$ in the symmetry condition, we arrive to \cref{eq:spt_setup0} with $O_g = A^{-1} U_g A$.

Within this setup, we prove that 
\begin{theorem}
  If \cref{eq:spt_setup0} holds for some $n,m\geq 3$, then
  \begin{enumerate}
  \item $g\mapsto \mu(g)$ is a one-dimensional representation of $G$.
  \item For every $g\in G$ there are two MPO tensors $X_g$ and $Y_g$ such that 
    \begin{equation} \label{eq:spt_boundary0}
    \begin{tikzpicture}
    \scope
    \clip (0.5,-0.5-\d) rectangle (3.5,0.5+\d);
    \foreach \x in {0,...,3}{
      \pic[blue,rotate=90] at (\x,0.5+\d) {halfstate};
      \pic[blue,rotate=-90] at (\x,-0.5-\d) {halfstate};
      \draw[red] (\x+0.5,0) circle (0.3);
    }
    \endscope
    \node[font=\tiny,anchor=east] at (0.6,0) {$\dots$};
    \node[font=\tiny,anchor=west] at (3.4,0) {$\dots$};
    \end{tikzpicture}
    =
    \mu^n(g)\
      \begin{tikzpicture}[blue,font=\scriptsize]
          \foreach \x in {0,...,2}{
            \pic[rotate=90] at (\x,0.5+\d) {halfstate};
            \pic[rotate=-90] at (\x,-0.5-\d) {halfstate};
            \pic[pic text = $X_g$, anchor=south west] at (\x,0.8) {boundarympo};
            \pic[pic text = $Y_g$, anchor=north west] at (\x,-0.8) {boundarympo};
          }
      \node[black,font=\tiny,anchor=east] at (-0.5,0.8) {$\dots$};
      \node[black,font=\tiny,anchor=west] at (2.5,0.8) {$\dots$};
      \node[black,font=\tiny,anchor=east] at (-0.5,-0.8) {$\dots$};
      \node[black,font=\tiny,anchor=west] at (2.5,-0.8) {$\dots$};
      \end{tikzpicture} \ ,
    \end{equation}
    and $V_n(Y_g) = \left(V_n(X_g)\right)^{-1}$ for all $n$. Moreover, $V_n(X_g)$ and $V_n(Y_g)$ form projective representations of $G$ with $V_n(X_g) V_n(X_h) = \lambda^n (g,h) V_n (X_g X_h)$ for a two-cocycle $\lambda$. In particular, $V_n(X_g) V_n(X_h)$ has only one block in its canonical form.
    \item There is a canonical way to assign an element from $H^3(G,\mathbb{C}^*)$ to the one-block MPO representation $g\mapsto X_g$.
    \item 
	$O_g$ has a four-particle non-translationally invariant MPO representation with tensors $O_g^{(1)},O_g^{(2)},O_g^{(3)},O_g^{(4)}$ in the sense of \cref{thm:1summary}, Point \ref{item:operatorMPO}, such that  the MPO $V_n(O_g^{(1)}O_g^{(2)}O_g^{(3)}O_g^{(4)})$ 
    forms a one-block projective MPO representation and its cohomology label coincides with that of the boundary. In particular, the MPO labels obtained from the vertical and horizontal boundaries are the same.
  \end{enumerate}
\end{theorem}

\section{Parent Hamiltonian}\label{sec:hamiltonian}

In this section, we prove that semi-injective PEPS are unique ground states of their parent Hamiltonian. Let us consider a semi-injective PEPS $\ket{\psi}$.  Corresponding to this state, we consider two parent Hamiltonian constructions. First, one can obtain the usual parent Hamiltonian by writing the state as a PEPS with the tensors in \cref{eq:pepstensor}. That is, consider a $2\times 2$ patch of the tensors. Let $S$ be the subspace generated by the tensors with arbitrary boundary conditions:
\begin{equation}\label{eq:Sdef}
  S = \left\{ 
  \begin{tikzpicture}
    \truncstateswo[green,red] (0,0) (2,2)
    \draw (0.5,0.5) rectangle (2.5,2.5);
    \draw (0,0) rectangle (3,3);
    \node at (2.75,1.5) {$\lambda$};
  \end{tikzpicture}\ 
  \vert 
  \lambda \in \mathbb{C}^{D_v^4 D_h^4}\right\}
\end{equation}
The Hamiltonian term $\tilde{h}_i$ centered around the plaquette state at position $i$ is just the projector onto $S^\perp$ and the Hamiltonian $\tilde H$ is the sum over all positions of these projectors:
\begin{align}
  \tilde{h}_i &= \mathrm{Proj}\left(S^\perp\right)_i \otimes \id\\
  \tilde{H} &= \sum_i \tilde{h}_i.
\end{align}

The second construction is  to invert the operators $O$ around a plaquette state at site $i$ and project $\phi_i$ to zero:
\begin{equation} 
h_i = \left(\prod_{\langle ji\rangle} O_j^{-1}\right)^\dagger P_i \left(\prod_{\langle ji\rangle} O_j^{-1}\right),
\end{equation}
where $j$ runs over all positions of operators that (partially) act on the plaquette state at position $i$ and the projector $P_i$ is the projector to the orthocomplement of $\mathbb{C}\ket{\phi}$: $P_i = (\id_i - \ket{\phi}_i\bra{\phi})\otimes \id$. Then the Hamiltonian $H$ is the sum of the different terms:
\begin{equation}
  H = \sum_i h_i.
\end{equation}
\begin{prop}\label{prop:parent_ham}
  The semi-injective PEPS $\ket{\psi}$ is the unique ground state of both $H$ and $\tilde{H}$ at all system sizes.
\end{prop}
\begin{proof}
  We first prove  that $H$ has a unique ground state. Then we prove that the kernel of $\tilde{H}$ is contained in that of $H$.
  
  To see that $\ker H$ is one-dimensional, consider the following similarity transform: 
  \begin{equation}
  \left(\prod_j O_j\right)^\dagger H \left(\prod_j O_j\right) = \sum_i P_i \otimes \id \otimes \bigotimes_{j} \left(O_j^{-1}\right)^\dagger O_j^{-1}, 
  \end{equation}
  where the product runs over all sites $j$ that are not neighbors of the projector $P_i$, and the identity acts on all virtual particles that are neighbors of the four-partite state $\ket{\phi}$. The kernel of each term in the sum is $\ket{\phi}_i\otimes \bigotimes_{j\notin i} \mathcal{H}_j$, where $j$ runs over all virtual particles that are not in the four-partite state $\ket{\phi}$. Clearly the intersection of these subspaces is $\bigotimes_i \ket{\phi}_i$, that is, the kernel of $H$ is one-dimensional.
  
  To see that $\ker \tilde{H}\leq \ker H$, notice that every state in $S_i$ (defined in \cref{eq:Sdef}) is in the kernel of $h_i$. Therefore,
  \begin{equation}
  \ker \tilde{h}_i \leq \ker h_i.
  \end{equation}
  Finally, as $\ker H = \cap_i \ker h_i$ and $\ker \tilde{H} = \cap_i \ker \tilde{h}_i$, the inclusion also holds for the kernel of the total Hamiltonians.
\end{proof}

\section{Background: Matrix Product States}\label{sec:mps}

In this Section we recall some basic properties of MPS. These definitions and theorems are mainly covered in Ref. \onlinecite{Cirac2017}. First, recall some basic properties of completely positive maps.

\begin{definition}
  A completely positive map $T: \rho \mapsto T(\rho) = \sum_i A_i \rho A_i^\dagger$ is 
  \begin{itemize}
  \item \emph{irreducible} if there is no non-trivial projector $P$ such that  $T(\rho ) = PT(\rho )P^{\dagger}$ for all $\rho = P\rho P^{\dagger}$. Otherwise $T$ is reducible.
  \item \emph{primitive} if $\exists n$ such that $T^n(\rho)>0$ for all $\rho\geq 0$.
  \end{itemize}
\end{definition}

Note that then the following statements hold:
\begin{prop}\label{prop:channel_props}
  Let $T: \rho \mapsto T(\rho) = \sum_i A_i \rho A_i^\dagger$ be a completely positive map with spectral radius $r$. Then $r$ is an eigenvalue with at least one positive semidefinite eigenvector. Moreover,
  \begin{itemize}
  \item $T$ is primitive if and only if $r$ has multiplicity one, the corresponding eigenvector is positive definite, and there are no other eigenvalues of magnitude $r$.
  \item if $T$ is irreducible but not primitive, then $r$ has multiplicity one, and all eigenvalues of magnitude $r$ are $r\cdot \exp[2\pi i n/K]$ for some $K$ and  $n=1,2\dots K$. We call $K$ the periodicity of $T$.
  \item $T$ is reducible if and only if $A^i P = P A_i P$ for some non-trivial projector $P$. 
  \end{itemize}
\end{prop}
For proofs, see e.g. Ref.~\onlinecite{Evans1978a,Wolf2012a}. Now we define Matrix Product States. 
\begin{definition} 
  An \emph{MPS tensor} is a tensor $A\in \mathbb{C}^D \otimes (\mathbb{C}^D)^* \otimes \mathbb{C}^d$, 
  \begin{equation}
  A= \sum_{i\alpha\beta} A^i_{\alpha\beta}  \ket{\alpha} \bra{\beta} \otimes \ket{i} = \sum_i A^i \otimes \ket{i}.
  \end{equation}
  For any $n\in\mathbb{N}$, the state $V_n(A)\in(\mathbb{C}^d)^{\otimes n}$ is then defined as
  \begin{equation}
    V_n(A) = \sum_{i_1 \dots i_n} \tr\left\{A^{i_1}\dots A^{i_n}\right\} \ket{i_1\dots i_n}.
  \end{equation}
  The \emph{transfer matrix} of $A$, $T_A$ is the completely positive map $T_A:\rho\mapsto \sum_i A_i \rho A_i^{\dagger}$. We say that $A$ is
  \begin{itemize}
  \item \emph{injective}, if  $\sum_i \tr\{A^i \rho\} \ket{i} = 0$ implies $\rho=0$.
  \item \emph{normal}, if $T_A$ is primitive.  
  \item \emph{periodic}, if $T_A$ is irreducible but not primitive. 
  \end{itemize}
  An MPS is called \emph{normal, injective or periodic}, if it can be generated by a normal, injective or periodic MPS tensor.
\end{definition}

We often depict an MPS tensor and the corresponding MPS as follows:
\begin{equation}
  A \equiv \ 
  \begin{tikzpicture}[baseline=-0.1cm]
    \pic[pic text=$A$] at (0,0) {mps};
  \end{tikzpicture} \quad \Rightarrow \quad 
  V_n(A) = \ 
  \begin{tikzpicture}[baseline=-0.1cm]
    \node at (2,0) {$\dots$};
    \foreach \x in {0,1,3}{
      \pic[pic text=$A$] at (\x, 0) {mps};
    }
    \draw (-0.5,0)--(-0.5,-0.5)--(3.5,-0.5)--(3.5,0);
  \end{tikzpicture} \ .
\end{equation}
The horizontal legs of the MPS tensor $A$ are often referred as the \emph{virtual indices}, while the vertical one as the \emph{physical index} of $A$. The dimension of the virtual indices, $D$, is called the \emph{bond dimension} of $A$.

Note that, unlike in Ref. \onlinecite{Cirac2017}, for convenience, we do not suppose that the spectral radius of a normal tensor is 1. Note also that an MPS tensor is injective if and only if it has a left inverse, $C$, such that $\sum_i A^i \otimes C^i = \id$ in the sense as depicted below:
\begin{equation}
  \begin{tikzpicture}[baseline=0.5cm]
    \pic[pic text=$A$] at (0,0) {mps};
    \pic[pic text =$C$] at (0,1) {mps=180};
  \end{tikzpicture}
  = 
  \begin{tikzpicture}[baseline=0.5cm]
    \draw (0,0) -- (0.5,0)--(0.5,1)--(0,1);
    \draw (1.5,0) -- (1,0)--(1,1)--(1.5,1);
  \end{tikzpicture} .
\end{equation}
Here, and in the following, we use the following graphical calculus\cite{Orus2013} of tensors. A tensor is depicted as a box or circle, with some lines attached to it. These lines represent the indices of the tensor. Tensor contraction is depicted by joining the lines. In the picture above, for example, we have contracted the physical indices of $A$ and $C$. The result is the identity tensor from the bottom indices to the top indices. We have omitted drawing a box for the identity.

A frequently used concept in MPS theory is the blocking of tensors. 
\begin{definition}[Blocking]
  The MPS tensor $B$ is a \emph{blocking} of $A$ if $B = \sum_{i_1 \dots i_k} A^{i_1} \dots A^{i_k} \otimes \ket{i_1 \dots i_k}$. Note that $V_n(B) = V_{kn}(A)$.
\end{definition}
We will often write the above contraction of tensors as a product. That is, for any two MPS tensors $C$ and $D$, $ CD :=\sum_{i j} C^{i} D^{j} \otimes \ket{ij}$. With this notation, $B = AA\dots A$. We will use this notation even if one of the tensors does not have a physical index.

Note that a normal tensor stays normal after blocking. Moreover, injective and normal MPS are the same up to blocking:
\begin{prop}\label{prop:normal2injective}
  Any injective tensor is proportional to a normal tensor. Conversely, for any normal tensor $\exists L_0\in \mathbb{N}$ such that it becomes injective after blocking any $L\geq L_0$ tensors. The minimal such $L_0$ is called the injectivity length.
\end{prop}
This statement was proven e.g. in Ref. \onlinecite{Sanz2010}. Note that $L_0$ might be bigger than the primitivity length of $T_A$, that is, the minimal $n$ for which $T_A^n(\rho)>0$ for all $\rho \geq 0$. There is, however, a universal bound depending only on the bond dimension $D$.

Note that being normal or injective are properties which are stable under taking tensor product of MPS tensors:
\begin{prop}\label{prop:injective_tensor_prod}
  The tensor product of two normal MPS tensors is normal. The tensor product of two injective MPS tensors is injective.
\end{prop}

\begin{proof}
  First, we prove that the tensor product of two normal tensors $A$ and $B$ is normal.   The transfer matrix of $A\otimes B$ is $T_A \otimes T_B$, where $T_A$ is the transfer matrix of $A$ and $T_B$ is the transfer matrix of $B$. Denote the spectrum of any operator $T$ by $\sigma (T)$. Then $\sigma(T_A\otimes T_B) = \sigma(T_A)\cdot\sigma(T_B)$. Therefore, $T_A\otimes T_B$ has a unique eigenvalue with magnitude (and value) equal the spectral radius. The corresponding eigenvector of $T_A\otimes T_B$ is $\rho_A \otimes \rho_B$ if $\rho_A$ and $\rho_B$ are the eigenvectors of $T_A$ and $T_B$ with maximum eigenvalue, respectively. $\rho_A\otimes\rho_B$ is positive and is full rank, so $T_A\otimes T_B$ is primitive.
  
  Second, the tensor product of two injective tensors is injective: if $A$ and $B$ are injective and $A^{-1}$ and $B^{-1}$ are their left inverses, then $A^{-1}\otimes B^{-1}$ is a left inverse of $A\otimes B$.
\end{proof}

\begin{prop}\label{prop:can_form_normal_mps}
  Given two normal tensors $A $ and $B$ with injectivity length at most $L$, the two MPS generated by them either become perpendicular in the thermodynamic limit, i.e.
  \begin{equation}
    \frac{\left|\Scalprod{V_n(A)}{V_n(B)}\right|}{\norm{V_n(A)}\cdot \norm{V_n(B)}}\to 0
  \end{equation}
  as $n\to \infty$, or the following three equivalent statements hold:
  \begin{itemize}
    \item $V_n(A)=\lambda^n V_n(B)$ for some $\lambda\in\mathbb{C}$ for all $n$
    \item $\exists n\geq 2L+1$ such that $V_n(A)=\lambda^n V_n(B)$ for some $\lambda\in\mathbb{C}$ 
    \item $A^i = \lambda X B^i X^{-1}$, for some $\lambda\in \mathbb{C}$ and this $X$ is unique up to a constant
  \end{itemize}
\end{prop}

We call the normal tensors $A$ and $B$ \emph{essentially different} if the MPSs generated by them are not proportional in the above sense. The proof of these statements can be found in Ref. \onlinecite{Cirac2017}.

\begin{corollary}\label{cor:normal_lin_indep}
  Given a set of pairwise essentially different normal tensors, $A_i$, $\exists N\in\mathbb{N}$ such that the MPS $V_n(A_i)$ are linearly independent for all $ n >N$.
\end{corollary}

\begin{prop}
  Any MPS $V_n(A)$ can be decomposed into a linear combination of normal and periodic MPSs:
  \begin{equation}
    V_n(A) = \sum_i \mu_i^n V_n(A_i),
  \end{equation}
  where each $A_i$ is either normal or periodic.
\end{prop}

The proof can be found in Ref. \onlinecite{Cirac2017}. We provide a simplified proof here. 
\begin{proof}
  We prove this by induction on the bond dimension $D$. If $D=1$, $A_i$ is proportional to a normal MPS. Suppose now that the statement is true for all $D<D_0$. Consider an MPS tensor $A$ with bond dimension $D_0$. If its transfer matrix $T_A$ is irreducible, then $A$ is either periodic or proportional to a normal MPS tensor. Otherwise, there exists a non-trivial projector $ P$ such that $A_i P = PA_i P$, see \Cref{prop:channel_props}. Then $V_n(A) = V_n (PAP) + V_n(QAQ)$ with $Q=1-P$. Finally, the bond dimension of $PAP$ (and $QAQ$) can be compressed to the rank of $P$ (corr. $Q$): write $P=YX$ for some  $X:\mathbb{C}^{D_0}\to \mathbb{C}^{D}$, $Y:\mathbb{C}^D \to \mathbb{C}^{D_0}$, $XY=\id_{D}$. Then $XAY$ generates the same MPS as $PAP$. The bond dimension of the resulting MPS  is smaller than $D_0$, thus by the induction hypothesis, they can be written as a linear combination of normal or periodic MPS.
\end{proof}

\begin{prop}\label{prop:periodic_mps_decomp}
  Let $A$ be a periodic MPS tensor with periodicity $K$. After blocking $K$ tensors, $V_n(A)$ decomposes into $K$ essentially different normal MPS:
  \begin{equation}
    V_{Kn}(A) = \sum_{i=1}^K V_{n}(B_i),
  \end{equation}
  where the $B_i$s are pairwise essentially different  normal MPS tensors on $K$ spins. Moreover, $V_n(A)=0$ if $n\notin K\mathbb{N}$.
\end{prop}
This statement has been proven as Lemma 5 in Ref. \onlinecite{Cadarso2013}. Proposition 9 from Ref. \onlinecite{Cirac2017} is a corollary of this:

\begin{corollary}\label{cor:mps_decomp}
  For any MPS tensor $A$ $\exists K$ such that after blocking $K$ tensors, $V_{Kn}(A)$ decomposes into the following linear combination of normal tensors:
  \begin{equation}
    V_{Kn}(A) = \sum_i \left(\sum_j \mu_{ij}^n\right) V_n(B_i),
  \end{equation}
  where the $B_i$s are pairwise essentially different normal tensors on $K$ sites. 
\end{corollary}

Finally, the following statement, together with \Cref{cor:normal_lin_indep}, provides the ``uniqueness'' of this decomposition:

\begin{prop}\label{prop:sum_power}
  If for $\mu_1,\dots \mu_r \in \mathbb{C}\backslash\{0\}$ and $\lambda_1,\dots \lambda_s\in\mathbb{C}\backslash\{0\}$
  \begin{equation}
    \sum_{i=1}^r \mu_i^n =\sum_{j=1}^{s} \lambda_j^n
  \end{equation}
  for all $n\in\mathbb{N}$, then $r=s$ and $\mu_i =\lambda_{p(i)}$ for some permutation $p$ and for all $i$ .
\end{prop}
This statement has been proven  as Lemma 9 in Ref. \onlinecite{Cuevas2015}.

We will also consider non-translationally invariant MPS. 

\begin{definition}
  Let $d_i$ and $D_i$ ($i=1\dots k$) be positive integers. Let $X_i = \sum_{j=1}^{d_i} X_i^j \otimes \ket{j}\in \mathbb{C}^{D_i} \otimes \left(\mathbb{C}^{D_{i+1}}\right)^* \otimes \mathbb{C}^{d_i}$ be tensors for $i=1\dots k$, where we identify $k+1$ with $1$. Then the \emph{non-translationally invariant MPS} defined by these tensors is 
  \begin{equation}
    V(X_1,\dots , X_k) = \sum_{i_1=1}^{d_1} \dots \sum_{i_k=1}^{d_k} \tr\left\{X_1^{i_1} \dots X_k^{i_k} \right\} \ket{i_1 \dots i_k} \ .
  \end{equation} 
  A non-translationally invariant MPS is called \emph{injective after blocking $l$ sites} if $\forall i=1\dots k$ the tensor $X_i X_{i+1} \dots X_{i+l}$ satisfies that if $\tr \left\{ \rho X_i^{j_1} X_{i+1}^{j_2} \dots X_{i+l-1}^{j_l} \right\} \ket{j_1 \dots j_l} = 0$, then $\rho=0$.
\end{definition}

\begin{prop}
  Let $X_1,\dots X_k$ define a non-translationally invariant MPS that is injective after blocking $l$ sites. Then the MPS is also injective after blocking any $m\geq l$ sites. 
\end{prop}

\begin{proof}
  We prove this by induction on $m$. For $m=l$, the statement is true by assumption. Suppose that the MPS is injective after blocking $m$ tensors. Let $\rho\in \mathbb{C}^{D_{i+m}}\otimes \mathbb{C}^{D_{i}}$ such that for $m+1$ consecutive sites 
  \begin{equation}
    \sum_{j_1\dots j_{m+1}}\tr \Big\{ \rho X_i^{j_1} X_{i+1}^{j_2} \dots X_{i+m}^{j_{m+1}} \Big\} \cdot \ket{j_1 \dots j_{m+1}} = 0
  \end{equation}
  for some $i$. Then, as the tensor $X_i \dots X_{i+m-1}$ is injective,
  \begin{equation}
    X_{i+m}^{j_{m+1}} \rho  = 0 \quad \forall j_{m+1}\in \{1,2, \dots, d_{i+m}\}.
  \end{equation}   
  Take any matrix $M\in \mathbb{C}^{D_i} \otimes \mathbb{C}^{D_{i+1}}$. Then  
  \begin{equation}
    0 = X_{i+1}^{j_{2}}  \dots X_{i+m}^{j_{m+1}} \rho M \in \mathbb{C}^{D_{i+1}} \otimes \mathbb{C}^{D_{i+1}} .
  \end{equation} 
 Then
  \begin{equation}
    \sum_{j_2\dots j_{m+1}}\tr \Big\{  X_{i+1}^{j_{2}}  \dots X_{i+m}^{j_{m+1}} \rho M \Big\} \cdot \ket{j_1 \dots j_{m+1}} = 0 \ .
  \end{equation}  
  The block of the $m$ consecutive tensors $X_{i+1} \dots X_{i+m}$ is injective, therefore $\rho M=0$. As $M$ was arbitrary, $\rho=0$, thus the MPS is injective after blocking $m+1$ sites.
\end{proof}
Finally, we introduce Matrix Product Operators (MPO).
\begin{definition}
  A \emph{Matrix Product Operator} is an operator written in MPS form:
  \begin{equation}
   V_n(X) = \sum_{i_1\dots i_n, j_1 \dots j_n} \tr \{ X^{i_1 j_1 } \dots X^{i_n j_n } \} \ket{i_1 \dots i_n} \bra{j_1 \dots j_n} \ .
  \end{equation}
\end{definition}
As MPOs are just special MPSs, all the definitions and structure theorems above apply. In particular, we will use the terminology \emph{normal, injective, periodic} for MPOs too.

\section{Canonical form}\label{sec:canonicalform}

In this section we investigate when two semi-injective PEPS defined by $(\phi_A,O_A)$ and $(\phi_B,O_B)$  describe the same state for some (sufficiently large) system size. We find that this question can be decided locally: the two states are proportional for a large system size if and only if they are proportional on a $3\times 3$ torus. Moreover, the boundary degree of freedoms are related by an invertible MPO whose inverse is also an MPO. Finally, we show that $O_B^{-1}O_A$ has to be a product of two-particle invertible operators. In \cref{app:canoniclaform_examples}, we also provide some examples that explain why the situation is more complicated than in the case of injective PEPS.

Consider two semi-injective PEPS generated by $(\phi_A,O_A)$ and $(\phi_B,O_B)$. Suppose that on an $n\times m$ torus, they generate states that are proportional to each other:
\begin{equation}
  \begin{tikzpicture}
    \truncstateswo[blue,purple] (0,0) (4,3)
  \end{tikzpicture} = \mu_{n,m} \ 
  \begin{tikzpicture}
    \truncstateswo[green,orange, dashed] (0,0) (4,3)
  \end{tikzpicture},
\end{equation}
where the purple circle and the blue rectangle depicts $O_A$ and $\ket{\phi_A}$, while the orange dashed circle and the green rectangle depicts $O_B$ and $\ket{\phi_B}$ and $\mu_{n,m}\in\mathbb{C}$. Inverting $O_B$, we obtain
\begin{equation}\label{eq:setup}
  \begin{tikzpicture}
    \truncstateswo[blue,red] (0,0) (4,3)
  \end{tikzpicture} = \mu_{n,m} \ 
  \begin{tikzpicture}
    \truncstates[green] (0,0) (4,3)
  \end{tikzpicture},
\end{equation}
where the red circle denotes the invertible operator $O = O_B^{-1} O_A$. This equation is the starting point of our investigation below. First we prove that it hold for all system sizes:

\begin{prop}\label{prop:size_indep}
  If \cref{eq:setup} holds for some $n_0\geq 3,m_0\geq 3$, then it also holds for any $n,m\in\mathbb{N}$ and the proportionality constant is $\mu_{n,m} = \mu^{nm}$.
\end{prop}

\begin{proof}
  Take a minimal rank decomposition of the four-partite states w.r.t. the vertical cut. That is, write 
  \begin{equation}\label{eq:schmidt2}
  \begin{tikzpicture}
    \pic[blue] at (0,0) {state};
  \end{tikzpicture} = \ 
  \begin{tikzpicture}
    \pic[blue] at (0,0) {halfstate};
    \pic[blue,rotate=180] at (0,0) {halfstate};
  \end{tikzpicture} \quad \text{and} \quad 
  \begin{tikzpicture}
    \pic[green] at (0,0) {state};
  \end{tikzpicture} = \ 
  \begin{tikzpicture}
    \pic[green] at (0,0) {halfstate};
    \pic[green,rotate=180] at (0,0) {halfstate};
  \end{tikzpicture} \ .  
  \end{equation}
  Using this decomposition, \cref{eq:setup} reads as

  \def\spinRad{0.4}
  \def\schmidtCloseEnd{0.2}
  \def\schmidtFarEnd{0.33}
  \begin{equation}
  \begin{tikzpicture}[baseline = 1.4cm]
  \clip (0,.1) rectangle (4,3-.1);
  \foreach \x in {0,1,...,4}{
    \foreach \y in {0,...,3}{
      \draw[thick, draw=blue] (\x+\spinRad,\y+\spinRad) -- (\x+\spinRad,\y-\spinRad) ;
      \draw[thick, draw=blue] (\x-\spinRad,\y+\spinRad) -- (\x-\spinRad,\y-\spinRad) ;  	
      
      \draw[rounded corners=2pt,blue] (\x+\spinRad-\schmidtFarEnd,\y+\spinRad+0.1) rectangle (\x+\spinRad-\schmidtCloseEnd,\y-\spinRad-0.05);	
      \draw[rounded corners=2pt,blue] (\x-\spinRad+\schmidtFarEnd,\y+\spinRad+0.1) rectangle (\x-\spinRad+\schmidtCloseEnd,\y-\spinRad-0.05);	
      
      \draw[thick, draw=blue] (\x-\spinRad,\y+\spinRad) -- +(\schmidtCloseEnd,0) ;  
      \draw[thick, draw=blue] (\x-\spinRad,\y-\spinRad) -- +(\schmidtCloseEnd,0) ;  
      \draw[thick, draw=blue] (\x+\spinRad,\y+\spinRad) -- +(-\schmidtCloseEnd,0) ;  
      \draw[thick, draw=blue] (\x+\spinRad,\y-\spinRad) -- +(-\schmidtCloseEnd,0) ;  
      
      \draw[thick, draw=blue] (\x+\spinRad-\schmidtFarEnd,\y) -- (\x-\spinRad+\schmidtFarEnd,\y) ;  
      
      \node[circle, minimum height = 1mm, draw=black, fill, thick, inner sep=0pt] at (\x+\spinRad,\y+\spinRad) {};
      \node[circle, minimum height = 1mm, draw=black, fill, thick, inner sep=0pt] at (\x-\spinRad,\y+\spinRad) {};
      \node[circle, minimum height = 1mm, draw=black, fill, thick, inner sep=0pt] at (\x+\spinRad,\y-\spinRad) {};
      \node[circle, minimum height = 1mm, draw=black, fill, thick, inner sep=0pt] at (\x-\spinRad,\y-\spinRad) {};
      
      \draw[thick, draw=red] (\x-.5,\y-.5) circle (0.25);
      
    }
  }
  \end{tikzpicture} 
  = \mu_{n,m} \
  \begin{tikzpicture}[baseline = 1.4cm]
  \clip (0,.1) rectangle (4,3-.1);
  \foreach \x in {0,1,...,4}{
    \foreach \y in {0,...,3}{
      \draw[thick, draw=green] (\x+\spinRad,\y+\spinRad) -- (\x+\spinRad,\y-\spinRad) ;
      \draw[thick, draw=green] (\x-\spinRad,\y+\spinRad) -- (\x-\spinRad,\y-\spinRad) ;  	
      
      \draw[rounded corners=2pt,green] (\x+\spinRad-\schmidtFarEnd,\y+\spinRad+0.1) rectangle (\x+\spinRad-\schmidtCloseEnd,\y-\spinRad-0.05);	
      \draw[rounded corners=2pt,green] (\x-\spinRad+\schmidtFarEnd,\y+\spinRad+0.1) rectangle (\x-\spinRad+\schmidtCloseEnd,\y-\spinRad-0.05);	
      
      \draw[thick, draw=green] (\x-\spinRad,\y+\spinRad) -- +(\schmidtCloseEnd,0) ;  
      \draw[thick, draw=green] (\x-\spinRad,\y-\spinRad) -- +(\schmidtCloseEnd,0) ;  
      \draw[thick, draw=green] (\x+\spinRad,\y+\spinRad) -- +(-\schmidtCloseEnd,0) ;  
      \draw[thick, draw=green] (\x+\spinRad,\y-\spinRad) -- +(-\schmidtCloseEnd,0) ;  
      
      \draw[thick, draw=green] (\x+\spinRad-\schmidtFarEnd,\y) -- (\x-\spinRad+\schmidtFarEnd,\y) ;  
      
      \node[circle, minimum height = 1mm, draw=black, fill, thick, inner sep=0pt] at (\x+\spinRad,\y+\spinRad) {};
      \node[circle, minimum height = 1mm, draw=black, fill, thick, inner sep=0pt] at (\x-\spinRad,\y+\spinRad) {};
      \node[circle, minimum height = 1mm, draw=black, fill, thick, inner sep=0pt] at (\x+\spinRad,\y-\spinRad) {};
      \node[circle, minimum height = 1mm, draw=black, fill, thick, inner sep=0pt] at (\x-\spinRad,\y-\spinRad) {};
      
      
    }
  }
  \end{tikzpicture}.
  \end{equation}  
   This gives rise to an MPS description of the states with the following tensors:
  \begin{equation}\label{eq:mps_decomp}
  \begin{tikzpicture}[baseline=-.08cm]
  \node[rectangle, draw, ultra thick, minimum width=.8cm,  minimum height=.3cm, fill=blue!20] (A) at (0,0)  {};
  
  \draw[ultra thick, blue] (A.east) -- +(.3,0);
  \draw[ultra thick, blue] (A.west) -- +(-.3,0);
  
  \draw[ultra thick] (A.north) -- +(0,.2);
  \end{tikzpicture}
  = \ 
  \begin{tikzpicture}[baseline = 1.4cm]
  \clip (0,.1) rectangle (1,3-.1);
  \foreach \x in {0,1}{
    \foreach \y in {0,...,3}{

      \draw[thick, draw=blue] (\x+\spinRad,\y+\spinRad) -- (\x+\spinRad,\y-\spinRad) ;
      \draw[thick, draw=blue] (\x-\spinRad,\y+\spinRad) -- (\x-\spinRad,\y-\spinRad) ;  	
      
      \draw[rounded corners=2pt,blue] (\x+\spinRad-\schmidtFarEnd,\y+\spinRad+0.1) rectangle (\x+\spinRad-\schmidtCloseEnd,\y-\spinRad-0.05);	
      \draw[rounded corners=2pt,blue] (\x-\spinRad+\schmidtFarEnd,\y+\spinRad+0.1) rectangle (\x-\spinRad+\schmidtCloseEnd,\y-\spinRad-0.05);	
      
      \draw[thick, draw=blue] (\x-\spinRad,\y+\spinRad) -- +(\schmidtCloseEnd,0) ;  
      \draw[thick, draw=blue] (\x-\spinRad,\y-\spinRad) -- +(\schmidtCloseEnd,0) ;  
      \draw[thick, draw=blue] (\x+\spinRad,\y+\spinRad) -- +(-\schmidtCloseEnd,0) ;  
      \draw[thick, draw=blue] (\x+\spinRad,\y-\spinRad) -- +(-\schmidtCloseEnd,0) ;  
      
      \draw[thick, draw=blue] (\x+\spinRad-\schmidtFarEnd,\y) -- (\x-\spinRad+\schmidtFarEnd,\y) ;  
      
      \node[circle, minimum height = 1mm, draw=black, fill, thick, inner sep=0pt] at (\x+\spinRad,\y+\spinRad) {};
      \node[circle, minimum height = 1mm, draw=black, fill, thick, inner sep=0pt] at (\x-\spinRad,\y+\spinRad) {};
      \node[circle, minimum height = 1mm, draw=black, fill, thick, inner sep=0pt] at (\x+\spinRad,\y-\spinRad) {};
      \node[circle, minimum height = 1mm, draw=black, fill, thick, inner sep=0pt] at (\x-\spinRad,\y-\spinRad) {};
      
      \draw[thick, draw=red] (\x-.5,\y-.5) circle (0.25);
      
    }
  }
  \end{tikzpicture} \quad\text{and}\quad 
  \begin{tikzpicture}[baseline=-.08cm]
  \node[rectangle, draw, ultra thick, minimum width=.8cm,  minimum height=.3cm, fill=green!20] (A) at (0,0)  {};
  
  \draw[ultra thick, green] (A.east) -- +(.3,0);
  \draw[ultra thick, green] (A.west) -- +(-.3,0);
  
  \draw[ultra thick] (A.north) -- +(0,.2);
  \end{tikzpicture}
  = \ 
  \begin{tikzpicture}[baseline = 1.4cm]
  \clip (0,.1) rectangle (1,3-.1);
  \foreach \x in {0,1}{
    \foreach \y in {0,...,3}{

      \draw[thick, draw=green] (\x+\spinRad,\y+\spinRad) -- (\x+\spinRad,\y-\spinRad) ;
      \draw[thick, draw=green] (\x-\spinRad,\y+\spinRad) -- (\x-\spinRad,\y-\spinRad) ;  	
      
      \draw[rounded corners=2pt,green] (\x+\spinRad-\schmidtFarEnd,\y+\spinRad+0.1) rectangle (\x+\spinRad-\schmidtCloseEnd,\y-\spinRad-0.05);	
      \draw[rounded corners=2pt,green] (\x-\spinRad+\schmidtFarEnd,\y+\spinRad+0.1) rectangle (\x-\spinRad+\schmidtCloseEnd,\y-\spinRad-0.05);	
      
      \draw[thick, draw=green] (\x-\spinRad,\y+\spinRad) -- +(\schmidtCloseEnd,0) ;  
      \draw[thick, draw=green] (\x-\spinRad,\y-\spinRad) -- +(\schmidtCloseEnd,0) ;  
      \draw[thick, draw=green] (\x+\spinRad,\y+\spinRad) -- +(-\schmidtCloseEnd,0) ;  
      \draw[thick, draw=green] (\x+\spinRad,\y-\spinRad) -- +(-\schmidtCloseEnd,0) ;  
      
      \draw[thick, draw=green] (\x+\spinRad-\schmidtFarEnd,\y) -- (\x-\spinRad+\schmidtFarEnd,\y) ;  
      
      \node[circle, minimum height = 1mm, draw=black, fill, thick, inner sep=0pt] at (\x+\spinRad,\y+\spinRad) {};
      \node[circle, minimum height = 1mm, draw=black, fill, thick, inner sep=0pt] at (\x-\spinRad,\y+\spinRad) {};
      \node[circle, minimum height = 1mm, draw=black, fill, thick, inner sep=0pt] at (\x+\spinRad,\y-\spinRad) {};
      \node[circle, minimum height = 1mm, draw=black, fill, thick, inner sep=0pt] at (\x-\spinRad,\y-\spinRad) {};
      
      
    }
  }
  \end{tikzpicture},
  \end{equation}
  where the physical index of the MPS tensor is all physical indices of the virtual particles, while the virtual indices of the MPS corrspond to the virtual indices of the minimal rank decomposition of the four-partite states. These tensors are injective: the green tensor is just a tensor product of the Schmidt vectors, and as the Schmidt vectors (and their tensor product) are linearly independent, that tensor is injective. The blue tensor is obtained by acting with an invertible operator on the tensor product of Schmidt vectors, therefore it is also injective.

  Thus, using \cref{prop:can_form_normal_mps}, if \cref{eq:setup} holds for $n_0\geq 3$, $m_0\geq 3$, then it also holds when the system size in the horizontal direction is changed to any $n$ by keeping the system size in the vertical direction $m_0$. Therefore \cref{eq:setup} holds for $m_0$ and any $n$, and the proportionality constant is $\mu_{n,m_0} = \mu_{m_0}^n$ for some $\mu_{m_0} \in \mathbb{C}$. The  argumentation above holds w.r.t.\ the horizontal cut. Therefore the system size can be changed along the vertical direction too: as \cref{eq:setup} holds for $n,m_0$, it also holds for $n,m$ and the proportionality constant is then $\mu_{n,m_0}^{m/m_0} =\mu^{nm}$ for some $\mu\in\mathbb{C}$. 
\end{proof}

Note that this implies that it is decidable whether two semi-injective PEPS are equal for all system size. Moreover, it is also practically checkable: it is enough to calculate the overlap  between two states (and their norms) on a $3\times 3$ torus. The overlaps can be calculated by standard tensor network techniques. The cost of this computation scales as the 12th power of the Schmidt rank.  

Using \cref{prop:can_form_normal_mps}, we conclude that up to a constant there is a uniquely defined operator $X_n$ on the boundary for which
\begin{equation}\label{eq:boundary0}
  \begin{tikzpicture}
  \scope
  \clip (0.5,-0.5-\d) rectangle (3.5,0.5+\d);
  \foreach \x in {0,...,3}{
    \pic[blue,rotate=90] at (\x,0.5+\d) {halfstate};
    \pic[blue,rotate=-90] at (\x,-0.5-\d) {halfstate};
    \draw[red] (\x+0.5,0) circle (0.3);
  }
  \endscope
  \node[font=\tiny,anchor=east] at (0.6,0) {$\dots$};
  \node[font=\tiny,anchor=west] at (3.4,0) {$\dots$};
  \end{tikzpicture}
  = \mu^n \ 
  \begin{tikzpicture}
    \scope
      \clip (0.5,-1.1-\d) rectangle (3.5,1.1+\d);
      \draw (0,0.8)--(4,0.8);
      \draw (0,-0.8)--(4,-0.8);
      \foreach \x in {0,...,3}{
        \pic[green,rotate=90] at (\x,0.5+\d) {halfstate};
        \pic[green,rotate=-90] at (\x,-0.5-\d) {halfstate};
        \draw[green] (\x,0.5)--+(0,0.6);
        \draw[green] (\x,-0.5)--+(0,-0.6);
        \draw[fill=gray] (\x,0.8) circle (0.1);
        \draw[fill=white] (\x,-0.8) circle (0.1);
      }
    \endscope 
    \draw[fill=gray,rounded corners=3pt] (1-0.15,0.7) rectangle (3.15,0.9);
    \draw[fill=white,rounded corners=3pt] (1-0.15,-0.7) rectangle (3.15,-0.9);
    \node[anchor=west,font=\scriptsize] at (3.5,0.8) {$X_{n}$};
    \node[anchor=west,font=\scriptsize] at (3.5,-0.8) {$X^{-1}_{n}$};
    \node[font=\tiny,anchor=east] at (0.6,0) {$\dots$};
    \node[font=\tiny,anchor=west] at (3.4,0) {$\dots$};
  \end{tikzpicture} \ .
\end{equation}
This construction, however, does not reveal anything about the properties of the gauges $X_n$ and $X_n^{-1}$: they are globally defined and the definition depends on the system size. In the following we explore their structure and show that they can both be written as a normal MPOs.
\begin{theorem}\label{prop:boundary_mpo}
  Suppose \cref{eq:setup} holds for some $n,m\geq 3$. Then there are two MPO tensors $X$ and $Y$ such that 
  \begin{equation} \label{eq:boundary1}
  \begin{tikzpicture}
  \scope
  \clip (0.5,-0.5-\d) rectangle (3.5,0.5+\d);
  \foreach \x in {0,...,3}{
    \pic[blue,rotate=90] at (\x,0.5+\d) {halfstate};
    \pic[blue,rotate=-90] at (\x,-0.5-\d) {halfstate};
    \draw[red] (\x+0.5,0) circle (0.3);
  }
  \endscope
  \node[font=\tiny,anchor=east] at (0.6,0) {$\dots$};
  \node[font=\tiny,anchor=west] at (3.4,0) {$\dots$};
  \end{tikzpicture}
  =
  \mu^n\
    \begin{tikzpicture}[green,font=\scriptsize]
    \foreach \x in {0,...,2}{
      \pic[rotate=90] at (\x,0.5+\d) {halfstate};
      \pic[rotate=-90] at (\x,-0.5-\d) {halfstate};
      \pic[pic text = $X$, anchor=south west] at (\x,0.8) {boundarympo};
      \pic[pic text = $Y$, anchor=north west] at (\x,-0.8) {boundarympo};
    }
    \node[black,font=\tiny,anchor=east] at (-0.5,0.8) {$\dots$};
    \node[black,font=\tiny,anchor=west] at (2.5,0.8) {$\dots$};
    \node[black,font=\tiny,anchor=east] at (-0.5,-0.8) {$\dots$};
    \node[black,font=\tiny,anchor=west] at (2.5,-0.8) {$\dots$};
    \end{tikzpicture} \ ,
  \end{equation}
  where $\mu\in\mathbb{C}$ is the proportionality constant from \cref{prop:size_indep}, and  $V_n(Y) = \left(V_n(X)\right)^{-1}$ for every size $n$ and both $X$ and $Y$ become injective after blocking two tensors.
\end{theorem} 

Before proceeding to the proof, notice that
\begin{lemma} \label{lem:boundaryMPSinjective}
  The l.h.s. of \cref{eq:boundary1} can be described by an MPS that becomes injective after blocking two tensors.
\end{lemma}

\begin{proof}
  Take a minimal rank decomposition of the operators $O$:
  \begin{equation}
    \begin{tikzpicture}[baseline=-0.1cm]
      \draw[thick, red] (0,0) circle (0.5);
    \end{tikzpicture}
    =
    \begin{tikzpicture}[baseline=0.4cm]
      \draw[thick, red] (0,0) rectangle (0.4,1);
      \draw[thick, red] (0.6,0) rectangle (1,1);
      \draw[thick, red] (0.4,0.5)--(0.6,0.5);
    \end{tikzpicture}.
  \end{equation}
  Then the l.h.s. of \cref{eq:boundary1} is an MPS with  MPS tensor 
  \begin{equation}\label{eq:boundary_mps_decomp}
    \begin{tikzpicture}
      \pic at (0,0) {mps};
    \end{tikzpicture}
    =
    \begin{tikzpicture}
      \pic[blue,rotate=90] at (0.5,0.5+\d) {halfstate};
      \pic[blue,rotate=-90] at (0.5,-0.5-\d) {halfstate};
      \draw[red] (\d-\dd,-\d-\dd) rectangle (\d+\dd,+\d+\dd); 
      \draw[red] (1-\d-\dd,-\d-\dd) rectangle (1-\d+\dd,+\d+\dd); 
      \draw[red] (\d-\dd,0) --+(-0.15,0);
      \draw[red] (1-\d+\dd,0) --+(0.15,0);
    \end{tikzpicture} \ ,
  \end{equation}
  where the physical indices of the MPS are both the physical indices and the two virtual indices belonging to the decomposition of $\ket{\phi_A}$ on the r.h.s. of \cref{eq:boundary_mps_decomp}, while the virtual indices of the MPS are the virtual indices belonging to the decomposition of $O$ on the r.h.s. of \cref{eq:boundary_mps_decomp}.
  
  We prove now that this MPS tensor is injective after blocking two tensors. To see this, block two tensors and note that contracting the middle indices gives back $O$:
  \begin{equation}
    \begin{tikzpicture}
      \pic at (0,0) {mps};
      \pic at (1,0) {mps};
    \end{tikzpicture}
    =
    \begin{tikzpicture}
      \foreach \x in {0,1}{
        \pic[blue,rotate=90] at (\x+0.5,0.5+\d) {halfstate};
        \pic[blue,rotate=-90] at (\x+0.5,-0.5-\d) {halfstate};
        \draw[red] (\x+\d-\dd,-\d-\dd) rectangle (\x+\d+\dd,+\d+\dd); 
        \draw[red] (\x+1-\d-\dd,-\d-\dd) rectangle (\x+1-\d+\dd,+\d+\dd); 
        \draw[red] (\x+\d-\dd,0) --+(-0.15,0);
        \draw[red] (\x+1-\d+\dd,0) --+(0.15,0);
      }
    \end{tikzpicture}
    =
    \begin{tikzpicture}
      \foreach \x in {0,1}{
        \pic[blue,rotate=90] at (\x+0.5,0.5+\d) {halfstate};
        \pic[blue,rotate=-90] at (\x+0.5,-0.5-\d) {halfstate};
      }
      \draw[red] (\d-\dd,-\d-\dd) rectangle (\d+\dd,+\d+\dd); 
      \draw[red] (2-\d-\dd,-\d-\dd) rectangle (2-\d+\dd,+\d+\dd); 
        \draw[red] (\d-\dd,0) --+(-0.15,0);
        \draw[red] (2-\d+\dd,0) --+(0.15,0);
      \draw[red] (1,0) circle (0.3);
    \end{tikzpicture} \ .
  \end{equation}
  Inverting $O$ does not change the injectivity of the MPS tensor, as it is an invertible operation on its physical indices. Therefore it is enough to prove that 
  \begin{equation}
    \begin{tikzpicture}
      \draw (0,0)--+(0,0.8);
      \draw (1,0)--+(0,0.8);
      \pic at (0,0) {mps};
      \pic at (1,0) {mps};
      \draw[rounded corners=3pt,fill=white] (-0.15,0.4) rectangle (1.15,0.6); 
      \node at (-.5,.5) {\scriptsize $O^{-1}$};
    \end{tikzpicture}
    =
    \begin{tikzpicture}
      \foreach \x in {0,1}{
        \pic[blue,rotate=90] at (\x+0.5,0.5+\d) {halfstate};
        \pic[blue,rotate=-90] at (\x+0.5,-0.5-\d) {halfstate};
      }
      \draw[red] (\d-\dd,-\d-\dd) rectangle (\d+\dd,+\d+\dd); 
      \draw[red] (2-\d-\dd,-\d-\dd) rectangle (2-\d+\dd,+\d+\dd); 
        \draw[red] (\d-\dd,0) --+(-0.15,0);
        \draw[red] (2-\d+\dd,0) --+(0.15,0);
    \end{tikzpicture}
    = v_i \otimes w_j
  \end{equation}
  is injective. Both $v_i$ and $w_j$ are linearly independent, as the Schmidt vectors of $O$ are linearly independent and the one body reduced densities of the four-partite states are full rank. Therefore the vectors $v_i\otimes w_j$ are also linearly independent, that is, the corresponding tensor is injective.  
\end{proof}

We now proceed to the proof of \cref{prop:boundary_mpo}.
\begin{proof}[Proof of \cref{prop:boundary_mpo}]
  We first prove  that $X_n$ and $X_n^{-1}$ are proportional to an MPS. Write the l.h.s. of \cref{eq:boundary0} as an MPS with two physical indices:
  \begin{equation}
    \begin{tikzpicture}
      \draw (0,0)--(1,0);
      \draw (0.35,0)--+(0,0.4);
      \draw (0.65,0)--+(0,0.4);
      \draw[fill=white] (0.2,-0.2) rectangle (0.8,0.2);
    \end{tikzpicture}
    =
    \begin{tikzpicture}
      \pic[blue,rotate=90] at (0.5,0.5+\d) {halfstate};
      \pic[blue,rotate=-90] at (0.5,-0.5-\d) {halfstate};
      \draw[red] (\d-\dd,-\d-\dd) rectangle (\d+\dd,+\d+\dd); 
      \draw[red] (1-\d-\dd,-\d-\dd) rectangle (1-\d+\dd,+\d+\dd); 
      \draw[red] (\d-\dd,0) --+(-0.15,0);
      \draw[red] (1-\d+\dd,0) --+(0.15,0);
    \end{tikzpicture} \ ,
  \end{equation}
  where the left physical index of the MPS tensor corresponds to the indices on the top of the r.h.s. (physical and virtual indices of the Schmidt vector), while the right one to the indices on the bottom of the r.h.s., and the virtual indices of the MPS correspond to the Schmidt index of the decomposition of $O$. With this notation, \cref{eq:boundary0} reads as 
  \begin{equation}\label{eq:mps_double_boundary}
    \begin{tikzpicture}
      \foreach \i in {0,1,2}{
        \draw (\i,0)--(\i+1,0);
        \draw (\i+0.35,0)--+(0,0.4);
        \draw (\i+0.65,0)--+(0,0.4);
        \draw[fill=white] (\i+0.2,-0.2) rectangle (\i+0.8,0.2);
      }
      \node[font=\tiny,anchor=east] at (0,0) {$\dots$};
      \node[font=\tiny,anchor=west] at (3.0,0) {$\dots$};
    \end{tikzpicture}\ =\
    \begin{tikzpicture}[baseline = -0.1cm]
      \scope
        \clip (0.5,-1.1-\d) rectangle (3.5,1.1+\d);
        \draw (0,0.8)--(4,0.8);
        \draw (0,-0.8)--(4,-0.8);
        \foreach \x in {0,...,3}{
          \pic[green,rotate=90] at (\x,0.5+\d) {halfstate};
          \pic[green,rotate=-90] at (\x,-0.5-\d) {halfstate};
          \draw[green] (\x,0.5)--+(0,0.6);
          \draw[green] (\x,-0.5)--+(0,-0.6);
          \draw[fill=gray] (\x,0.8) circle (0.1);
          \draw[fill=white] (\x,-0.8) circle (0.1);
        }
      \endscope 
      \draw[fill=gray,rounded corners=3pt] (1-0.15,0.7) rectangle (3.15,0.9);
      \draw[fill=white,rounded corners=3pt] (1-0.15,-0.7) rectangle (3.15,-0.9);
      \node[anchor=west,font=\scriptsize] at (3.8,0.8) {$X_{n}$};
      \node[anchor=west,font=\scriptsize] at (3.8,-0.8) {$X^{-1}_{n}$};
      \node[font=\tiny,anchor=east] at (0.6,0.8) {$\dots$};
      \node[font=\tiny,anchor=west] at (3.4,0.8) {$\dots$};
      \node[font=\tiny,anchor=east] at (0.6,-0.8) {$\dots$};
      \node[font=\tiny,anchor=west] at (3.4,-0.8) {$\dots$};      
    \end{tikzpicture} \ .
  \end{equation}
  Applying a product linear functional on the lower half of the r.h.s. (and the right indices of the MPS on the l.h.s.), the equation changes to 
  \begin{equation}
    \begin{tikzpicture}[baseline=-0.1cm]
      \foreach \i in {0,1,2}{
        \draw (\i,0)--(\i+1,0);
        \draw (\i+0.35,0)--+(0,0.4);
        \draw (\i+0.65,0)--+(0,0.4);
        \draw[fill=white] (\i+0.2,-0.2) rectangle (\i+0.8,0.2);
        \draw[fill=gray] (\i+0.65,0.5) circle (0.1);
      }
      \node[font=\tiny,anchor=east] at (0,0) {$\dots$};
      \node[font=\tiny,anchor=west] at (3.0,0) {$\dots$};
    \end{tikzpicture}\ =\ \lambda_n \times \ 
    \begin{tikzpicture}[baseline = 0.4cm]
      \scope
        \clip (0.5,0) rectangle (3.5,1.1+\d);
        \draw (0,0.8)--(4,0.8);
        \foreach \x in {0,...,3}{
          \pic[green,rotate=90] at (\x,0.5+\d) {halfstate};
          \draw[green] (\x,0.5)--+(0,0.6);
        }
      \endscope 
      \draw[fill=gray,rounded corners=3pt] (1-0.15,0.7) rectangle (3.15,0.9);
      \node[anchor=west,font=\scriptsize] at (3.8,0.8) {$X_{n}$};
      \node[font=\tiny,anchor=east] at (0.6,0.8) {$\dots$};
      \node[font=\tiny,anchor=west] at (3.4,0.8) {$\dots$};
    \end{tikzpicture} \ ,
  \end{equation}
  for some $\lambda_n\in \mathbb{C}$. Notice that the Schmidt vectors on the r.h.s. can be inverted: they are an injective mapping from the Schmidt index to the physical degrees of freedom, as they are linearly independent. Therefore,
  \begin{equation}\label{eq:Xn_mpo}
    \begin{tikzpicture}
      \foreach \i in {0,1,2}{
        \draw (\i,0)--(\i+1,0);
        \draw (\i+0.35,0)--+(0,0.8);
        \draw (\i+0.65,0)--+(0,0.4);
        \draw[fill=white] (\i+0.2,-0.2) rectangle (\i+0.8,0.2);
        \draw[fill=gray] (\i+0.65,0.5) circle (0.1);
        \draw[fill=white] (\i+0.35,0.5) circle (0.1);
      }
      \node[font=\tiny,anchor=east] at (0,0) {$\dots$};
      \node[font=\tiny,anchor=west] at (3.0,0) {$\dots$};
    \end{tikzpicture}\ =\ \lambda_n \times \ 
    \begin{tikzpicture}
      \scope
        \clip (0.5,0.5) rectangle (3.5,1.1);
        \draw (0,0.8)--(4,0.8);
        \foreach \x in {0,...,3}{
          \draw (\x,0.5)--+(0,0.6);
        }
      \endscope 
      \node[font=\tiny,anchor=east] at (0.6,0.8) {$\dots$};
      \node[font=\tiny,anchor=west] at (3.4,0.8) {$\dots$};
      \draw[fill=gray,rounded corners=3pt] (1-0.15,0.7) rectangle (3.15,0.9);
      \node[anchor=west,font=\scriptsize] at (3.8,0.8) {$X_{n}$};
    \end{tikzpicture} \ ,
  \end{equation}  
  where the white circle depicts the inverse of the Schmidt vectors of $\ket{\phi_B}$. This shows that $X_n$ (and similarly $X_n^{-1}$) is an MPS with some MPS tensor $\tilde{X}$ (and $\tilde{Y}$) as long as the l.h.s. is not 0. It is thus sufficient to prove that there is a translationally invariant product linear functional (the gray circles), which is independent of $n$, that does not map the l.h.s. to 0. 
  
  Consider two linear functionals acting on the MPS tensor:
  \begin{equation}
    \begin{tikzpicture}
      \foreach \i in {0,1,2}{
        \draw (\i,0)--(\i+1,0);
        \draw (\i+0.35,0)--+(0,0.4);
        \draw (\i+0.65,0)--+(0,0.4);
        \draw[fill=white] (\i+0.2,-0.2) rectangle (\i+0.8,0.2);
        \draw[fill=gray] (\i+0.65,0.5) circle (0.1);
        \draw[fill=white] (\i+0.35,0.5) circle (0.1);
      }
    \end{tikzpicture}\ =\ \tr\left\{ M^n\right\}. 
  \end{equation}
  We show now that there are linear functionals $a,b$ such that for the corresponding $M_{a,b}$ $\tr\{M_{a,b}\}\neq 0$. Let us consider the map  $F:(a,b)\mapsto \tr\{M_{a,b}\}$. Graphically, this map is
  \begin{equation}
    F = 
    \begin{tikzpicture}
      \draw (0,0)--(1,0);
      \draw (0.35,0)--+(0,0.4);
      \draw (0.65,0)--+(0,0.4);
      \draw[fill=white] (0.2,-0.2) rectangle (0.8,0.2);
      \draw (0,0)--(0,-0.35)--(1,-0.35)--(1,0);
    \end{tikzpicture}
    \ =\ 
    \begin{tikzpicture}
      \pic[blue,rotate=90] at (0.5,0.5+\d) {halfstate};
      \pic[blue,rotate=-90] at (0.5,-0.5-\d) {halfstate};
      \draw[red] (\d-\dd,-\d-\dd) rectangle (\d+\dd,+\d+\dd); 
      \draw[red] (1-\d-\dd,-\d-\dd) rectangle (1-\d+\dd,+\d+\dd); 
      \draw[red] (\d-\dd,0) --+(-0.15,0);
      \draw[red] (1-\d+\dd,0) --+(0.15,0);
      \draw[red] (\d-\dd-0.15,0)--(\d-\dd-0.15,-0.8)--(1-\d+\dd+0.15,-0.8)--(1-\d+\dd+0.15,0);
    \end{tikzpicture} \ .
  \end{equation}   
  Notice that $F$ equals to the operator $O$ with left and right side interchanged applied to the tensor product of the Schmidt vectors of $\ket{\phi}$. As $O$ is invertible, $F$ is not zero. Therefore there are linear functionals $a,b$ such that $F(a,b)=\tr\{M_{a,b}\}\neq 0$. As $\tr\{M_{a,b}\}\neq 0$, $M_{a,b}$ is not nilpotent and thus 
  \begin{equation}\label{eq:tr_m}
    \tr \{M_{a,b}^n\} = \sum_{i=1}^{R} \xi_i^n
  \end{equation}
  for some $\xi_1,\dots \xi_R \in \mathbb{C}\backslash\{0\}$, $R>0$. Let $S:=\{ n\in\mathbb{N}\mid\tr \{M_{a,b}^n\} \neq 0\}$. Notice that $|S|=\infty$. Then, choosing the linear functional appearing in \cref{eq:Xn_mpo} to be $b$, the l.h.s.\ is non-zero for all system sizes $n\in S$. Therefore, $X_n$ can be written as an MPO for all $n\in S$. Similarly, using the linear functional $a$ instead of $b$ on the lower part of \cref{eq:mps_double_boundary}, we arrive to the conclusion that $X_n^{-1}$ is also a non-zero MPO for all $n\in S$, for the same $S$. 
 
  Therefore there is a $\lambda_{n}\in \mathbb{C}$ such that $\forall n \in S$,
  \begin{equation}\label{eq:boundary1a}
  \begin{tikzpicture}
  \scope
  \clip (0.5,-0.5-\d) rectangle (3.5,0.5+\d);
  \foreach \x in {0,...,3}{
    \pic[blue,rotate=90] at (\x,0.5+\d) {halfstate};
    \pic[blue,rotate=-90] at (\x,-0.5-\d) {halfstate};
    \draw[red] (\x+0.5,0) circle (0.3);
  }
  \endscope
  \node[font=\tiny,anchor=east] at (0.6,0) {$\dots$};
  \node[font=\tiny,anchor=west] at (3.4,0) {$\dots$};
  \end{tikzpicture}
  =
  \lambda_{n} \mu^{n}\
    \begin{tikzpicture}[green,font=\scriptsize]
    \foreach \x in {0,...,2}{
      \pic[rotate=90] at (\x,0.5+\d) {halfstate};
      \pic[rotate=-90] at (\x,-0.5-\d) {halfstate};
      \pic[pic text = $\tilde{X}$, anchor=south west] at (\x,0.8) {boundarympo};
      \pic[pic text = $\tilde{Y}$, anchor=north west] at (\x,-0.8) {boundarympo};
    }
    \node[black,font=\tiny,anchor=east] at (-0.5,0.8) {$\dots$};
    \node[black,font=\tiny,anchor=west] at (2.5,0.8) {$\dots$};
    \node[black,font=\tiny,anchor=east] at (-0.5,-0.8) {$\dots$};
    \node[black,font=\tiny,anchor=west] at (2.5,-0.8) {$\dots$};
    \end{tikzpicture} \ .
  \end{equation}  
  Here, $\mu_n$ is the proportionality constant appearing in \cref{eq:boundary1}, and $V_n(\tilde{X})$ and $V_n(\tilde{Y})$ are translationally invariant MPOs on $n$ sites, such that $V_n(\tilde{Y}) = \left(V_n(\tilde{X})\right)^{-1}/\lambda_{n}$. Their defining tensors, $\tilde{X}$ and $\tilde{Y}$ are independent of $n$. Note that the MPOs $V_n(\tilde{X})$ and $V_n(\tilde{Y})$ are defined for $\forall n\in \mathbb{N}$, but we have not yet proven that \cref{eq:boundary1a} holds for $n\notin S$.
  
  In the following we prove that \cref{eq:boundary1a} also holds  $\forall n \in \mathbb{N}$ for some injective MPO $V_n(X)$, $V_n(Y)$ with $\lambda_n=1$. 
  
  Using \cref{cor:mps_decomp}, there exists $K\in \mathbb{N}$ such that after blocking $K$ tensors, both  $V_n(\tilde{X})$ and $V_n(\tilde{Y})$ ($n\in K\mathbb{N}$) can be decomposed into a linear combination of normal MPOs. As the tensor product of normal MPSs is again a normal MPS (\cref{prop:injective_tensor_prod}), 
  $V_n(\tilde{X})\otimes V_n(\tilde{Y})$ has a decomposition into normal MPO that are tensor products. Denote these essentially different normal MPO by $V_n(X_i)\otimes V_n(Y_i)$. That is, $\forall n\in K\mathbb{N}$
  \begin{equation}\label{eq:XY_normaldecomp}
    V_n(\tilde{X})\otimes V_n(\tilde{Y}) = \sum_{i=1}^L \sum_{j=1}^{M_i} \zeta_{ij}^n V_n(X_i)\otimes V_n(Y_i),
  \end{equation} 
  where $V_n(X_i)\otimes V_n(Y_i)$ are essentially different normal MPOs. Using this decomposition in \cref{eq:boundary1a}, the l.h.s.\ is described by a normal MPO (\cref{lem:boundaryMPSinjective}), while the r.h.s. is described by the sum above for an infinite number of  system sizes (indeed, for all $n\in K \mathbb{N}\cap S$). As essentially different MPSs become linearly independent for large system sizes (\cref{cor:normal_lin_indep}), \cref{eq:XY_normaldecomp} can describe a normal MPO only if either $L=1$ or otherwise all but one $i$ satisfy
  \begin{equation}
    \sum_{j=1}^{M_i} \zeta_{ij}^n =0 \quad \forall n\in S\cap K\mathbb{N}.
  \end{equation}
  Recalling that 
  \cref{eq:tr_m} vanishes $\forall n\in \mathbb{N}\backslash S$, we conclude that 
  \begin{equation}\label{eq:power_sum_zero}
    \sum_{k=1}^{R} \xi_k^n \sum_{j=1}^{M_i} \zeta_{ij}^n = \sum_{kj} (\xi_k \zeta_{ij})^n = 0   \quad \forall n\in K\mathbb{N},
  \end{equation}
  where $i$ is chosen such that the sum of $\zeta_{ij}^n$ vanishes $\forall n \in S \cap k\mathbb{N}$. Applying \cref{prop:sum_power} to \cref{eq:power_sum_zero}, all $(\xi_k\zeta_{ij})^K=0$, that is, $\zeta_{ij}=0$ for all $j$ and all but one $i$. Therefore, $L=1$ in \cref{eq:XY_normaldecomp}. Using \cref{prop:periodic_mps_decomp}, we conclude that $V_n(\tilde{X})\otimes V_n(\tilde{Y})$ does not contain periodic MPO, therefore $K=1$. Thus, both the l.h.s.\ and the r.h.s.\ of \cref{eq:boundary1a} are proportional to normal MPOs. Using \cref{prop:can_form_normal_mps}, we conclude that the equality in \cref{eq:boundary1a} holds $\forall n \in \mathbb{N}$. We have therefore proven that there are normal MPO tensors $X$ and $Y$ (the ones appearing in the unique normal MPO in \cref{eq:XY_normaldecomp}) such that $\forall n \in \mathbb{N}$ and some $\lambda_n \in \mathbb{C}$
  \begin{equation}
  \begin{tikzpicture}
  \scope
  \clip (0.5,-0.5-\d) rectangle (3.5,0.5+\d);
  \foreach \x in {0,...,3}{
    \pic[blue,rotate=90] at (\x,0.5+\d) {halfstate};
    \pic[blue,rotate=-90] at (\x,-0.5-\d) {halfstate};
    \draw[red] (\x+0.5,0) circle (0.3);
  }
  \endscope
  \node[font=\tiny,anchor=east] at (0.6,0) {$\dots$};
  \node[font=\tiny,anchor=west] at (3.4,0) {$\dots$};
  \end{tikzpicture}
  =
  \lambda_{n} \mu^{n}\
  \begin{tikzpicture}[green,font=\scriptsize]
    \foreach \x in {0,...,2}{
      \pic[rotate=90] at (\x,0.5+\d) {halfstate};
      \pic[rotate=-90] at (\x,-0.5-\d) {halfstate};
      \pic[pic text = $X$, anchor=south west] at (\x,0.8) {boundarympo};
      \pic[pic text = $Y$, anchor=north west] at (\x,-0.8) {boundarympo};
    }
    \node[black,font=\tiny,anchor=east] at (-0.5,0.8) {$\dots$};
    \node[black,font=\tiny,anchor=west] at (2.5,0.8) {$\dots$};
    \node[black,font=\tiny,anchor=east] at (-0.5,-0.8) {$\dots$};
    \node[black,font=\tiny,anchor=west] at (2.5,-0.8) {$\dots$};
  \end{tikzpicture} \ .
  \end{equation}  
  These MPO tensors also satisfy $V_n(Y) = \left(V_n(X)\right)^{-1}/\lambda_{n} $ for all  $n \in S$. As  both $V_n(Y)$ and $V_n(X)$ are normal MPOs, the equality holds $\forall n\in \mathbb{N}$ and thus $\lambda_n=\lambda^n$ for some $\lambda\in \mathbb{C}$. Absorbing this constant into $Y$, $V_n(Y) = \left(V_n(X)\right)^{-1}$ and 
  \begin{equation}
  \begin{tikzpicture}
  \scope
  \clip (0.5,-0.5-\d) rectangle (3.5,0.5+\d);
  \foreach \x in {0,...,3}{
    \pic[blue,rotate=90] at (\x,0.5+\d) {halfstate};
    \pic[blue,rotate=-90] at (\x,-0.5-\d) {halfstate};
    \draw[red] (\x+0.5,0) circle (0.3);
  }
  \endscope
  \node[font=\tiny,anchor=east] at (0.6,0) {$\dots$};
  \node[font=\tiny,anchor=west] at (3.4,0) {$\dots$};
  \end{tikzpicture}
  =
  \mu^{n}\
  \begin{tikzpicture}[green,font=\scriptsize]
    \foreach \x in {0,...,2}{
      \pic[rotate=90] at (\x,0.5+\d) {halfstate};
      \pic[rotate=-90] at (\x,-0.5-\d) {halfstate};
      \pic[pic text = $X$, anchor=south west] at (\x,0.8) {boundarympo};
      \pic[pic text = $Y$, anchor=north west] at (\x,-0.8) {boundarympo};
    }
    \node[black,font=\tiny,anchor=east] at (-0.5,0.8) {$\dots$};
    \node[black,font=\tiny,anchor=west] at (2.5,0.8) {$\dots$};
    \node[black,font=\tiny,anchor=east] at (-0.5,-0.8) {$\dots$};
    \node[black,font=\tiny,anchor=west] at (2.5,-0.8) {$\dots$};
  \end{tikzpicture} \ .
  \end{equation}  
\end{proof}

\begin{corollary}\label{cor:boundary_unique}
  Suppose that $\forall n\in\mathbb{N}$ \cref{eq:boundary1} holds also for some other MPO $V_n(\tilde{X})$ and $V_n(\tilde{Y})$ and $V_n(\tilde{Y}) = \left(V_n(\tilde{X})\right)^{-1}$. Then $V_n(\tilde{X}) = \lambda^n V_n(X)$ and $V_n(\tilde{Y}) = \lambda^{-n} Y^{(n)}$ for some $\lambda\in\mathbb{C}$.
\end{corollary}

\begin{proof}
  Due to uniqueness of the gauge in \cref{eq:boundary0}, $V_n(\tilde{X}) = \lambda_n V_n(X)$ and $V_n(\tilde{Y}) = \lambda_n^{-1} V_n(Y)$. Decomposing $V_n(\tilde{X})$ and $V_n(\tilde{Y})$ to their canonical forms, we see that the only normal MPS appearing in the decomposition is $V_n(X)$ and $V_n(Y)$, and that $\lambda_n = \sum_i \lambda_i^n$ and $\lambda_n^{-1} = \sum_i \eta_i^n$. But then $1= \sum_{ij} (\lambda_i\eta_j)^n$ and thus by \cref{prop:sum_power}, $\lambda_n = \lambda^n$.
\end{proof}

It turns out that the fact that the boundaries of the two semi-injective PEPS are  related by an MPO severely restricts the form of $O$. We will indeed find that 

\begin{prop}\label{prop:o_twolayer}
  The operator $O$ from \cref{eq:setup} can be written as a product of invertible two-body operators:
  \begin{equation}
    O = \left(O_{14}\otimes O_{23}\right) \cdot \left(O_{12} \otimes O_{34}\right) = \left(\tilde{O}_{12} \otimes \tilde{O}_{34}\right) \cdot \left(\tilde{O}_{14}\otimes \tilde{O}_{23}\right)\ ,
  \end{equation}
  where the particles are numbered clockwise from the upper left corner and $O_{ij}$ acts on particles $i$ and $j$.  Pictorially,
  \begin{equation}\label{eq:op_twolayer}
    \begin{tikzpicture}[baseline=-0.1cm]
    \draw[red] (0,0) circle (0.5);
    \end{tikzpicture} =
    \begin{tikzpicture}[baseline=-0.1cm,rotate=90]
    \draw[red] (-0.5,-0.5) rectangle (-0.2,0.5);
    \draw[red] (0.2,-0.5) rectangle (0.5,0.5);
    \draw[red] (-0.6,-0.6) rectangle (0.6,-0.1);
    \draw[red] (-0.6,0.1) rectangle (0.6,0.6);      
    \end{tikzpicture} =
    \begin{tikzpicture}[baseline=-0.1cm]
    \draw[red] (-0.5,-0.5) rectangle (-0.2,0.5);
    \draw[red] (0.2,-0.5) rectangle (0.5,0.5);
    \draw[red] (-0.6,-0.6) rectangle (0.6,-0.1);
    \draw[red] (-0.6,0.1) rectangle (0.6,0.6);      
    \end{tikzpicture} \ .
  \end{equation}

\end{prop}

We will prove that $O$ has a four site long non-translationally invariant MPO decomposition, with the property that cutting the MPO into two halves yields a minimal rank decomposition of $O$. Moreover, we will show that the product of the Schmidt vectors of $O$ and $O^{-1}$ are tensor products. Before proceeding to the proof, we show that if $O$ and $O^{-1}$ are both MPO of this form, $O$ has to have the two-layer structure \eqref{eq:op_twolayer}.

\begin{lemma}\label{thm:mpo_layers}
Consider two non-translationally invariant MPOs on $n=2k$ sites with tensors $X_1,\dots X_n$ and $Y_1,\dots Y_n$. Suppose that 
\begin{enumerate}
\item $V(X_1,\dots, X_n)\cdot V(Y_1,\dots,Y_n) =\id$
\item Both $X_i X_{i+1}$ and $Y_i Y_{i+1}$ are injective for all $i=1,\dots,n$ with $n+1\equiv 1$.
\item The product of $X_i X_{i+1}$ and $Y_i Y_{i+1}$ factorizes as depicted:
\begin{align}\label{eq:mpofactor1}
\begin{tikzpicture}[scale=0.5]
  \pic[pic text = $X_{i}$, font=\footnotesize] at (0,1) {mpotensor=gray};
  \pic[pic text = $X_{i+1}$, font=\footnotesize] at (1,1) {mpotensor=gray};
  \pic[pic text = $Y_{i}$, font=\footnotesize] at (0,0) {mpotensor};
  \pic[pic text = $Y_{i+1}$, font=\footnotesize] at (1,0) {mpotensor};
\end{tikzpicture}
&=
\begin{tikzpicture}[scale=0.5]
\draw[thick] (0,-0.5)--(0,1.5);
\draw[thick] (1,-0.5)--(1,1.5);
\draw[thick] (-0.5,0)--(1.5,0);
\draw[thick] (-0.5,1)--(1.5,1);
\fill[white] (-0.2,-0.2) rectangle (1.2,1.2);
\draw[thick, rounded corners=3pt] (-0.2,-0.2) rectangle (0.2,1.2);
\draw[thick, rounded corners=3pt] (0.8,-0.2) rectangle (1.2,1.2);
\end{tikzpicture}\\ \label{eq:mpofactor2}
\begin{tikzpicture}[scale=0.5]
  \pic[pic text = $Y_{i}$, font=\footnotesize] at (0,1) {mpotensor};
  \pic[pic text = $Y_{i+1}$, font=\footnotesize] at (1,1) {mpotensor};
  \pic[pic text = $X_{i}$, font=\footnotesize] at (0,0) {mpotensor=gray};
  \pic[pic text = $X_{i+1}$, font=\footnotesize] at (1,0) {mpotensor=gray};
\end{tikzpicture}
&=
\begin{tikzpicture}[scale=0.5]
\draw[thick] (0,-0.5)--(0,1.5);
\draw[thick] (1,-0.5)--(1,1.5);
\draw[thick] (-0.5,0)--(1.5,0);
\draw[thick] (-0.5,1)--(1.5,1);
\fill[white] (-0.2,-0.2) rectangle (1.2,1.2);
\draw[thick, rounded corners=3pt, pattern=north west lines] (-0.2,-0.2) rectangle (0.2,1.2);
\draw[thick, rounded corners=3pt, pattern=north west lines] (0.8,-0.2) rectangle (1.2,1.2);
\end{tikzpicture}.
\end{align}
\end{enumerate}
Then $V(X_1,\dots, X_n)$ (and $V(Y_1,\dots, Y_n)$) admits a two layer description:
\begin{equation} \label{eq:mpo2layer1}
\begin{tikzpicture}[scale=0.5]
  \node at (-1,0) {$\dots$};
  \pic[pic text= $Y_1$, font=\footnotesize] at (0,0) {mpotensor};
  \pic[pic text= $Y_2$, font=\footnotesize] at (1,0) {mpotensor};
  \pic[pic text= $Y_3$, font=\footnotesize] at (2,0) {mpotensor};
  \pic[pic text= $Y_4$, font=\footnotesize] at (3,0) {mpotensor};
  \node at (4,0) {$\dots$};
\end{tikzpicture}
=
\begin{tikzpicture}[scale=0.5]
  \node at (-1,0.5) {$\dots$};
  \node at (4,0.5) {$\dots$};
  \scope
    \clip (-0.5,-0.5) rectangle (3.5,1.5);
    \foreach \x in {0,1,2,3}{
      \draw[thick] (\x,-0.5) -- (\x,1.5);
    }
    \foreach \x in {0,2}{
      \draw[thick, fill=white,rounded corners] (\x-0.2,-0.2) rectangle (\x+1.2,0.2);
    }
    \foreach \x in {-1,1,3}{
      \draw[thick, fill=white,rounded corners] (\x-0.2,0.8) rectangle (\x+1.2,1.2);
    }
  \endscope
\end{tikzpicture}  ,  
\end{equation}
where all two-body operators on the r.h.s.\ are invertible. \cref{eq:mpo2layer1} also holds when shifted by one site (with other invertible operators):
\begin{equation} \label{eq:mpo2layer2}
\begin{tikzpicture}[scale=0.5]
  \node at (-1,0) {$\dots$};
  \pic[pic text= $Y_1$, font=\footnotesize] at (0,0) {mpotensor};
  \pic[pic text= $Y_2$, font=\footnotesize] at (1,0) {mpotensor};
  \pic[pic text= $Y_3$, font=\footnotesize] at (2,0) {mpotensor};
  \pic[pic text= $Y_4$, font=\footnotesize] at (3,0) {mpotensor};
  \node at (4,0) {$\dots$};
\end{tikzpicture}
=
\begin{tikzpicture}[scale=0.5]
  \node at (-1,0.5) {$\dots$};
  \node at (4,0.5) {$\dots$};
  \scope
    \clip (-0.5,-0.5) rectangle (3.5,1.5);
    \foreach \x in {0,1,2,3}{
      \draw[thick] (\x,-0.5) -- (\x,1.5);
    }
    \foreach \x in {0,2}{
      \draw[thick, fill=white, postaction={pattern=north west lines},rounded corners] (\x-0.2,0.8) rectangle (\x+1.2,1.2);
    }
    \foreach \x in {-1,1,3}{
      \draw[thick, fill=white, postaction={pattern=north west lines},rounded corners] (\x-0.2,-0.2) rectangle (\x+1.2,0.2);
    }
  \endscope
\end{tikzpicture} \ .    
\end{equation}
\end{lemma}

Note that for the translationally invariant setting, conditions 2 and 3 are satisfied naturally after blocking some tensors. 

\begin{proof}
  Take a Schmidt decomposition of the tensors $X_1,\dots, X_n$ and $Y_1,\dots, Y_n$ in an alternating way:
  \begin{align}
  \begin{tikzpicture}[scale=0.5]
    \node at (-1,0) {$\dots$};
    \pic[pic text= $X_1$, font=\footnotesize] at (0,0) {mpotensor=gray};
    \pic[pic text= $X_2$, font=\footnotesize] at (1,0) {mpotensor=gray};
    \pic[pic text= $X_3$, font=\footnotesize] at (2,0) {mpotensor=gray};
    \pic[pic text= $X_4$, font=\footnotesize] at (3,0) {mpotensor=gray};
    \node at (4,0) {$\dots$};
  \end{tikzpicture}
  &=
  \begin{tikzpicture}[scale=0.5]
  \node at (-0.9,0.5) {$\dots$};
  \node at (3.9,0.5) {$\dots$};
  \foreach \x in {0,2}{
    \draw[thick] (\x,1)--++(-0.5,0);
    \draw[thick] (\x,0)--++(0.5,0);
  }
  \foreach \x in {1,3}{
    \draw[thick] (\x,0)--++(-0.5,0);
    \draw[thick] (\x,1)--++(0.5,0);
  }
  \foreach \x in {0,1,2,3}{
    \draw[thick] (\x,-0.5) -- (\x,1.5);
    \draw[thick,fill=gray] (\x,0) circle (0.2);
    \draw[thick,fill=gray] (\x,1) circle (0.2);
  }
  \end{tikzpicture}\\
  \begin{tikzpicture}[scale=0.5]
    \node at (-1,0) {$\dots$};
    \pic[pic text= $Y_1$, font=\footnotesize] at (0,0) {mpotensor};
    \pic[pic text= $Y_2$, font=\footnotesize] at (1,0) {mpotensor};
    \pic[pic text= $Y_3$, font=\footnotesize] at (2,0) {mpotensor};
    \pic[pic text= $Y_4$, font=\footnotesize] at (3,0) {mpotensor};
    \node at (4,0) {$\dots$};
  \end{tikzpicture}
  &=
  \begin{tikzpicture}[scale=0.5]
  \node at (-0.9,0.5) {$\dots$};
  \node at (3.9,0.5) {$\dots$};
  \foreach \x in {0,2}{
    \draw[thick] (\x,0)--++(-0.5,0);
    \draw[thick] (\x,1)--++(0.5,0);
  }
  \foreach \x in {1,3}{
    \draw[thick] (\x,1)--++(-0.5,0);
    \draw[thick] (\x,0)--++(0.5,0);
  }
  \foreach \x in {0,1,2,3}{
    \draw[thick] (\x,-0.5) -- (\x,1.5);
    \draw[thick,fill=white] (\x,0) circle (0.2);
    \draw[thick,fill=white] (\x,1) circle (0.2);
  }
  \end{tikzpicture}    
  \end{align}
  We will prove that the two-body operators defined this way are invertible. They naturally have to be injective from the outside to the middle indices, otherwise $V(X_1,\dots X_n)$ and $V(Y_1,\dots Y_n)$ would not be invertible. Suppose that there is an operator which is not injective from the middle to the outside. Suppose it happens in the lower layer of $V(X_1,\dots X_n)$. Consider a 2-site part of the MPO.
  \begin{equation}
  \begin{tikzpicture}[scale=0.5]
  \draw[thick] (0,-0.5)--(0,1.5);
  \draw[thick] (1,-0.5)--(1,1.5);
  \draw[thick] (-0.5,0)--(1.5,0);
  \draw[thick] (-0.5,1)--(1.5,1);
  \fill[white] (-0.2,-0.2) rectangle (1.2,1.2);
  \draw[thick, rounded corners=3pt] (-0.2,-0.2) rectangle (0.2,1.2);
  \draw[thick, rounded corners=3pt] (0.8,-0.2) rectangle (1.2,1.2);
  \end{tikzpicture}
  =
  \begin{tikzpicture}[scale=0.5]
    \pic[pic text = $X_{i}$, font=\footnotesize] at (0,1) {mpotensor=gray};
    \pic[pic text = $X_{i+1}$, font=\footnotesize] at (1,1) {mpotensor=gray};
    \pic[pic text = $Y_{i}$, font=\footnotesize] at (0,0) {mpotensor};
    \pic[pic text = $Y_{i+1}$, font=\footnotesize] at (1,0) {mpotensor};
  \end{tikzpicture}
  =
  \begin{tikzpicture}[scale=0.5]
  \foreach \x in {0}{
    \draw[thick] (\x,3)--++(-0.5,0);
    \draw[thick] (\x,2)--++(0.5,0);
  }
  \foreach \x in {1}{
    \draw[thick] (\x,2)--++(-0.5,0);
    \draw[thick] (\x,3)--++(0.5,0);
  }
  \foreach \x in {0,1}{
    \draw[thick] (\x,1.5) -- (\x,3.5);
    \draw[thick,fill=gray] (\x,2) circle (0.2);
    \draw[thick,fill=gray] (\x,3) circle (0.2);
  }
  \foreach \x in {0}{
    \draw[thick] (\x,0)--++(-0.5,0);
    \draw[thick] (\x,1)--++(0.5,0);
  }
  \foreach \x in {1}{
    \draw[thick] (\x,1)--++(-0.5,0);
    \draw[thick] (\x,0)--++(0.5,0);
  }
  \foreach \x in {0,1}{
    \draw[thick] (\x,-0.5) -- (\x,1.5);
    \draw[thick,fill=white] (\x,0) circle (0.2);
    \draw[thick,fill=white] (\x,1) circle (0.2);
  }
  \end{tikzpicture}    
  \end{equation}
  As we took a minimal rank decomposition, the outer tensors on the l.h.s. are invertible. Therefore, the product of the operators in the middle is a product:
  \begin{equation}
  \begin{tikzpicture}[scale=0.5]
  \draw[thick] (0,-0.5)--(0,1.5);
  \draw[thick] (1,-0.5)--(1,1.5);
  \draw[thick] (0,0)--(1,0);
  \draw[thick] (0,1)--(1,1);
  \foreach \x in {(0,0),(1,0)}{
    \draw[thick,fill=white] \x circle (0.2);
  }
  \foreach \x in {(0,1),(1,1)}{
    \draw[thick,fill=gray] \x circle (0.2);
  }
  \end{tikzpicture}
  =
  \begin{tikzpicture}[scale=0.5]
  \draw[thick] (0,-0.5)--(0,1.5);
  \draw[thick] (1,-0.5)--(1,1.5);
  \draw[thick] (0,0)--(1,0);
  \draw[thick] (0,1)--(1,1);
  \fill[white] (-0.2,-0.2) rectangle (1.2,1.2);
  \draw[thick, rounded corners=3pt] (-0.2,-0.2) rectangle (0.2,1.2);
  \draw[thick, rounded corners=3pt] (0.8,-0.2) rectangle (1.2,1.2);
  \end{tikzpicture}
  \end{equation}
  Therefore if the gray operator is not injective from top to bottom, then its kernel factorizes. Suppose the left operator on the r.h.s.\ has a non trivial kernel. Then we can insert a non-trivial projector $y$ on top that does not change the value of the product:
  \begin{equation}
  \begin{tikzpicture}[baseline=0.2cm,scale=0.5]
  \draw[thick] (0,-0.5)--(0,2.5);
  \draw[thick] (1,-0.5)--(1,2.5);
  \draw[thick] (0,0)--(1,0);
  \draw[thick] (0,1)--(1,1);
  \foreach \x in {(0,0),(1,0)}{
    \draw[thick,fill=white] \x circle (0.2);
  }
  \foreach \x in {(0,1),(1,1)}{
    \draw[thick,fill=gray] \x circle (0.2);
  }
  \draw[thick, fill=yellow] (0,2) circle (0.2);
  \end{tikzpicture}
  =
  \begin{tikzpicture}[scale=0.5]
  \draw[thick] (0,-0.5)--(0,1.5);
  \draw[thick] (1,-0.5)--(1,1.5);
  \draw[thick] (0,0)--(1,0);
  \draw[thick] (0,1)--(1,1);
  \foreach \x in {(0,0),(1,0)}{
    \draw[thick,fill=white] \x circle (0.2);
  }
  \foreach \x in {(0,1),(1,1)}{
    \draw[thick,fill=gray] \x circle (0.2);
  }
  \end{tikzpicture}.
  \end{equation}
  Inserting this back into the product $V(X_1,\dots X_n)\cdot V(Y_1,\dots Y_n)$, we get that 
  \begin{equation}
  \begin{tikzpicture}[scale=0.5]
    \node at (-0.9,2) {$\dots$};
    \node at (3.9,2) {$\dots$};  \foreach \x in {0,2}{
      \draw[thick] (\x,4)--++(-0.5,0);
      \draw[thick] (\x,2)--++(0.5,0);
    }
    \foreach \x in {1,3}{
      \draw[thick] (\x,2)--++(-0.5,0);
      \draw[thick] (\x,4)--++(0.5,0);
    }
    \foreach \x in {0,1,2,3}{
      \draw[thick] (\x,1.5) -- (\x,4.5);
      \draw[thick,fill=gray] (\x,2) circle (0.2);
      \draw[thick,fill=gray] (\x,4) circle (0.2);
    }
    \foreach \x in {0,2}{
      \draw[thick] (\x,0)--++(-0.5,0);
      \draw[thick] (\x,1)--++(0.5,0);
    }
    \foreach \x in {1,3}{
      \draw[thick] (\x,1)--++(-0.5,0);
      \draw[thick] (\x,0)--++(0.5,0);
    }
    \foreach \x in {0,1,2,3}{
      \draw[thick] (\x,-0.5) -- (\x,1.5);
      \draw[thick,fill=white] (\x,0) circle (0.2);
      \draw[thick,fill=white] (\x,1) circle (0.2);
    }
    \draw[fill=yellow] (0,3) circle (0.2);
  \end{tikzpicture}
  =\ 
  \begin{tikzpicture}[scale=0.5]
    \node at (-0.9,2) {$\dots$};
    \node at (3.9,2) {$\dots$};
    \foreach \x in {0,1,2,3}{
      \draw[thick] (\x,-0.5) -- (\x,4.5);
    }
  \end{tikzpicture}    
  \end{equation}
  As $V(Y_1,\dots Y_n)$ is invertible, its left inverse is unique and equal to $V(X_1,\dots , X_n)$. Therefore 
  \begin{equation}
  \begin{tikzpicture}[scale=0.5]
    \node at (-0.9,1) {$\dots$};
    \node at (3.9,1) {$\dots$};
    \foreach \x in {0,2}{
      \draw[thick] (\x,2)--++(-0.5,0);
      \draw[thick] (\x,0)--++(0.5,0);
    }
    \foreach \x in {1,3}{
      \draw[thick] (\x,0)--++(-0.5,0);
      \draw[thick] (\x,2)--++(0.5,0);
    }
    \foreach \x in {0,1,2,3}{
      \draw[thick] (\x,-0.5) -- (\x,2.5);
      \draw[thick,fill=gray] (\x,0) circle (0.2);
      \draw[thick,fill=gray] (\x,2) circle (0.2);
    }
    \draw[fill=yellow] (0,1) circle (0.2);
  \end{tikzpicture}
  =
  \begin{tikzpicture}[scale=0.5]
    \node at (-0.9,0.5) {$\dots$};
    \node at (3.9,0.5) {$\dots$};
    \foreach \x in {0,2}{
      \draw[thick] (\x,1)--++(-0.5,0);
      \draw[thick] (\x,0)--++(0.5,0);
    }
    \foreach \x in {1,3}{
      \draw[thick] (\x,0)--++(-0.5,0);
      \draw[thick] (\x,1)--++(0.5,0);
    }
    \foreach \x in {0,1,2,3}{
      \draw[thick] (\x,-0.5) -- (\x,1.5);
      \draw[thick,fill=gray] (\x,0) circle (0.2);
      \draw[thick,fill=gray] (\x,1) circle (0.2);
    }
  \end{tikzpicture}
  \end{equation}
  By assumption, the tensors defining the MPO are injective after blocking at least two sites. Therefore, by inverting all but one tensor, we conclude that 
  \begin{equation}
  \begin{tikzpicture}[scale=0.5]
  \foreach \x in {0}{
    \draw[thick] (\x,2)--++(-0.5,0);
    \draw[thick] (\x,0)--++(0.5,0);
  }
  \foreach \x in {0}{
    \draw[thick] (\x,-0.5) -- (\x,2.5);
    \draw[thick,fill=gray] (\x,0) circle (0.2);
    \draw[thick,fill=gray] (\x,2) circle (0.2);
  }
  \draw[fill=yellow] (0,1) circle (0.2);
  \end{tikzpicture}
  =
  \begin{tikzpicture}[scale=0.5]
  \foreach \x in {0}{
    \draw[thick] (\x,1)--++(-0.5,0);
    \draw[thick] (\x,0)--++(0.5,0);
  }
  \foreach \x in {0}{
    \draw[thick] (\x,-0.5) -- (\x,1.5);
    \draw[thick,fill=gray] (\x,0) circle (0.2);
    \draw[thick,fill=gray] (\x,1) circle (0.2);
  }
  \end{tikzpicture}
  \end{equation}
  But this is not possible unless the yellow tensor is the identity. Thus, the two-body operators are invertible.
\end{proof}

We  now proceed to the proof of \cref{prop:o_twolayer}. Note that it is enough to show that both $O$ and $O^{-1}$ admit an MPO description that satisfy the conditions of \cref{thm:mpo_layers}.

\begin{proof}[Proof of \cref{prop:o_twolayer}]

  Write the l.h.s.\ of \cref{eq:boundary1} as an injective MPS with tensors defined in \cref{eq:boundary_mps_decomp}. The r.h.s.\ of \cref{eq:boundary1} is also an injective MPS. Therefore, the generating tensors are related by a gauge transformation:
  \begin{equation}
    \begin{tikzpicture}
      \pic[blue,rotate=90] at (0.5,0.5+\d) {halfstate};
      \pic[blue,rotate=-90] at (0.5,-0.5-\d) {halfstate};
      \draw[red] (\d-\dd,-\d-\dd) rectangle (\d+\dd,+\d+\dd); 
      \draw[red] (1-\d-\dd,-\d-\dd) rectangle (1-\d+\dd,+\d+\dd); 
      \draw[red] (\d-\dd,0) --+(-0.15,0);
      \draw[red] (1-\d+\dd,0) --+(0.15,0);
      \scope[xshift=0.30cm]
        \draw[red,rounded corners=2pt] (1-\d-\dd,-\d-\dd) rectangle (1-\d+\dd,+\d+\dd);       
        \draw[red] (1-\d+\dd,0.15) --+(0.15,0);
        \draw[red] (1-\d+\dd,-0.15) --+(0.15,0);      
      \endscope
      \scope[xshift=-0.30cm]
        \draw[red,rounded corners=2pt] (\d-\dd,-\d-\dd) rectangle (\d+\dd,+\d+\dd); 
        \draw[red] (\d-\dd,0.15) --+(-0.15,0);
        \draw[red] (\d-\dd,-0.15) --+(-0.15,0);
      \endscope 

    \end{tikzpicture} \ 
    =\ \mu \
    \begin{tikzpicture}[green,font=\scriptsize]
      \pic[rotate=90] at (0,0.5+\d) {halfstate};
      \pic[rotate=-90] at (0,-0.5-\d) {halfstate};
      \pic[anchor = south west, pic text =$X$] at (0,0.8) {boundarympo};
      \pic[anchor = north west, pic text =$Y$] at (0,-0.8) {boundarympo};
    \end{tikzpicture} \ .
  \end{equation}  
  Absorbing the gauge in the decomposition of the operator, we have
  \begin{equation}\label{eq:boundary_mpo_tensor_def}
    \begin{tikzpicture}
      \pic[blue,rotate=90] at (0.5,0.5+\d) {halfstate};
      \pic[blue,rotate=-90] at (0.5,-0.5-\d) {halfstate};
      \draw[red] (\d-\dd,-\d-\dd) rectangle (\d+\dd,+\d+\dd); 
      \draw[red] (1-\d-\dd,-\d-\dd) rectangle (1-\d+\dd,+\d+\dd); 
      \draw[red] (\d-\dd,0.15) --+(-0.15,0);
      \draw[red] (1-\d+\dd,0.15) --+(0.15,0);
      \draw[red] (\d-\dd,-0.15) --+(-0.15,0);
      \draw[red] (1-\d+\dd,-0.15) --+(0.15,0);
    \end{tikzpicture} \
    =\ \mu \
    \begin{tikzpicture}[green,font=\scriptsize]
      \pic[rotate=90] at (0,0.5+\d) {halfstate};
      \pic[rotate=-90] at (0,-0.5-\d) {halfstate};
      \pic[anchor = south west, pic text =$X$] at (0,0.8) {boundarympo};
      \pic[anchor = north west, pic text =$Y$] at (0,-0.8) {boundarympo};
    \end{tikzpicture} \ ,
  \end{equation}
  where the red rectangles depict a minimal rank decomposition of the operator $O$. As $V_n(X)$ and $V_n(Y)$ are inverses of each other, the inverse relation of \cref{eq:boundary1} reads
  \begin{equation} \label{eq:boundary_inv}
    \begin{tikzpicture}[blue,font=\scriptsize]
        \foreach \x in {0,...,2}{
          \pic[rotate=90] at (\x,0.5+\d) {halfstate};
          \pic[rotate=-90] at (\x,-0.5-\d) {halfstate};
          \pic[pic text = $Y$, anchor=south west] at (\x,0.8) {boundarympo};
          \pic[pic text = $X$, anchor=north west] at (\x,-0.8) {boundarympo};
        }
    \node[black,font=\tiny,anchor=east] at (-0.5,0.8) {$\dots$};
    \node[black,font=\tiny,anchor=west] at (2.5,0.8) {$\dots$};
    \node[black,font=\tiny,anchor=east] at (-0.5,-0.8) {$\dots$};
    \node[black,font=\tiny,anchor=west] at (2.5,-0.8) {$\dots$};
    \end{tikzpicture} \ 
    =
    \mu^n\
    \begin{tikzpicture}
    \scope
    \clip (0.5,-0.5-\d) rectangle (3.5,0.5+\d);
    \foreach \x in {0,...,3}{
      \pic[green,rotate=90] at (\x,0.5+\d) {halfstate};
      \pic[green,rotate=-90] at (\x,-0.5-\d) {halfstate};
      \draw[red,dashed] (\x+0.5,0) circle (0.3);
    }
    \endscope
    \node[font=\tiny,anchor=east] at (0.6,0) {$\dots$};
    \node[font=\tiny,anchor=west] at (3.4,0) {$\dots$};
    \end{tikzpicture}\ ,
  \end{equation}  
  where the dashed red circles denote $O^{-1}$. Therefore, with an appropriate minimal rank decomposition of $O^{-1}$, the generating tensors are related as follows:
  \begin{equation}
    \begin{tikzpicture}[blue,font=\scriptsize]
      \pic[rotate=90] at (0,0.5+\d) {halfstate};
      \pic[rotate=-90] at (0,-0.5-\d) {halfstate};
      \pic[anchor = south west, pic text =$Y$] at (0,0.8) {boundarympo};
      \pic[anchor = north west, pic text =$X$] at (0,-0.8) {boundarympo};
    \end{tikzpicture} \ 
    =\ \mu\ 
    \begin{tikzpicture}
      \pic[green,rotate=90] at (0.5,0.5+\d) {halfstate};
      \pic[green,rotate=-90] at (0.5,-0.5-\d) {halfstate};
      \draw[red,densely dashed] (\d-\dd,-\d-\dd) rectangle (\d+\dd,+\d+\dd); 
      \draw[red,densely dashed] (1-\d-\dd,-\d-\dd) rectangle (1-\d+\dd,+\d+\dd); 
      \draw[red] (\d-\dd,0.15) --+(-0.15,0);
      \draw[red] (1-\d+\dd,0.15) --+(0.15,0);
      \draw[red] (\d-\dd,-0.15) --+(-0.15,0);
      \draw[red] (1-\d+\dd,-0.15) --+(0.15,0);
    \end{tikzpicture} \ ,
  \end{equation}  
  where the dashed rectangles denote the Schmidt decomposition of $O^{-1}$. Therefore, applying the Schmidt vectors of $O$ and then $O^{-1}$ to the Schmidt vectors of $\ket{\phi_A}$, we obtain
  \begin{equation}\label{eq:schmidt_then_inverse}
    \begin{tikzpicture}
      \pic[blue,rotate=90] at (0.5,0.5+\d) {halfstate};
      \pic[blue,rotate=-90] at (0.5,-0.5-\d) {halfstate};
      \draw[red] (\d-\dd,-\d-\dd) rectangle (\d+\dd,+\d+\dd); 
      \draw[red] (1-\d-\dd,-\d-\dd) rectangle (1-\d+\dd,+\d+\dd); 
      \draw[red] (\d-\dd,0.15) --+(-0.15-\dd,0);
      \draw[red] (1-\d+\dd,0.15) --+(0.15+\dd,0);
      \draw[red] (\d-\dd,-0.15) --+(-0.15-\dd,0);
      \draw[red] (1-\d+\dd,-0.15) --+(0.15+\dd,0);
      \draw[red,densely dashed] (\d-2*\dd,-\d-2*\dd) rectangle (\d+2*\dd,+\d+2*\dd); 
      \draw[red,densely dashed] (1-\d-2*\dd,-\d-2*\dd) rectangle (1-\d+2*\dd,+\d+2*\dd); 
      \draw[red] (\d-2*\dd,0.2) --+(-0.2,0);
      \draw[red] (1-\d+2*\dd,0.2) --+(0.2,0);
      \draw[red] (\d-2*\dd,-0.2) --+(-0.2,0);
      \draw[red] (1-\d+2*\dd,-0.2) --+(0.2,0);
    \end{tikzpicture} \
    =\ \mu\ 
    \begin{tikzpicture}[font=\scriptsize]
      \pic[green,rotate=90] at (0,0.5+\d) {halfstate};
      \pic[green,rotate=-90] at (0,-0.5-\d) {halfstate};
      \pic[green,anchor = south west, pic text =$X$] at (0,0.8) {boundarympo};
      \pic[green, anchor = north west, pic text =$Y$] at (0,-0.8) {boundarympo};
      \begin{scope}[xshift=-0.5cm]
        \draw[red,densely dashed] (\d-\dd,-\d-\dd) rectangle (\d+\dd,+\d+\dd); 
        \draw[red,densely dashed] (1-\d-\dd,-\d-\dd) rectangle (1-\d+\dd,+\d+\dd); 
        \draw[red] (\d-\dd,0.15) --+(-0.15,0);
        \draw[red] (1-\d+\dd,0.15) --+(0.15,0);
        \draw[red] (\d-\dd,-0.15) --+(-0.15,0);
        \draw[red] (1-\d+\dd,-0.15) --+(0.15,0);
      \end{scope}
    \end{tikzpicture} \ =
    \begin{tikzpicture}[blue,font=\scriptsize]
      \pic[rotate=90] at (0,0.5+\d) {halfstate};
      \pic[rotate=-90] at (0,-0.5-\d) {halfstate};
      \pic[anchor = south west, pic text =$Y$] at (0,0.8) {boundarympo};
      \pic[anchor = north west, pic text =$X$] at (0,-0.8) {boundarympo};
      \pic[anchor = south west, pic text =$X$] at (0,1.3) {boundarympo};
      \pic[anchor = north west, pic text =$Y$] at (0,-1.3) {boundarympo};
    \end{tikzpicture} \ .    
  \end{equation}  
 Contracting two copies of \cref{eq:schmidt_then_inverse}, the middle operator is $O O^{-1} = \id$, so
  \begin{equation}
    \begin{tikzpicture}
      \foreach \x in {0,1}{
        \pic[blue,rotate=90] at (\x+0.5,0.5+\d) {halfstate};
        \pic[blue,rotate=-90] at (\x+0.5,-0.5-\d) {halfstate};
      }
      \draw[red] (\d-\dd,-\d-\dd) rectangle (\d+\dd,+\d+\dd); 
      \draw[red,densely dashed] (\d-2*\dd,-\d-2*\dd) rectangle (\d+2*\dd,+\d+2*\dd); 
      \draw[red] (\d-2*\dd,0.2) --+(-0.2,0);
      \draw[red] (\d-\dd,-0.15) --+(-0.15-\dd,0);
      \draw[red] (\d-\dd,0.15) --+(-0.15-\dd,0);
      \draw[red] (\d-2*\dd,-0.2) --+(-0.2,0);
      \draw[red] (2-\d-\dd,-\d-\dd) rectangle (2-\d+\dd,+\d+\dd); 
      \draw[red,densely dashed] (2-\d-2*\dd,-\d-2*\dd) rectangle (2-\d+2*\dd,+\d+2*\dd); 
      \draw[red] (2-\d+\dd,0.15) --+(0.15+\dd,0);
      \draw[red] (2-\d+\dd,-0.15) --+(0.15+\dd,0);
      \draw[red] (2-\d+2*\dd,0.2) --+(0.2,0);
      \draw[red] (2-\d+2*\dd,-0.2) --+(0.2,0);
    \end{tikzpicture} \
    =
    \begin{tikzpicture}[blue,font=\scriptsize]
      \foreach \x in {0,1}{
        \pic[rotate=90] at (\x,0.5+\d) {halfstate};
        \pic[rotate=-90] at (\x,-0.5-\d) {halfstate};
        \pic[anchor = south west, pic text =$Y$] at (\x,0.8) {boundarympo};
        \pic[anchor = north west, pic text =$X$] at (\x,-0.8) {boundarympo};
        \pic[anchor = south west, pic text =$X$] at (\x,1.3) {boundarympo};
        \pic[anchor = north west, pic text =$Y$] at (\x,-1.3) {boundarympo};
      }
    \end{tikzpicture} \ .     
  \end{equation}
  Notice that the l.h.s. is a product w.r.t. the vertical cut, whereas the r.h.s. is product w.r.t. the horizontal cut. Therefore both sides have to be product w.r.t. both vertical and horizontal cuts. Note that then $V_n(X)$ and $V_n(Y)$ satisfy the conditions of \cref{thm:mpo_layers} and thus are products of invertible two-body operators in the sense of \cref{eq:mpo2layer1,eq:mpo2layer2}. Similarly, both terms on the l.h.s. factorize w.r.t. the horizontal cut. As the  one-body reduced densities of $\ket{\phi_A}$ are full rank, the product of the Schmidt vectors of $O$ and $O^{-1}$ factorize:
  
  \begin{equation}\label{eq:Oinv12}
    \left(O^{-1}\right)^{(13)}_{kl}  O^{(13)}_{ij} = A^{(1)}_{ik} \otimes A^{(3)}_{jl}.
  \end{equation}
  The same holds for the Schmidt vectors of all neighboring bipartition in any order. Similarly, the equation holds for the bipartition $(13)-(24)$ and also for the reordering of $O$ and $O^{-1}$. \cref{eq:Oinv12} can be pictorially represented as
  \begin{equation}\label{eq:factorize-1}
  \begin{tikzpicture}
  \def \c {purple}
  \draw[red, thick] (0.8,-0.4) rectangle (1.2,0.4);
  \draw[red, thick] (1.2,0.2)--++(0.3,0);
  \draw[red, thick] (1.2,-0.2)--++(0.3,0);
  \draw[\c, thick,densely dashed] (0.7,-0.5) rectangle (1.3,0.5);
  \draw[\c, thick] (1.3,0.3)--++(0.2,0);
  \draw[\c, thick] (1.3,-0.3)--++(0.2,0);
  \end{tikzpicture}
  =
  \begin{tikzpicture}
  \draw[thick,fill=white,rounded corners=3pt] (0.4,0.1) rectangle (0.6,0.6);
  \draw[thick,fill=white,rounded corners=3pt] (0.4,-0.1) rectangle (0.6,-0.6);
  \draw[thick] (0.6,0.2)--++(0.2,0);
  \draw[thick] (0.6,0.5)--++(0.2,0);
  \draw[thick] (0.6,-0.2)--++(0.2,0);
  \draw[thick] (0.6,-0.5)--++(0.2,0);
  \end{tikzpicture}\ .
  \end{equation}
  Consider  the operator
  \begin{equation}
    Z = \left(O^{(13)}_{j_1 j_2} \otimes O^{(24)}_{j_3 j_4} \right) O^{-1} \left(O^{(12)}_{i_1 i_2} \otimes O^{(34)}_{i_3 i_4}\right) = 
    \begin{tikzpicture}[baseline=0.4cm]
      \draw[thick, purple,densely dashed] (-0.3,-0.3) rectangle (1.3,1.3);
      \foreach \x/\y  in {0.2/0.15,0.8/-0.15}{
        \foreach \z in {0,1}{
          \draw[thick, red] (\z,\x)--+(0,\y); 
          \draw[thick, red] (\x,\z)--+(\y,0); 
        } 
      }
      \draw[thick, red] (-0.2,-0.2) rectangle (1.2,0.2);
      \draw[thick, red] (-0.2,0.8) rectangle (1.2,1.2);
      \draw[thick, red] (-0.4,-0.4) rectangle (0.2,1.4);
      \draw[thick, red] (0.8,-0.4) rectangle (1.4,1.4);
    \end{tikzpicture}
    \ .
  \end{equation}
  Note that $Z$ factorizes w.r.t. the bipartition $(13)-(24)$: to see this, decompose $O^{-1}$ w.r.t. the bipartition $(12)-(34)$. Then
  \begin{equation}
    O^{-1} \left(O^{(12)}_{i_1 i_2} \otimes O^{(34)}_{i_3 i_4}\right) =
    \begin{tikzpicture}[rotate=90]
      \foreach \i in {0,180}{
        \begin{scope}[rotate around={\i:(1.5,0)}]
          \def \c {purple}
          \draw[red, thick] (0.8,-0.4) rectangle (1.2,0.4);
          \draw[red, thick] (1.2,0.2)--++(0.2,0);
          \draw[red, thick] (1.2,-0.2)--++(0.2,0);
          \draw[\c, thick,densely dashed] (0.7,-0.5) rectangle (1.3,0.5);
          \draw[\c, thick] (1.3,0.3)--++(0.2,0);
          \draw[\c, thick] (1.3,-0.3)--++(0.2,0);  
        \end{scope}
      }
    \end{tikzpicture}
    =
    \begin{tikzpicture}[rotate=90]
      \foreach \i in {0,180}{
        \begin{scope}[rotate around={\i:(0.8,0)}]
          \draw[thick,fill=white,rounded corners=3pt] (0.4,0.1) rectangle (0.6,0.6);
          \draw[thick,fill=white,rounded corners=3pt] (0.4,-0.1) rectangle (0.6,-0.6);
          \draw[thick] (0.6,0.2)--++(0.1,0);
          \draw[thick] (0.6,0.5)--++(0.2,0);
          \draw[thick] (0.6,-0.2)--++(0.1,0);
          \draw[thick] (0.6,-0.5)--++(0.2,0);
        \end{scope}
      }
    \end{tikzpicture}\ ,
  \end{equation}  
  therefore it factorizes w.r.t. the bipartition $(13)-(24)$, and so does $Z$. Similarly, $Z$ also factorizes w.r.t. the bipartition $(12)-(34)$. Therefore, $Z$ is a four-partite product, 
  \begin{equation}\label{eq:O_mpo_decomp}
  Z = 
  \begin{tikzpicture}[baseline=0.4cm]
  \draw[thick, purple,densely dashed] (-0.3,-0.3) rectangle (1.3,1.3);
  \foreach \x/\y  in {0.2/0.15,0.8/-0.15}{
    \foreach \z in {0,1}{
      \draw[thick, red] (\z,\x)--+(0,\y); 
      \draw[thick, red] (\x,\z)--+(\y,0); 
    } 
  }
  \draw[thick, red] (-0.2,-0.2) rectangle (1.2,0.2);
  \draw[thick, red] (-0.2,0.8) rectangle (1.2,1.2);
  \draw[thick, red] (-0.4,-0.4) rectangle (0.2,1.4);
  \draw[thick, red] (0.8,-0.4) rectangle (1.4,1.4);
  \end{tikzpicture}
  =
  \begin{tikzpicture}[baseline=0.4cm]
  \draw[thick] (0,0)--(0,0.3);
  \draw[thick] (0,0)--(0.3,0);
  \draw[thick] (0,1)--(0.3,1);
  \draw[thick] (0,1)--(0,0.7);
  \draw[thick] (1,1)--(0.7,1);
  \draw[thick] (1,1)--(1,0.7);
  \draw[thick] (1,0)--(0.7,0);
  \draw[thick] (1,0)--(1,0.3);
  \foreach \x in {(0,0),(0,1),(1,0),(1,1)}{
    \draw[thick,fill=white] \x circle (0.1);
  }
  \end{tikzpicture}
  \end{equation}
  As contracting the open indices of $Z$ gives back the operator $O O^{-1} O = O$, and as $Z$ has a tensor product structure, this construction gives rise to an MPO description of $O$.
  Similarly, contracting only the vertical (horizontal) indices the lower (upper) two layers gives $O^{-1}O =\id$ ($OO^{-1} =\id$) on the lower (upper) two layers, and thus we obtain a minimal rank decomposition of $O$ in the horizontal (vertical) cut. As the Schmidt vectors are linearly independent, the MPO tensors become injective after blocking two tensors. 

  The above construction can be repeated for $O^{-1}$. This leads to an MPO decomposition of $O^{-1}$.
  
  These two decompositions satisfy the conditions of \cref{thm:mpo_layers}: the MPOs become injective after blocking two tensors, moreover, the product of two neighboring tensors of $O$ and $O^{-1}$ factorizes. Therefore, $O$ (and $O^{-1}$) is a product of invertible two-body operators.    
%
%
\end{proof}

The above form provides an equivalent characterization of when two semi-injective PEPS are equal for all system sizes. Before stating the theorem, we introduce two swap operators on four particles. The horizontal swap, $H_A$, exchanges the virtual particles of $\ket{\phi_A}$ in the horizontal direction:
\begin{equation}
  \begin{tikzpicture}
    \pic at (0,0) {state};
    \node at (0,0) {4};
    \node at (0,1) {1};
    \node at (1,1) {2};
    \node at (1,0) {3};
  \end{tikzpicture} 
  \rightarrow 
  \begin{tikzpicture}
    \pic at (0,0) {state};
    \node at (0,0) {3};
    \node at (0,1) {2};
    \node at (1,1) {1};
    \node at (1,0) {4};
  \end{tikzpicture}  . 
\end{equation}
 The vertical swap, $V_A$, reflects the particles of $\ket{\phi_A}$ in the vertical direction:
 \begin{equation}
   \begin{tikzpicture}
     \pic at (0,0) {state};
     \node at (0,0) {4};
     \node at (0,1) {1};
     \node at (1,1) {2};
     \node at (1,0) {3};
   \end{tikzpicture} 
   \rightarrow 
   \begin{tikzpicture}
     \pic at (0,0) {state};
     \node at (0,0) {1};
     \node at (0,1) {4};
     \node at (1,1) {3};
     \node at (1,0) {2};
   \end{tikzpicture}.   
 \end{equation}
We denote the product of $H_A$ and $V_A$ as $S_A$: $S_A=H_A V_A = V_A H_A$. 
Define $H_B, V_B $ and $S_B$ similarly for $\ket{\phi_B}$. 
Note that $H_A$ and $H_B$ are different in general as the Hilbert spaces of the virtual particles might differ.
\begin{theorem}\label{thm:canonicalform}
  Two semi-injective PEPS are equal (\cref{eq:setup} holds) if and only if the following conditions are satisfied:
  \begin{itemize}
    \item The operator $O$ factorizes into two-body operators as
    \begin{equation}
      \begin{tikzpicture}[baseline=-0.1cm]
      \draw[red] (0,0) circle (0.5);
      \end{tikzpicture} =
      \begin{tikzpicture}[baseline=-0.1cm,rotate=90]
      \draw[red] (-0.5,-0.5) rectangle (-0.2,0.5);
      \draw[red] (0.2,-0.5) rectangle (0.5,0.5);
      \draw[red] (-0.6,-0.6) rectangle (0.6,-0.1);
      \draw[red] (-0.6,0.1) rectangle (0.6,0.6);      
      \end{tikzpicture}       
    \end{equation}
    \item	The Schmidt vectors of the four-partite states satisfy:
    \begin{align} \label{eq:swap_schmidt_vert}
      \begin{tikzpicture}[baseline=-0.1cm]
      \draw[thick,blue] (0,0.7)--(0,0.2)--(1,0.2)--(1,0.7); 
      \draw[thick,blue] (0,-0.7)--(0,-0.2)--(1,-0.2)--(1,-0.7); 
      \foreach \x in {(0,0.2),(0,-0.2),(1,0.2),(1,-0.2)} \filldraw \x circle (0.1);
      \draw[rounded corners=3pt,blue] (-0.2,0.7) rectangle (1.2,0.9);
      \draw[rounded corners=3pt,blue] (-0.2,-0.7) rectangle (1.2,-0.9);  
      \draw (0.5,0.9)--++(0,0.2);
      \draw (0.5,-0.9)--++(0,-0.2);
      \draw[red, thick] (0.5,0) ellipse (0.8cm and 0.5cm);
      \end{tikzpicture}
      &=
      \begin{tikzpicture}[baseline=-0.1cm]
      \def\c{green}
      \draw[thick,\c] (0,0.7)--(0,0.2)--(1,0.2)--(1,0.7); 
      \draw[thick,\c] (0,-0.7)--(0,-0.2)--(1,-0.2)--(1,-0.7); 
      \foreach \x in {(0,0.2),(0,-0.2),(1,0.2),(1,-0.2)} \filldraw \x circle (0.1);
      \draw[rounded corners=3pt,\c] (-0.2,0.7) rectangle (1.2,0.9);
      \draw[rounded corners=3pt,\c] (-0.2,-0.7) rectangle (1.2,-0.9);  
      \draw[thick,\c] (0.5,0.9)--++(0,0.2);      
      \draw[thick,\c] (0.5,-0.9)--++(0,-0.2);
      \end{tikzpicture}\\ \label{eq:swap_schmidt_hori}
      \begin{tikzpicture}[rotate=90,baseline=0.4cm]
      \draw[thick,blue] (0,0.7)--(0,0.2)--(1,0.2)--(1,0.7); 
      \draw[thick,blue] (0,-0.7)--(0,-0.2)--(1,-0.2)--(1,-0.7); 
      \foreach \x in {(0,0.2),(0,-0.2),(1,0.2),(1,-0.2)} \filldraw \x circle (0.1);
      \draw[rounded corners=3pt,blue] (-0.2,0.7) rectangle (1.2,0.9);
      \draw[rounded corners=3pt,blue] (-0.2,-0.7) rectangle (1.2,-0.9);  
      \draw (0.5,0.9)--++(0,0.2);
      \draw (0.5,-0.9)--++(0,-0.2);
      \draw[red, thick] (0.5,0) ellipse (0.8cm and 0.5cm);
      \end{tikzpicture}
      &=
      \begin{tikzpicture}[rotate=90,baseline=0.4cm]
      \def\c{green}
      \draw[thick,\c] (0,0.7)--(0,0.2)--(1,0.2)--(1,0.7); 
      \draw[thick,\c] (0,-0.7)--(0,-0.2)--(1,-0.2)--(1,-0.7); 
      \foreach \x in {(0,0.2),(0,-0.2),(1,0.2),(1,-0.2)} \filldraw \x circle (0.1);
      \draw[rounded corners=3pt,\c] (-0.2,0.7) rectangle (1.2,0.9);
      \draw[rounded corners=3pt,\c] (-0.2,-0.7) rectangle (1.2,-0.9);  
      \draw[thick,\c] (0.5,0.9)--++(0,0.2);      
      \draw[thick,\c] (0.5,-0.9)--++(0,-0.2);
      \end{tikzpicture}\ , 
    \end{align}
    where the horizontal ellipse denotes $H_BOH_A$, and the vertical ellipse denotes $V_BOV_A$.
  \end{itemize}
\end{theorem}

Note that the last two conditions are equivalent to the property that the two states are equal on an $n\times 1$ and a $1\times n$ torus for all $n$, therefore they are easily checkable.

\begin{proof}
  The necessity of these conditions is clear from above. We now prove the sufficiency. Let
  \begin{align}
    O^{-1} &=
    \begin{tikzpicture}[baseline=-0.1cm]
    \draw[red,dashed] (0,0) circle (0.5);
    \end{tikzpicture}\ ,\\
    H_BOH_A &=
    \begin{tikzpicture}[baseline=-0.1cm]
    \draw[red] (0,0) ellipse (0.7cm and 0.5cm);
    \end{tikzpicture}\ ,\\
    V_AO^{-1}V_B &=
    \begin{tikzpicture}[baseline=-0.1cm]
    \draw[red,dashed] (0,0) ellipse (0.5cm and 0.7cm);
    \end{tikzpicture}\ ,\\
    S_BOS_A &=
    \begin{tikzpicture}[baseline=-0.1cm]
    \draw[purple,rounded corners=2mm] (-0.6,-0.6) rectangle (0.6,0.6); 
    \end{tikzpicture}  \   .    
  \end{align}  
  Due to the two layer structure of $O$ and $O^{-1}$ (\cref{eq:op_twolayer}), the following operator is a product in the horizontal cut:
  \begin{equation}
    \begin{tikzpicture}[baseline=-0.1cm]
    \clip (0,-1) rectangle (3,1);
    \foreach \x in {0,1.5,3}{
      \draw[red] (\x,0) ellipse (0.7cm and 0.5cm);
      \draw[red, dashed] (\x+0.75,0) circle (0.5);
    }
    \end{tikzpicture} = A\otimes B\ ,
  \end{equation}
  where $H_BOH_A$ is the lower layer.
  The vertical swap of the previous operator is 
  \begin{equation}
    \begin{tikzpicture}[baseline=-0.1cm]
    \clip (0,-1) rectangle (3,1);
    \foreach \x in {0,1.5,3}{
      \draw[purple,rounded corners=2mm] (\x-0.6,-0.6) rectangle (\x+0.6,0.6); 
      \draw[red,dashed] (\x+0.75,0) ellipse (0.5cm and 0.7cm);
    }
    \end{tikzpicture} = B\otimes A\ ,
  \end{equation}  
  where $S_BOS_A$ is the lower layer.
  Consider now these operators acting on the semi-injective PEPS defined by $\ket{\phi_A}$  and $\id$:
  
  \begin{equation}
    \begin{tikzpicture}
    \clip (0.5,0.5) rectangle (4.4,4.4);
    \foreach \x in {0,1.3,2.6,3.9}
    \foreach \y in {0,1.3,2.6,3.9}{
      \draw[blue] (\x,\y) rectangle (\x+1,\y+1);
      \draw[red] (\x+0.5,\y+1.15) ellipse (0.62cm and 0.47cm);
      \draw[red, dashed] (\x+1.15,\y+1.15) circle (0.5cm); 
      \foreach \u in {0,1} \foreach \v in {0,1}
      \filldraw (\x+\u,\y+\v) circle (0.05);
    }
    \end{tikzpicture} = 
    \begin{tikzpicture}
    \clip (0.5,0.5) rectangle (4.4,4.4);
    \foreach \x in {0,1.3,2.6,3.9}
    \foreach \y in {0,1.3,2.6,3.9}{
      \draw[blue] (\x,\y) rectangle (\x+1,\y+1);
      \draw[purple,rounded corners=1.5mm] (\x-0.11,\y-0.11) rectangle (\x+1.11,\y+1.11); 
      \draw[red,dashed] (\x+1.15,\y+0.5) ellipse (0.47cm and 0.62cm);
      \foreach \u in {0,1} \foreach \v in {0,1}
      \filldraw (\x+\u,\y+\v) circle (0.05);
    }
    \end{tikzpicture}
  \end{equation}
  From \cref{eq:swap_schmidt_vert} and \eqref{eq:swap_schmidt_hori}, the action of the lower layers on each side is to change $\ket{\phi_A}$ to $\ket{\phi_B}$. Therefore,
  \begin{equation}
    \begin{tikzpicture}
      \truncstateswo[green,{red,dashed}] (0,0) (3,3);
    \end{tikzpicture} = 
    \begin{tikzpicture}
      \truncstates[green] (0,0) (3,3);
      \clip (0.5,0.5) rectangle (3.5,3.5);
      \foreach \x in {1,2,3} \foreach \y in {0,1,2,3}
        \draw[red,dashed] (\x,\y+0.5) ellipse (0.36cm and 0.47cm);
    \end{tikzpicture}\ .
  \end{equation}
  Using once more \cref{eq:swap_schmidt_hori}, the r.h.s.\ is a tensor product of $\phi_A$ at every position:
  \begin{equation}
    \begin{tikzpicture}
      \truncstateswo[green,{red,dashed}] (0,0) (3,3);
    \end{tikzpicture} = 
    \begin{tikzpicture}
      \truncstates[blue] (0,0) (3,3)
    \end{tikzpicture}\ .
  \end{equation}  
  applying $O$ on both sides on each site, we see that \cref{eq:setup} holds. 
\end{proof}

As a simple application, one can derive the canonical form of injective PEPS \cite{Perez-Garcia2010}.
\begin{corollary}
  Two injective PEPS generate the same state if and only if they are related by a product gauge transformation.
\end{corollary}

\begin{proof}
  The conditions of \cref{thm:canonicalform} that Schmidt vectors map to Schmidt vectors read as
  \begin{equation}
    \begin{tikzpicture}[baseline=-0.1cm]
    \draw[fill=black] (-0.2,0.2) circle (0.05cm);
    \draw[fill=black] (-0.2,-0.2) circle (0.05cm);
    \draw[fill=black] (0,0) circle (0.05cm);
    \draw[fill=black] (1,0) circle (0.05cm);
    \draw (0,0)--(1,0);
    \draw (-0.2,0.2) -- (-0.2,0.7);
    \draw (-0.2,-0.2) -- (-0.2,-0.7);
    \draw[red] (0.35,0) ellipse (0.8cm and 0.5cm);      
    \end{tikzpicture}
    =
    \begin{tikzpicture}[baseline=-0.1cm]
    \draw[fill=black] (-0.2,0.2) circle (0.05cm);
    \draw[fill=black] (-0.2,-0.2) circle (0.05cm);
    \draw[fill=black] (-0,0) circle (0.05cm);
    \draw[fill=black] (1,0) circle (0.05cm);
    \draw (0,0)--(1,0);
    \draw (-0.2,0.2) -- (-0.2,1);
    \draw (-0.2,-0.2) -- (-0.2,-1);
    \draw[fill=white] (-0.2,0.6) circle (0.2cm);
    \draw[fill=white] (-0.2,-0.6) circle (0.2cm);
    \node[anchor=west] at (0,0.6) {$X$};
    \node[anchor=west] at (0,-0.6) {$X^{-1}$};
    \end{tikzpicture}\ ,
  \end{equation}  
  therefore the operator $O$ is a product on the two leftmost particles (and on one particle it is the inverse of the other). Similarly the other condition implies that the operator is a product on the two rightmost particles. Therefore $O$ has a product structure in the desired form.
\end{proof}

We now show that if the span of the Schmidt vectors of both states w.r.t. both the vertical and horizontal cut contain product states, then  $\ket{\phi_A}$ and $\ket{\phi_B}$ are SLOCC\cite{Bennett1999}-equivalent, that is, there are invertible operators $O_1, O_2, O_3, O_4$ acting on the virtual particles  such that $O_1 \otimes O_2 \otimes O_3 \otimes O_4 \ket{\phi_A} = \ket{\phi_B}$. Pictorially,
\begin{equation}
  \begin{tikzpicture}
    \pic[blue] at (0,0) {state};
    \foreach \x/\y in {\d/\d,1-\d/\d,1-\d/1-\d,\d/1-\d}
      \draw[red] (\x,\y) circle (0.15cm);
  \end{tikzpicture} =
  \begin{tikzpicture}
    \pic[green] {state};
  \end{tikzpicture}\ .
\end{equation}

Note that there are examples for states that don't have product states in the span of their Schmidt vectors, but they generate the same state and are not SLOCC equivalent. 
For example consider 
\begin{equation}
  \ket{\phi_A} =  
  \begin{tikzpicture}
    \pic at (0,0) {state};
    \draw (\d,1+\d)--(1-\d,1+\d);
    \draw[fill=black] (\d,1+\d) circle (\rad);
    \draw[fill=black] (1-\d,1+\d) circle (\rad);
  \end{tikzpicture}\ , \qquad
  \ket{\phi_B} = 
  \begin{tikzpicture}
    \pic at (0,0) {state};
    \draw (\d,-\d)--(1-\d,-\d);
    \draw[fill=black] (\d,-\d) circle (\rad);
    \draw[fill=black] (1-\d,-\d) circle (\rad);
  \end{tikzpicture}\ ,
\end{equation}
then the semi-injective PEPS defined by $\ket{\phi_A}$ and $\id$ and $\ket{\phi_B}$ and $\id$  (more precisely the isomorphism that rearranges the tensor product to the right order) are the same on every torus, yet these states are not SLOCC equivalent.

\begin{theorem}\label{thm:slocc}
  If the span of the Schmidt vectors of both four-partite states in \cref{eq:setup} contains a product state  for both the vertical and horizontal cut on both sides, then the two four-partite states are SLOCC equivalent.
\end{theorem}
 
\begin{proof}
  By \Cref{thm:canonicalform}, \cref{eq:setup} implies
  \begin{equation}
    \begin{tikzpicture}
    \truncstates[blue] (0,0) (4,3)
    \foreach \x in {1,...,4}{
      \foreach \y in {1,...,3}{
        \draw[thin, draw=red] (\x-\d-\dd,\y-\d-\dd) rectangle (\x-\d+\dd,\y+\d+\dd);
        \draw[thin, draw=red] (\x+\d-\dd,\y-\d-\dd) rectangle (\x+\d+\dd,\y+\d+\dd);
        \draw[thin, draw=red] (\x-\d-\ddd,\y-\d-\ddd) rectangle (\x+\d+\ddd,\y-\d+\ddd);
        \draw[thin, draw=red] (\x-\d-\ddd,\y+\d-\ddd) rectangle (\x+\d+\ddd,\y+\d+\ddd);
      }
    }
    \end{tikzpicture} = \ 
    \begin{tikzpicture}
      \truncstates[green] (0,0) (4,3)
    \end{tikzpicture}
  \end{equation}
  Inverting the upper layer, we get 
  \begin{equation}
    \begin{tikzpicture}
    \truncstates[blue] (0,0) (4,3);
    \foreach \x in {1,...,4}{
      \foreach \y in {1,...,3}{
        \draw[thin, draw=red] (\x-\d-\dd,\y-\d-\dd) rectangle (\x-\d+\dd,\y+\d+\dd);
        \draw[thin, draw=red] (\x+\d-\dd,\y-\d-\dd) rectangle (\x+\d+\dd,\y+\d+\dd);
      }
    }
    \end{tikzpicture} = \ 
    \begin{tikzpicture}
    \truncstates[green] (0,0) (4,3);
    \foreach \x in {1,...,4}{
      \foreach \y in {1,...,3}{
        \draw[thin, draw=red, densely dashed] (\x-\d-\ddd,\y-\d-\ddd) rectangle (\x+\d+\ddd,\y-\d+\ddd);
        \draw[thin, draw=red, densely dashed] (\x-\d-\ddd,\y+\d-\ddd) rectangle (\x+\d+\ddd,\y+\d+\ddd);
      }
    }
    \end{tikzpicture}
  \end{equation}
  The l.h.s. is product in the vertical direction, the r.h.s. in the horizontal direction. Therefore the two sides describe a state that factorizes in both directions. Let  $\ket{\xi}$ be this state and denote this state with a purple square. Then, 
  \begin{equation}
    \begin{tikzpicture}
      \truncstates[blue] (0,0) (4,3);
      \foreach \x in {1,...,4}{
        \foreach \y in {1,...,3}{
          \draw[thin, draw=red] (\x-\d-\dd,\y-\d-\dd) rectangle (\x-\d+\dd,\y+\d+\dd);
          \draw[thin, draw=red] (\x+\d-\dd,\y-\d-\dd) rectangle (\x+\d+\dd,\y+\d+\dd);
        }
      }
    \end{tikzpicture} = \ 
    \begin{tikzpicture}
      \truncstates[green] (0,0) (4,3);
      \foreach \x in {1,...,4}{
        \foreach \y in {1,...,3}{
          \draw[thin, draw=red,densely dashed] (\x-\d-\ddd,\y-\d-\ddd) rectangle (\x+\d+\ddd,\y-\d+\ddd);
          \draw[thin, draw=red,densely dashed] (\x-\d-\ddd,\y+\d-\ddd) rectangle (\x+\d+\ddd,\y+\d+\ddd);
        }
      }
    \end{tikzpicture}= \ 
    \begin{tikzpicture}
      \truncstates[purple] (0,0) (4,3)
    \end{tikzpicture}
  \end{equation}
  Equivalently, for the Schmidt vectors we get 
  
  \begin{align} \label{eq:schmidt_vert}
    \begin{tikzpicture}[baseline=-0.1cm]  
    \draw[thick,blue] (0,0.7)--(0,0.2)--(1,0.2)--(1,0.7); 
    \draw[thick,blue] (0,-0.7)--(0,-0.2)--(1,-0.2)--(1,-0.7); 
    \foreach \x in {(0,0.2),(0,-0.2),(1,0.2),(1,-0.2)} \filldraw \x circle (0.1);
    \draw[rounded corners=3pt,blue] (-0.2,0.7) rectangle (1.2,0.9);
    \draw[rounded corners=3pt,blue] (-0.2,-0.7) rectangle (1.2,-0.9);  
    \draw (0.5,0.9)--++(0,0.2);
    \draw (0.5,-0.9)--++(0,-0.2);
    \draw[red, thick] (-0.2,-0.4) rectangle (0.2,0.4);
    \draw[red, thick] (0.8,-0.4) rectangle (1.2,0.4);
    \end{tikzpicture}
    &=
    \begin{tikzpicture}[baseline=-0.1cm]
    \def\c{purple}
    \draw[thick,\c] (0,0.7)--(0,0.2)--(1,0.2)--(1,0.7); 
    \draw[thick,\c] (0,-0.7)--(0,-0.2)--(1,-0.2)--(1,-0.7); 
    \foreach \x in {(0,0.2),(0,-0.2),(1,0.2),(1,-0.2)} \filldraw \x circle (0.1);
    \draw[rounded corners=3pt,\c] (-0.2,0.7) rectangle (1.2,0.9);
    \draw[rounded corners=3pt,\c] (-0.2,-0.7) rectangle (1.2,-0.9);  
    \draw[thick,\c] (0.5,0.9)--++(0,0.2);      
    \draw[thick,\c] (0.5,-0.9)--++(0,-0.2);
    \end{tikzpicture}\\ 
    \begin{tikzpicture}[rotate=90,baseline=0.4cm]
    \draw[thick,green] (0,0.7)--(0,0.2)--(1,0.2)--(1,0.7); 
    \draw[thick,green] (0,-0.7)--(0,-0.2)--(1,-0.2)--(1,-0.7); 
    \foreach \x in {(0,0.2),(0,-0.2),(1,0.2),(1,-0.2)} \filldraw \x circle (0.1);
    \draw[rounded corners=3pt,green] (-0.2,0.7) rectangle (1.2,0.9);
    \draw[rounded corners=3pt,green] (-0.2,-0.7) rectangle (1.2,-0.9);  
    \draw (0.5,0.9)--++(0,0.2);
    \draw (0.5,-0.9)--++(0,-0.2);
    \draw[red, thick,densely dashed] (-0.2,-0.4) rectangle (0.2,0.4);
    \draw[red, thick,densely dashed] (0.8,-0.4) rectangle (1.2,0.4);
    \end{tikzpicture}
    &=
    \begin{tikzpicture}[rotate=90,baseline=0.4cm]
    \def\c{purple}
    \draw[thick,\c] (0,0.7)--(0,0.2)--(1,0.2)--(1,0.7); 
    \draw[thick,\c] (0,-0.7)--(0,-0.2)--(1,-0.2)--(1,-0.7); 
    \foreach \x in {(0,0.2),(0,-0.2),(1,0.2),(1,-0.2)} \filldraw \x circle (0.1);
    \draw[rounded corners=3pt,\c] (-0.2,0.7) rectangle (1.2,0.9);
    \draw[rounded corners=3pt,\c] (-0.2,-0.7) rectangle (1.2,-0.9);  
    \draw[thick,\c] (0.5,0.9)--++(0,0.2);      
    \draw[thick,\c] (0.5,-0.9)--++(0,-0.2);
    \end{tikzpicture}
  \end{align}
  If the span of the Schmidt vectors on the l.h.s. contains a product vector, the same is true for the Schmidt vectors on the r.h.s. Therefore, choosing a product Schmidt vector on the bottom in \cref{eq:schmidt_vert} and applying a product linear functional, we get that for some not necessarily invertible operators,
  \begin{equation} \label{eq:slocc1}
    \begin{tikzpicture}[baseline=0.5cm]
    \draw[thick,blue] (0,0.7)--(0,0.2)--(1,0.2)--(1,0.7); 
    \foreach \x in {(0,0.2),(1,0.2)} \filldraw \x circle (0.1);
    \draw[rounded corners=3pt,blue] (-0.2,0.7) rectangle (1.2,0.9);
    \draw (0.5,0.9)--++(0,0.2);
    \draw[red, thick] (-0.2,0) rectangle (0.2,0.4);
    \draw[red, thick] (0.8,0) rectangle (1.2,0.4);
    \end{tikzpicture}
    =
    \begin{tikzpicture}[baseline = 0.5cm]
    \def\c{purple}
    \draw[thick,\c] (0,0.7)--(0,0.2)--(1,0.2)--(1,0.7); 
    \foreach \x in {(0,0.2),(1,0.2)} \filldraw \x circle (0.1);
    \draw[rounded corners=3pt,\c] (-0.2,0.7) rectangle (1.2,0.9);
    \draw[thick,\c] (0.5,0.9)--++(0,0.2);      
    \end{tikzpicture}\ .
  \end{equation}  
  A similar equation also holds for the lower part, as well as for both sides of $\ket{\phi_B}$. Inverting the operators appearing in \cref{eq:schmidt_vert}, by the same argument, we obtain the inverse relation 
  \begin{equation} \label{eq:slocc2}
    \begin{tikzpicture}[baseline=0.5cm]
    \draw[thick,blue] (0,0.7)--(0,0.2)--(1,0.2)--(1,0.7); 
    \foreach \x in {(0,0.2),(1,0.2)} \filldraw \x circle (0.1);
    \draw[rounded corners=3pt,blue] (-0.2,0.7) rectangle (1.2,0.9);
    \draw (0.5,0.9)--++(0,0.2);
    \end{tikzpicture}
    =
    \begin{tikzpicture}[baseline = 0.5cm]
    \def\c{purple}
    \draw[thick,\c] (0,0.7)--(0,0.2)--(1,0.2)--(1,0.7); 
    \foreach \x in {(0,0.2),(1,0.2)} \filldraw \x circle (0.1);
    \draw[rounded corners=3pt,\c] (-0.2,0.7) rectangle (1.2,0.9);
    \draw[thick,\c] (0.5,0.9)--++(0,0.2);      
    \draw[red, dashed, thick] (-0.2,0) rectangle (0.2,0.4);
    \draw[red, dashed, thick] (0.8,0) rectangle (1.2,0.4);
    \end{tikzpicture}\ .
  \end{equation}  
  and similarly along all other cuts. 
  \cref{eq:slocc1} and \eqref{eq:slocc2} ensure that the one particle operators can be chosen invertible, thus $\ket{\phi_A}$ and $\ket{\xi}$ are SLOCC equivalent. Similarly, $\ket{\phi_B}$ and $\ket{\xi}$ are SLOCC equivalent. Therefore $\ket{\phi_A}$ and $\ket{\phi_B}$ are SLOCC equivalent.
\end{proof}

\begin{corollary}\label{cor:qubit_slocc}
  If two semi-injective PEPS, defined by qubit four-partite states with genuine four-partite entanglement, are equal, then the four-partite states are SLOCC equivalent.
\end{corollary}

\begin{proof}
  Notice that if the four-partite states are entangled for both the  vertical and horizontal cut, then the span of the Schmidt vectors is at least two-dimensional. As any two-dimensional subspace contains a product vector, the previous theorem applies.
\end{proof}

Based on \cref{cor:qubit_slocc}, we provide a full classification of semi-injective PEPS defined with four-partite qubit states in \cref{app:qubitclassify}. 

\section{SPT phases}\label{sec:spt}

In this section we show how the third cohomology labeling of the SPT phases\cite{Chen2011a,Chen2013} extends to semi-injective PEPS. First, we show how to assign an element from the third cohomology group $H^3(G,\mathbb{C}^*)$ to a (projective) MPO representation of $G$. Here, and in the following, the action of $G$ on  $\mathbb{C}^*$ is trivial. Second,  given a group of on-site symmetries of an semi-injective PEPS, there are  three MPO representations associated to it: the boundary along the vertical cut, the boundary along the horizontal cut and finally the symmetry operators themselves. We show that the associated third cohomology labels coincide. The importance of this statement is twofold. First, the labeling is encoded in the local operators already, thus one does not have to look at the boundary of the system to find the labeling. Second, the labeling corresponding the vertical and horizontal boundary coincides despite the model not necessarily having rotational symmetry. 

\subsection{Third cohomology labeling of MPO representations}\label{sec:cohomology_label}

Consider a group $G$ and a projective MPO representation thereof, that is,  a tensor $X_g$ that generates an MPO $V_n(X_g) $ for all $ g\in G$ such that $V_n(X_g) V_n(X_h) = \lambda_n(g,h) V_n(X_{gh})$ for all $g,h\in G$, where $\lambda_n(g,h)\in\mathbb{C}$. We will restrict ourselves to MPO representations for which $\lambda_n(g,h) = \lambda^n(g,h)$. We call such MPO projective representations \emph{one-block projective MPO representations}. In this section, we show how to assign an element from the third cohomology group $H^3(G,\mathbb{C}^*)$ to such a representation.

We first show that we can suppose w.l.o.g.\ that $X_g$ is normal. The proof is analogous to \Cref{prop:boundary_mpo}.
\begin{lemma}
  Let $g\mapsto \tilde{X}_g$ be a one-block projective MPO representation of a group $G$, that is,  $V_n(\tilde{X}_g) V_n(\tilde{X}_h) = \lambda^n(g,h) V_n(\tilde{X}_{gh})$ for some $\lambda(g,h)\in \mathbb{C}$. Then $\forall g \in G$ there is a normal tensor $X_g$ such that $V_n(\tilde{X}_g) = V_n(X_g)$.
\end{lemma} 
\begin{proof}
  First we prove that $V_n(\tilde{X}_e) =\mu^n \id$ for some $\mu\in\mathbb C$, therefore there is a normal tensor $X_e$  such that $V_n(\tilde{X}_e) =V_n(X_e)$. Then, as $V_n(\tilde{X}_g)V_n(\tilde{X}_{g^{-1}})=\lambda^n(g,g^{-1}) \mu^n \id$, we will see that $V_n(\tilde{X}_g)$ can also be described with a normal MPO. 
  
  To see that $V_n(\tilde{X}_e) =\mu^n \id$, notice that, as $g\mapsto V_n(\tilde{X}_e)$ is a representation, $V_n(\tilde{X}_e) =\mu_n \id$ and that $V_n(\tilde{X}_e) V_n(\tilde{X}_e) =\mu_n^2 \id = \mu_n \lambda^n(e,e) \id$. Therefore, $\mu_n = \lambda^n(e,e)$.
  
  Let $K$ be such that after blocking $K$ tensors, $V_n(\tilde{X}_g)$ and $V_n(\tilde{X}_{g^{-1}})$ can be decomposed into a sum of $N$ and $M$ normal MPOs, respectively. That is, $\forall n \in K\mathbb{N}$
  \begin{eqnarray}
    V_n(\tilde{X}_g) = \sum_{i=1}^N V_n(\tilde{X}^{(i)}_g), \\
    V_n(\tilde{X}_{g^{-1}}) = \sum_{i=1}^M V_n(\tilde{X}^{(i)}_{g^{-1}}).    
  \end{eqnarray}
  Then their product, $\lambda^n(g,g^{-1}) \mu^n \id$, can be decomposed into a sum of at least $MN$ not necessarily essentially different normal MPOs:
  \begin{equation}
    \lambda^n(g,g^{-1}) \mu^n \id = \sum_{i=1}^N \sum_{j=1}^M V_n(\tilde{X}^{(i)}_g) V_n(\tilde{X}^{(j)}_{g^{-1}}).
  \end{equation}
  Let $L$ be such that after blocking $L$ tensors, all of these MPOs can be decomposed into normal MPOs: $\forall n\in KL\mathbb{N}$
  \begin{equation}
    \lambda^n(g,g^{-1}) \mu^n \id = \sum_{i=1}^N \sum_{j=1}^M \sum_{k=1}^{K_{ij}}V_n(Z_g^{ijk}),
  \end{equation}   
  for some normal tensors $Z_g^{ijk}$. If $i\neq i'$ or $j\neq j'$, $Z_g^{ijk}$ and $Z_g^{i'j'k'}$ are not necessarily essentially different. Collecting the essentially different terms yields
  \begin{equation}
    \lambda^n(g,g^{-1}) \mu^n \id = \sum_{i=1}^R \sum_{j=1}^{S_i} \xi_j^n V_n(Z_g^{i}),
  \end{equation}
  where $R$ is the number of essentially different terms, $Z_g^i$ are a maximal  pairwise essentially different subset of $Z_g^{ijk}$ and $S_i$ is the multiplicity with which $Z_g^i$ appears. Note that 
  \begin{equation}
    \sum_{i=1}^R S_i = \sum_{i=1}^N \sum_{j=1}^M  K_{ij}.
  \end{equation}
  As essentially different normal MPOs become linearly independent for sufficiently large  system sizes (\cref{cor:normal_lin_indep}),  \cref{prop:sum_power} implies that there can only be one term in this decomposition, that is, $R=1$ and moreover $S_1=1$. As all $K_{ij}\geq 1$, we have $N=M=1$ and thus $V_n(\tilde{X}_g)$ is normal. Therefore, $V_n(\tilde{X}_g)$ can be described by a normal MPO tensor $X_g$.
\end{proof}

The central tool in this section is comparing normal and non-normal MPS tensors that generate the same state. 
We only state the results here, the proofs are provided in \cref{app:mps_reductions}.

\begin{restatable}{prop}{reductionexist}\label{thm:reductionexist}
  Let $A$ be a normal MPS tensor, $B$ an MPS tensor such that for some $\lambda\in \mathbb{C}$
  \begin{equation}
   V_n(B) = \lambda^n V_n(A) \quad \forall n\in \mathbb{N}.
  \end{equation}
   Then there exist matrices $V,W$ such that $VW=\id$ and $\forall n \in \mathbb{N}$ and $(i_1,i_2\dots i_n) \in \{1,2,\dots d\}^n$,
  \begin{equation}
    V B^{i_1} \dots B^{i_n} W = A^{i_1} \dots A^{i_n}
  \end{equation}
\end{restatable}


\begin{restatable}{definition}{reductiondef}
  The pair of operators $V,W$ in \cref{thm:reductionexist} is called a \emph{reduction} from $B$ to $A$.
\end{restatable}

\begin{restatable}{prop}{reductionnilpotent}\label{prop:reduction_nilpotent}
  Let $V,W$ be a reduction from $B$ to $A$. Let $N^i= B^i - W A^i V$. Then the algebra generated by $N^i$ is nilpotent.
\end{restatable}

\begin{restatable}{definition}{reductionnilpotencylength}\label{def:nilpotencylength}
  Let $V,W$ be a reduction from $B$ to $A$. Let $N^i= B^i - W A^i V$. Then the \emph{nilpotency length} of the reduction is the minimal $N_0$ such that $\forall n \geq N_0$
  \begin{equation}
    N^{i_1}\dots N^{i_n} = 0.
  \end{equation}
\end{restatable}

The main statement is that any two reductions are related:

\begin{restatable}{theorem}{reductionrelate}\label{thm:uniqueness}
  Let $V,W$ and $\tilde{V},\tilde{W}$ be two reductions from $B$ to a normal tensor $A$. Let the nilpotency length of both reductions be at most $N_0$. Then $\exists \lambda \in \mathbb{C}$ such that for any $n>2N_0$,
  \begin{align}
    V B^{i_1} B^{i_2} \dots B^{i_{n}} &= \lambda \tilde{V} B^{i_1} B^{i_2} \dots B^{i_{n}} \\  
    B^{i_1} B^{i_2} \dots B^{i_{n}} W &= \lambda^{-1} B^{i_1} B^{i_2} \dots B^{i_{n}} \tilde{W}. \label{eq:BBBW}
  \end{align}
\end{restatable}

Let us now continue how to assign an element of the third cohomology group to a one-block projective MPO representation. This discussion is essentially the same as in Ref.~\onlinecite{Chen2011a}. We include here the construction for completeness.

Let $X_{g,h} = \sum_{ijk} X^{ij}_g \otimes X^{jk}_h \otimes \ket{i}\bra{k}$ be the MPO tensor describing the product of two MPOs. As $X_{g,h}$ and $X_{gh}$ describe the same state and $X_{gh}$ is injective, $X_{g,h}$ can be reduced to $X_{gh}$ by \cref{thm:reductionexist}. Let us fix such a reduction $V(g,h), W(g,h)$ for any pair of group elements. We will assign a complex scalar to these reductions. We show that this scalar forms a three-cocycle. Different reductions then lead to different three-cocycles. We show, however, that their ratio forms a three-coboundary. Therefore, the equivalence class of the scalars is an element from the third cohomology group.

Starting from the reductions $V(g,h)$, $W(g,h)$, there are two natural ways to reduce the product of three MPOs: 
\begin{equation}
  \begin{tikzpicture}[font=\scriptsize]
    \draw (-1.5,2) -- (1.5,2);
    \pic[pic text=$X_k$] at (0,0) {mpotensor};
    \pic[pic text=$X_h$] at (0,1) {mpotensor};
    \pic[pic text=$X_g$] at (0,2) {mpotensor};
    \pic[x={(-1,0)}] at (-1.25,0.5) {reduction};
    \pic at (-1.25,0.5) {reduction};
    \pic[x={(-1,0)}] at (-2,1.5) {longreduction};
    \pic at (-2,1.5) {longreduction};
    \node at (-1.6,2.5) {$V(g,hk)$};
    \node at (-1.0,-0.5) {$V(h,k)$};
    \node at (1.6,2.5) {$W(g,hk)$};
    \node at (1.0,-0.5) {$W(h,k)$};
  \end{tikzpicture}
  =
  \begin{tikzpicture}[rotate=180,font=\scriptsize]
    \draw (-1.5,2) -- (1.5,2);
    \pic[pic text=$X_g$,rotate=180] at (0,0) {mpotensor};
    \pic[pic text=$X_h$,rotate=180] at (0,1) {mpotensor};
    \pic[pic text=$X_k$,rotate=180] at (0,2) {mpotensor};
    \pic[x={(-1,0)}] at (-1.25,0.5) {reduction};
    \pic at (-1.25,0.5) {reduction};
    \pic[x={(-1,0)}] at (-2,1.5) {longreduction};
    \pic at (-2,1.5) {longreduction};
    \node[rotate=180] at (-1.6,2.5) {$W(gh,k)$};
    \node[rotate=180] at (-1.6,-0.1) {$W(g,h)$};
    \node[rotate=180] at (1.6,2.5) {$V(gh,k)$};
    \node[rotate=180] at (1.0,-0.5) {$V(g,h)$};
\end{tikzpicture}
=
\begin{tikzpicture}[font=\scriptsize]
  \pic[pic text = $X_{ghk}$] at (0,0) {mpotensor};
\end{tikzpicture}
\end{equation}
By \cref{thm:uniqueness}, there exists a complex scalar $\lambda(g,h,k)\in\mathbb C$ such that for any sufficiently long chain,
\begin{equation}\label{eq:mpo_reduction_gh}
  \begin{tikzpicture}[font=\scriptsize]
    \draw (2,2) -- (3.5,2);
    \foreach \x in {0,1,2}{
      \pic[pic text=$X_k$] at (\x,0) {mpotensor};
      \pic[pic text=$X_h$] at (\x,1) {mpotensor};
      \pic[pic text=$X_g$] at (\x,2) {mpotensor};
    }
    \pic[x={(-1,0)}] at (-3.25,0.5) {reduction};
    \pic[x={(-1,0)}] at (-4,1.5) {longreduction};
    \node at (3.6,2.5) {$W(g,hk)$};
    \node at (3.0,-0.5) {$W(h,k)$};
  \end{tikzpicture}
  =
  \lambda(g,h,k)
  \begin{tikzpicture}[yscale=-1,font=\scriptsize]
    \draw (2,2) -- (3.5,2);
    \foreach \x in {0,1,2}{
      \pic[pic text=$X_g$,yscale=-1] at (\x,0) {mpotensor};
      \pic[pic text=$X_h$,yscale=-1] at (\x,1) {mpotensor};
      \pic[pic text=$X_k$,yscale=-1] at (\x,2) {mpotensor};
    }
    \pic[x={(-1,0)}] at (-3.25,0.5) {reduction};
    \pic[x={(-1,0)}] at (-4,1.5) {longreduction};
    \node[yscale=-1] at (3.6,2.5) {$W(gh,k)$};
    \node[yscale=-1] at (3.3,-0.5) {$W(g,h)$};
  \end{tikzpicture}
\end{equation}
We show now that this scalar $\lambda$ forms a three-cocycle due to associativity of the product. 
 For the fixed reductions $V(g,h)$ and $W(g,h)$, denote the l.h.s.\ of \cref{eq:mpo_reduction_gh} as $[g[hk]]$, the r.h.s.\ as $[[gh]k]$. Consider a product of four MPOs, $ghkl$, and the  following sequence of reductions:
\begin{equation}
  [[[gh]k]l] \rightarrow [[gh][kl]] \rightarrow [g[h[kl]]] \rightarrow [g[[hk]l]] \rightarrow [[g[hk]]l] \rightarrow [[[gh]k]l] .
\end{equation}
In this sequence, every member can be transformed to the next by changing the reduction of three consecutive group elements. Therefore, every member is related to the previous one by a scalar. Writing out these scalars, we obtain
\begin{equation}
  [[[gh]k]l] = \underbrace{\lambda(gh,k,l)^{-1} \cdot \lambda(g,h,kl)^{-1} \cdot \lambda(h,k,l) \cdot \lambda(g,hk,l)\cdot \lambda(g,h,k)}_{=1} \cdot [[[gh]k]l].
\end{equation}
As this relation is the defining relation for the three-cocycles, $\lambda:G^3\to \mathbb{C}^*$ is a three-cocycle, where $G$ acts trivially on $\mathbb{C}^*$. 

Note that the above construction depends on the fixed reductions $V(g,h),W(g,h)$ of the product of two operators. In general, changing the reduction also changes the scalar. This change, however, is not arbitrary: we prove now that it forms a three-coboundary. Consider another reduction $\tilde{V}(g,h)$ and $\tilde{W}(g,h)$ with corresponding three-cocycle $\tilde{\lambda}$. Then, denoting the reduction with $\tilde{V}(g,h)$ and $\tilde{W}(g,h)$ by round brackets (in the sense as above), using \cref{thm:uniqueness}, 
\begin{equation}
  (gh) = \omega(g,h) [gh]
\end{equation}
for some $\omega(g,h)\in \mathbb{C}$. Therefore, the two scalars $\lambda$ and $\tilde{\lambda}$ are related as follows:
\begin{align}
  ((gh)k) &= \omega(g,h) \omega(gh,k) [[gh]k]\\
  (g(hk)) &= \omega(h,k) \omega(g,hk) [g[hk]].
\end{align}
Therefore, the relation between $\lambda$ and $\tilde{\lambda}$ is
\begin{equation}
  \tilde{\lambda}(g,h,k) =  \frac{\omega(g,h) \omega(gh,k)}{\omega(h,k) \omega(g,hk) }\lambda(g,h,k)
\end{equation}
This is the defining relation of three-coboundaries, thus $\lambda/\tilde{\lambda}:G^3\to \mathbb{C}^*$ is a three-coboundary. Therefore, $\lambda$, by construction, is a three-cocycle defined up to a three-coboundary, thus, by the definition of the cohomology group, it is  an element from  $H^3(G,\mathbb{C^*})$.

Next, consider MPO representations that are translationally invariant after blocking two tensors $X$ and $Y$. The previous method assigns two possibly different labels from $H^3(G,\mathbb{C}^*)$ to the two MPO tensors $XY$ and $YX$. We will show now that these two labels are in fact equal.

\begin{prop}\label{prop:alternating_labeling}
  Let $V_n(X_gY_g)$ be a one-block projective MPO representation of $G$. Then $V_n(Y_g X_g)$ is also a one-block projective MPO representation of $G$ and their third cohomology label is the same.
\end{prop}

\begin{proof}
  As $V_n(Y_g X_g)$ is the same MPO as $V_n(X_gY_g)$, but shifted by one lattice site, it is a one-block projective MPO representation. W.l.o.g, one can suppose that both $X_g Y_g$ and $Y_g X_g$ are injective: they contain only one block, thus they can be reduced to injective MPOs. Thus, incorporating the reductions into $X_g$ and $Y_g$, we obtain two new tensors such that both $X_g Y_g$ and $Y_g X_g$ are injective. 
  
  Let $V(g,h)$ and $W(g,h)$ be reductions corresponding to the product of $X_g Y_g$ and $X_h Y_h$, while $\tilde{V}(g,h)$ and $\tilde{W}(g,h)$ be reductions for the product of $Y_g X_g$ and $X_h Y_h$. Then \cref{prop:alternating_reduction} in \cref{app:mps_reductions} implies that $V(g,h)$ and $\tilde{W}(g,h)$ reduces (up to a scalar) a chain of odd number of MPO tensors: 
  \begin{equation}
    \begin{tikzpicture}[scale=0.7, font=\footnotesize]
      \foreach \x/\t in {0/X,1/Y,2/X}{
        \pic[pic text = $\t_h$] at (\x,0)  {mpotensor};
        \pic[pic text = $\t_g$] at (\x,1)  {mpotensor};
      }
      \pic[x={(-1,0)}] at (-3.25,0.5) {reduction};
      \pic at (-1.25,0.5) {reduction};
      \node at (-1.7,1.2) {$V(g,h)$};
      \node at (3.85,1.2) {$\tilde{W}(g,h)$};
    \end{tikzpicture}
    =
    \mu(g,h)\ 
    \begin{tikzpicture}[scale=0.7]
      \foreach \x/\t in {0/X,1/Y,2/X}   
        \pic[pic text = $\t_{gh}$, font=\footnotesize] at (\x,0) {mpotensor};
    \end{tikzpicture}
  \end{equation}
  Therefore, for the product of three MPOs corresponding to $g,h$ and $k$ and a chain consisting of an odd number of MPO tensors,
  \begin{equation}\label{eq:alternate_cohomology_1}
    \begin{tikzpicture}[rotate=180,scale=0.7,font=\footnotesize]
      \node[rotate=180] at (-1.7,-0.2) {$\tilde{W}(g,h)$};
      \node[rotate=180] at (3.85,-0.2) {$V(g,h)$};
      \node[rotate=180] at (-2.7,2.2) {$\tilde{W}(gh,k)$};
      \node[rotate=180] at (4.75,2.2) {$V(gh,k)$};
      \draw (-1.5,2) -- (3.5,2);
      \foreach \x/\t in {0/X,1/Y,2/X} {
        \pic[rotate=180,pic text = $\t_g$, font=\footnotesize] at (\x,0) {mpotensor};
        \pic[rotate=180,pic text = $\t_h$, font=\footnotesize] at (\x,1) {mpotensor};
        \pic[rotate=180,pic text = $\t_k$, font=\footnotesize] at (\x,2) {mpotensor};
      }
      \pic[x={(-1,0)}] at (-3.25,0.5) {reduction};
      \pic at (-1.25,0.5) {reduction};
      \pic[x={(-1,0)}] at (-4,1.5) {longreduction};
      \pic at (-2,1.5) {longreduction};
    \end{tikzpicture}
    = \mu(g,h) \mu(gh,k)\ 
    \begin{tikzpicture}[scale=0.7]
      \foreach \x/\t in {0/X,1/Y,2/X}   
        \pic[pic text = $\t_{ghk}$, font=\scriptsize] at (\x,0) {mpotensor};
    \end{tikzpicture}
  \end{equation}
  Similarly,   
  \begin{equation}
    \begin{tikzpicture}[scale=0.7,font=\footnotesize]
      \node at (-1.7,-0.2) {$V(h,k)$};
      \node at (3.85,-0.2) {$\tilde{W}(h,k)$};
      \node at (-2.7,2.2) {$V(g,hk)$};
      \node at (4.75,2.2) {$\tilde{W}(g,hk)$};
      \draw (-1.5,2) -- (3.5,2);
      \foreach \x/\t in {0/X,1/Y,2/X} {
        \pic[pic text = $\t_k$, font=\footnotesize] at (\x,0) {mpotensor};
        \pic[pic text = $\t_h$, font=\footnotesize] at (\x,1) {mpotensor};
        \pic[pic text = $\t_g$, font=\footnotesize] at (0\x,2) {mpotensor};
      }
      \pic[x={(-1,0)}] at (-3.25,0.5) {reduction};
      \pic at (-1.25,0.5) {reduction};
      \pic[x={(-1,0)}] at (-4,1.5) {longreduction};
      \pic at (-2,1.5) {longreduction};
    \end{tikzpicture}
  = \mu(g,hk) \mu(h,k)\ 
    \begin{tikzpicture}[scale=0.7]
      \foreach \x/\t in {0/X,1/Y,2/X}   
        \pic[pic text = $\t_{ghk}$, font=\scriptsize] at (\x,0) {mpotensor};
    \end{tikzpicture}
  \end{equation}  
  If the above chain is long enough, changing the order of the reductions $\tilde{W}$ changes the above equation only by a scalar $\tilde{\lambda}(g,h,k)$:
  \begin{equation}
    \begin{tikzpicture}[scale=0.7,font=\footnotesize]
      \node at (-1.7,-0.2) {$V(h,k)$};
      \node at (-1.5,2.7) {$V(g,hk)$};
      \draw (-1.5,2) -- (2.5,2);
      \draw (-0.5,0) -- (3.5,0);
      \foreach \x/\t in {0/X,1/Y,2/X} {
        \pic[pic text = $\t_k$, font=\footnotesize] at (\x,0) {mpotensor};
        \pic[pic text = $\t_h$, font=\footnotesize] at (\x,1) {mpotensor};
        \pic[pic text = $\t_g$, font=\footnotesize] at (0\x,2) {mpotensor};
      }
      \pic at (-1.25,0.5) {reduction};
      \pic at (-2,1.5) {longreduction};
      \begin{scope}[y={(0,-1)},yshift=2cm]
        \node at (3.7,-0.3) {$\tilde{W}(g,h)$};
        \node at (3.75,2.6) {$\tilde{W}(gh,k)$};
        \pic[x={(-1,0)}] at (-4,1.5) {longreduction};
        \pic[x={(-1,0)}] at (-3.25,0.5) {reduction};
      \end{scope}
    \end{tikzpicture}
  = \tilde{\lambda}(g,h,k) \mu(g,hk) \mu(h,k)\ 
    \begin{tikzpicture}[scale=0.7]
      \foreach \x/\t in {0/X,1/Y,2/X}   
        \pic[pic text = $\t_{ghk}$, font=\scriptsize] at (\x,0) {mpotensor};
    \end{tikzpicture} \ .
  \end{equation}    
  Similarly, changing  the order of the reductions on the left side, we get (notice that the scalar associated to changing the order of the reductions on the left side is the inverse of that on the right side, see \cref{thm:uniqueness})
  \begin{equation}
    \begin{tikzpicture}[rotate=180,scale=0.7,font=\footnotesize]
      \node[rotate=180] at (-1.7,-0.2) {$\tilde{W}(g,h)$};
      \node[rotate=180] at (3.1,-0.5) {$V(g,h)$};
      \node[rotate=180] at (-1.6,2.6) {$\tilde{W}(gh,k)$};
      \node[rotate=180] at (3.55,2.6) {$V(gh,k)$};
      \draw (-1.5,2)--(3.5,2); 
      \foreach \x/\t in {0/X,1/Y,2/X} {
        \pic[rotate=180,pic text = $\t_g$, font=\footnotesize] at (\x,0) {mpotensor};
        \pic[rotate=180,pic text = $\t_h$, font=\footnotesize] at (\x,1) {mpotensor};
        \pic[rotate=180,pic text = $\t_k$, font=\footnotesize] at (\x,2) {mpotensor};
      }
      \pic[x={(-1,0)}] at (-3.25,0.5) {reduction};
      \pic at (-1.25,0.5) {reduction};
      \pic[x={(-1,0)}] at (-4,1.5) {longreduction};
      \pic at (-2,1.5) {longreduction};
    \end{tikzpicture}
    = \frac{\tilde{\lambda}(g,h,k) \mu(g,hk) \mu(h,k)}{\lambda(g,h,k)}\ 
    \begin{tikzpicture}[scale=0.7]
      \foreach \x/\t in {0/X,1/Y,2/X}   
        \pic[pic text =$\t_{ghk}$, font=\scriptsize] at (\x,0) {mpotensor};
    \end{tikzpicture}  .
  \end{equation}
  Comparing this equation with \cref{eq:alternate_cohomology_1}, we conclude that 
  \begin{equation}
    \frac{\lambda(g,h,k)}{\tilde{\lambda}(g,h,k)} = \frac{\mu(g,hk) \mu(h,k)}{\mu(g,h) \mu(gh,k)} \ .
  \end{equation}
  Therefore, the two scalars differ only by a three coboundary. That is, the two third cohomology labels corresponding to $X_g Y_g$ and $Y_g X_g$ coincide.
\end{proof}

\subsection{Third cohomology labeling of semi-injective PEPS}
We investigate the following setup. Let $G$ be a group, $O_g$ a faithful (not necessarily unitary) representation of $G$. Let $\ket{\phi}$ be a four-partite state with full rank one-particle reduced densities. Suppose $\forall g\in G$, $O_g$ is a symmetry of the semi-injective PEPS defined by $\ket{\phi}$ and $\id$:
\begin{equation}\label{eq:spt_setup}
\begin{tikzpicture}
\truncstateswo[blue,red] (0,0) (4,3)
\end{tikzpicture} = \mu_{n,m}(g) \ 
\begin{tikzpicture}
\truncstates[blue] (0,0) (4,3)
\end{tikzpicture}\ ,
\end{equation}
where the blue squares represent $\ket{\phi}$, the red operators $O_g$.

Note that this setup can readily be applied for unitary symmetries of semi-injective PEPS: let the semi-injective PEPS be defined by the four-partite state $\ket{\phi}$ and an invertible operator $A$. Let the unitary representation of the symmetry group $G$ be $U_g$. Then, by inverting $A$ in the symmetry condition, we arrive to \cref{eq:spt_setup} with $O_g = A^{-1} U_g A$.

\begin{prop}
  If \cref{eq:spt_setup} holds for some $n,m\geq 3$, then it holds for all $n,m$ and $\mu_{n,m}(g) = \mu^{nm}(g)$ for some one-dimensional representation $\mu$ of $G$.
\end{prop}

\begin{proof}
  Apply \cref{prop:size_indep} and notice that $\mu$ is a representation.
\end{proof}

We show now that the action of the symmetries show up on the boundary as a projective MPO representation of the group $G$.

\begin{prop}\label{prop:spt_boundary}
  If \cref{eq:spt_setup} holds, then for every $g\in G$ there are two MPO tensors $X_g$ and $Y_g$ such that 
  \begin{equation} \label{eq:spt_boundary}
  \begin{tikzpicture}
  \scope
  \clip (0.5,-0.5-\d) rectangle (3.5,0.5+\d);
  \foreach \x in {0,...,3}{
    \pic[blue,rotate=90] at (\x,0.5+\d) {halfstate};
    \pic[blue,rotate=-90] at (\x,-0.5-\d) {halfstate};
    \draw[red] (\x+0.5,0) circle (0.3);
  }
  \endscope
  \node[font=\tiny,anchor=east] at (0.6,0) {$\dots$};
  \node[font=\tiny,anchor=west] at (3.4,0) {$\dots$};
  \end{tikzpicture}
  =
  \mu^n(g)\
    \begin{tikzpicture}[blue,font=\scriptsize]
        \foreach \x in {0,...,2}{
          \pic[rotate=90] at (\x,0.5+\d) {halfstate};
          \pic[rotate=-90] at (\x,-0.5-\d) {halfstate};
          \pic[pic text = $X_g$, anchor=south west] at (\x,0.8) {boundarympo};
          \pic[pic text = $Y_g$, anchor=north west] at (\x,-0.8) {boundarympo};
        }
    \node[black,font=\tiny,anchor=east] at (-0.5,0.8) {$\dots$};
    \node[black,font=\tiny,anchor=west] at (2.5,0.8) {$\dots$};
    \node[black,font=\tiny,anchor=east] at (-0.5,-0.8) {$\dots$};
    \node[black,font=\tiny,anchor=west] at (2.5,-0.8) {$\dots$};
    \end{tikzpicture} \ ,
  \end{equation}
  and $V_n(Y_g) = \left(V_n(X_g)\right)^{-1}$ for all $n$. Moreover, $V_n(X_g)$ and $V_n(Y_g)$ form projective representations of $G$ with $V_n(X_g) V_n(X_h) = \lambda^n (g,h) V_n (X_g X_h)$ for a two-cocycle $\lambda$. In particular, $V_n(X_g) V_n(X_h)$ has only one block in its canonical form.
\end{prop}

\begin{proof}
  From \cref{prop:boundary_mpo}, the existence of $X_g$ and $Y_g$ is clear. From \cref{cor:boundary_unique}, it is also true that $V_n(X_g) V_n(X_h) = \lambda^n(g,h) V_n(X_g X_h)$. Due to associativity, $\lambda(g,h)\lambda(gh,k) = \lambda(g,hk)\lambda(h,k)$, and thus $\lambda$ forms a two-cocycle.
\end{proof}

Note that if we allow for blocking, there is a length scale $K$ for which $\lambda^{Kn}(g,h)$ becomes constant 1. On the other hand, the labeling with an element from the third cohomology group $H^3(G,\mathbb{C})$ corresponding to the  $g\mapsto X_g$  one-block projective MPO representation of $G$ is a scale-invariant labeling.

In the following, we show that the classification of the boundary MPO representation $V_n(X_g)$ also shows up in the MPO defined by $O_g$. To see this, we define a translationally invariant (on four sites) MPO from $O_g$ that we call $V_n(\tilde{O}_g)$. Write $O_g$ as an MPO in \cref{eq:O_mpo_decomp}, and open one of the indices. We call this tensor $\tilde{O}_g$. Pictorially,
\begin{equation}
  O_g = 
  \begin{tikzpicture}[baseline=0.4cm]
    \draw[thick] (0,0) rectangle (1,1);
    \foreach \x in {(0,0),(0,1),(1,0),(1,1)}{
      \draw[thick,fill=white] \x circle (0.1);
    }
  \end{tikzpicture} \quad
  \Rightarrow  \quad 
  \tilde{O}_g = 
  \begin{tikzpicture}[baseline=0.4cm]
  \draw[thick] (0,0)--(0,1) -- (1,1) -- (1,0);
  \draw[thick] (0,0)--(0.3,0) -- (0.3,-0.3) -- (-0.3,-0.3);
  \draw[thick] (1,0)--(0.7,0) -- (0.7,-0.3) -- (1.3,-0.3);
  \foreach \x in {(0,0),(0,1),(1,0),(1,1)}{
    \draw[thick,fill=white] \x circle (0.1);
  }
  \end{tikzpicture}
\end{equation}
This MPO plays an important role in the third cohomology labeling of semi-injective PEPS.

\begin{prop}\label{prop:spt_onsite_mpo}
  The MPOs $V_n(\tilde{O}_g)$ form a one-block projective MPO representation of $G$. Its third cohomology label is the same as that of $V_n(X_g)$.
\end{prop}

\begin{proof}
  As the product of the Schmidt vectors of $O_g$ and $O_g^{-1}$ factorizes, the tensor $\sum_j \tilde{O}_g^{ij} \tilde{O}_{g^{-1}}^{jk}$ has the following structure:
  \begin{equation}
    \begin{tikzpicture}[scale=0.5]
    \foreach \x in {0,1,2,3}{
      \pic at (\x,0) {mpotensor=gray};
      \pic at (\x,1) {mpotensor};
    }
    \end{tikzpicture}
    =
    \begin{tikzpicture}[scale=0.5]
    \draw[thick] (-0.5,0)--(3.5,0);
    \draw[thick] (-0.5,1)--(3.5,1);
    \fill[white] (-0.2,-0.2) rectangle (3.2,1.2);
    \foreach \x in {0,1,2,3} \draw[thick] (\x,-0.5)--(\x,1.5);
    \draw[thick, fill=white, rounded corners=3pt] (-0.2,-0.2) rectangle (0.2,1.2);
    \draw[thick, fill=white, rounded corners=3pt] (2.8,-0.2) rectangle (3.2,1.2);
    \end{tikzpicture},
  \end{equation}
  with 
  \begin{equation}
    \begin{tikzpicture}[scale=0.5]
      \draw[thick] (0,0)--(1,0);
      \draw[thick] (0,1)--(1,1);
      \foreach \x in {0,1} \draw[thick] (\x,-0.5)--(\x,1.5);
      \draw[thick, fill=white, rounded corners=3pt] (-0.2,-0.2) rectangle (0.2,1.2);
      \draw[thick, fill=white, rounded corners=3pt] (0.8,-0.2) rectangle (1.2,1.2);
    \end{tikzpicture}=
    \begin{tikzpicture}[scale=0.5]
      \foreach \x in {0,1} \draw[thick] (\x,-0.5)--(\x,1.5);
    \end{tikzpicture}\ .
  \end{equation}
  Therefore, $V_n(\tilde{O}_g) V_n(\tilde{O}_{g^{-1}})=\id$, as it is the $n$-fold product of this tensor. 
  
  We prove now that $V_n(\tilde{O}_g)V_n(\tilde{O}_h)V_n(\tilde{O}_{(gh)^{-1}}) = \id$, and thus $V_n(\tilde{O}_g)V_n(\tilde{O}_h)=V_n(\tilde{O}_{gh})$.
  
  Consider the MPS tensor defined by the Schmidt vectors of $O_g$, $O_h$ and then $O_{(gh)^{-1}}$ acting on the Schmidt vectors of $\ket{\phi}$. Then, similar to \cref{eq:schmidt_then_inverse}, this tensor can be written as 
  \begin{equation}
    \begin{tikzpicture}
      \pic[blue,rotate=90] at (0.5,0.5+\d) {halfstate};
      \pic[blue,rotate=-90] at (0.5,-0.5-\d) {halfstate};
      \draw[red] (\d-\dd,-\d-\dd) rectangle (\d+\dd,+\d+\dd); 
      \draw[red] (1-\d-\dd,-\d-\dd) rectangle (1-\d+\dd,+\d+\dd); 
      \draw[red] (\d-\dd,0.15) --+(-0.15-\dd,0);
      \draw[red] (1-\d+\dd,0.15) --+(0.15+\dd,0);
      \draw[red] (\d-\dd,-0.15) --+(-0.15-\dd,0);
      \draw[red] (1-\d+\dd,-0.15) --+(0.15+\dd,0);
      \draw[green] (\d-2*\dd,-\d-2*\dd) rectangle (\d+2*\dd,+\d+2*\dd); 
      \draw[green] (1-\d-2*\dd,-\d-2*\dd) rectangle (1-\d+2*\dd,+\d+2*\dd); 
      \draw[green] (\d-2*\dd,0.2) --+(-0.2,0);
      \draw[green] (1-\d+2*\dd,0.2) --+(0.2,0);
      \draw[green] (\d-2*\dd,-0.2) --+(-0.2,0);
      \draw[green] (1-\d+2*\dd,-0.2) --+(0.2,0);
      \draw[red,dashed] (\d-3*\dd,-\d-3*\dd) rectangle (\d+3*\dd,+\d+3*\dd); 
      \draw[red,dashed] (1-\d-3*\dd,-\d-3*\dd) rectangle (1-\d+3*\dd,+\d+3*\dd); 
      \draw[red] (\d-3*\dd,0.25) --+(-0.2,0);
      \draw[red] (1-\d+3*\dd,0.25) --+(0.2,0);
      \draw[red] (\d-3*\dd,-0.25) --+(-0.2,0);
      \draw[red] (1-\d+3*\dd,-0.25) --+(0.2,0);
    \end{tikzpicture} \ = \
    \begin{tikzpicture}[blue,font=\scriptsize]
      \foreach \x in {0}{
        \pic[rotate=90] at (\x,0.5+\d) {halfstate};
        \pic[rotate=-90] at (\x,-0.5-\d) {halfstate};
        \pic[anchor = south west, pic text =$Y_{gh}$] at (\x,0.8) {boundarympo};
        \pic[anchor = north west, pic text =$X_{gh}$] at (\x,-0.8) {boundarympo};
        \pic[anchor = south west, pic text =$X_{h}$] at (\x,1.3) {boundarympo};
        \pic[anchor = north west, pic text =$Y_{h}$] at (\x,-1.3) {boundarympo};
        \pic[anchor = south west, pic text =$X_{g}$] at (\x,1.8) {boundarympo};
        \pic[anchor = north west, pic text =$Y_{g}$] at (\x,-1.8) {boundarympo};
      }
    \end{tikzpicture} \ ,     
  \end{equation}  
  where the red solid rectangle denotes the Schmidt vectors of $O_g$, the green one that of $O_h$, and the dashed one that of $O_{(gh)^{-1}}$. Joining two such tensors, the middle operator is $O_g O_h O_{(gh)^{-1}} = \id$, so
  \begin{equation}
    \begin{tikzpicture}
      \foreach \x in {0,1}{
        \pic[blue,rotate=90] at (\x+0.5,0.5+\d) {halfstate};
        \pic[blue,rotate=-90] at (\x+0.5,-0.5-\d) {halfstate};
      }
      \draw[red] (\d-\dd,-\d-\dd) rectangle (\d+\dd,+\d+\dd); 
      \draw[red] (2-\d-\dd,-\d-\dd) rectangle (2-\d+\dd,+\d+\dd); 
      \draw[red] (\d-\dd,0.15) --+(-0.15-\dd,0);
      \draw[red] (2-\d+\dd,0.15) --+(0.15+\dd,0);
      \draw[red] (\d-\dd,-0.15) --+(-0.15-\dd,0);
      \draw[red] (2-\d+\dd,-0.15) --+(0.15+\dd,0);
      \draw[green] (\d-2*\dd,-\d-2*\dd) rectangle (\d+2*\dd,+\d+2*\dd); 
      \draw[green] (2-\d-2*\dd,-\d-2*\dd) rectangle (2-\d+2*\dd,+\d+2*\dd); 
      \draw[green] (\d-2*\dd,0.2) --+(-0.2,0);
      \draw[green] (2-\d+2*\dd,0.2) --+(0.2,0);
      \draw[green] (\d-2*\dd,-0.2) --+(-0.2,0);
      \draw[green] (2-\d+2*\dd,-0.2) --+(0.2,0);
      \draw[red,dashed] (\d-3*\dd,-\d-3*\dd) rectangle (\d+3*\dd,+\d+3*\dd); 
      \draw[red,dashed] (2-\d-3*\dd,-\d-3*\dd) rectangle (2-\d+3*\dd,+\d+3*\dd); 
      \draw[red] (\d-3*\dd,0.25) --+(-0.2,0);
      \draw[red] (2-\d+3*\dd,0.25) --+(0.2,0);
      \draw[red] (\d-3*\dd,-0.25) --+(-0.2,0);
      \draw[red] (2-\d+3*\dd,-0.25) --+(0.2,0);
    \end{tikzpicture} \
    =
    \begin{tikzpicture}[blue,font=\scriptsize]
      \foreach \x in {0,1}{
        \pic[rotate=90] at (\x,0.5+\d) {halfstate};
        \pic[rotate=-90] at (\x,-0.5-\d) {halfstate};
        \pic[anchor = south west, pic text =$Y_{gh}$] at (\x,0.8) {boundarympo};
        \pic[anchor = north west, pic text =$X_{gh}$] at (\x,-0.8) {boundarympo};
        \pic[anchor = south west, pic text =$X_{h}$] at (\x,1.3) {boundarympo};
        \pic[anchor = north west, pic text =$Y_{h}$] at (\x,-1.3) {boundarympo};
        \pic[anchor = south west, pic text =$X_{g}$] at (\x,1.8) {boundarympo};
        \pic[anchor = north west, pic text =$Y_{g}$] at (\x,-1.8) {boundarympo};
      }
    \end{tikzpicture} \ .     
  \end{equation}  
  
  As the l.h.s.\ factorizes w.r.t.\ the vertical cut, and the r.h.s.\ factorizes w.r.t.\ the horizontal cut, and the one particle reduced densities of $\ket\phi$ are full rank, the product of the Schmidt vectors of $O_g$, $O_h$ and $O_{(gh)^{-1}}$ also factorizes, and thus
  \begin{equation}
    \begin{tikzpicture}[scale=0.5]
    \foreach \x in {0,1,2,3}{
      \pic at (\x,0) {mpotensor=gray};
      \pic at (\x,1) {mpotensor};
      \pic at (\x,2) {mpotensor};
    }
    \end{tikzpicture}
    =
    \begin{tikzpicture}[scale=0.5]
    \draw[thick] (-0.5,0)--(3.5,0);
    \draw[thick] (-0.5,1)--(3.5,1);
    \draw[thick] (-0.5,2)--(3.5,2);
    \fill[white] (-0.2,-0.2) rectangle (3.2,2.2);
    \foreach \x in {0,1,2,3} \draw[thick] (\x,-0.5)--(\x,2.5);
    \draw[thick, fill=white, rounded corners=3pt] (-0.2,-0.2) rectangle (0.2,2.2);
    \draw[thick, fill=white, rounded corners=3pt] (2.8,-0.2) rectangle (3.2,2.2);
    \end{tikzpicture},
  \end{equation}
  with 
  \begin{equation}
    \begin{tikzpicture}[scale=0.5]
      \draw[thick] (0,0)--(1,0);
      \draw[thick] (0,1)--(1,1);
      \draw[thick] (0,2)--(1,2);
      \foreach \x in {0,1} \draw[thick] (\x,-0.5)--(\x,2.5);
      \draw[thick, fill=white, rounded corners=3pt] (-0.2,-0.2) rectangle (0.2,2.2);
      \draw[thick, fill=white, rounded corners=3pt] (0.8,-0.2) rectangle (1.2,2.2);
    \end{tikzpicture}=
    \begin{tikzpicture}[scale=0.5]
      \foreach \x in {0,1} \draw[thick] (\x,-0.5)--(\x,2.5);
    \end{tikzpicture}\ .
  \end{equation}

Therefore  $V_n(\tilde{O}_g)V_n(\tilde{O}_h)V_n(\tilde{O}_{(gh)^{-1}}) = \id$, as it is the $n$-fold product of this tensor. This means that $V_n(\tilde{O}_g)$ is an MPO representation.

As an MPO representation is also a one-block projective MPO representation, one can label this MPO representation with an element from the third cohomology group $H^3(G,\mathbb{C}^*)$. We now show that this label coincides with that of the projective MPO representation of $G$ on the boundary. To see this, partially contract the MPS tensors describing the boundary of the state (defined in \cref{eq:boundary_mpo_tensor_def}). That is, contract only the lower indices:
  \begin{equation}\label{eq:O_coh_tensor}
    \begin{tikzpicture}
      \pic[blue,rotate=90] at (0.5,0.5+\d) {halfstate};
      \pic[blue,rotate=-90] at (0.5,-0.5-\d) {halfstate};
      \draw[red] (\d-\dd,-\d-\dd) rectangle (\d+\dd,+\d+\dd); 
      \draw[red] (1-\d-\dd,-\d-\dd) rectangle (1-\d+\dd,+\d+\dd); 
      \draw[red] (\d-\dd,0.15) --+(-0.15,0);
      \draw[red] (1-\d+\dd,0.15) --+(0.15,0);
      \draw[red] (\d-\dd,-0.15) --+(-0.15,0);
      \draw[red] (1-\d+\dd,-0.15) --+(0.15,0);
      \draw[red] (1-\d+\dd+0.15,-0.15)--(1-\d+\dd+0.15,-1)--(\d-\dd-0.15,-1)--(\d-\dd-0.15,-0.15);
    \end{tikzpicture} \
    =\ \mu \
    \begin{tikzpicture}[blue,font=\scriptsize]
      \pic[rotate=90] at (0,0.5+\d) {halfstate};
      \pic[rotate=-90] at (0,-0.5-\d) {halfstate};
      \pic[anchor = south west, pic text =$X_g$] at (0,0.8) {boundarympo};
      \pic[anchor = north west, pic text =$Y_g$] at (0,-0.8) {boundarympo};
      \draw[black] (0.5,-0.8)--(0.5,-1.3)--(-0.5,-1.3)--(-0.5,-0.8);
    \end{tikzpicture} \ .
  \end{equation}
Notice that the red MPO tensor acting on the l.h.s.\ is exactly $\tilde{O}_g$. After contracting these tensors, \cref{eq:O_coh_tensor} reads
\begin{equation}
\begin{tikzpicture}[scale=0.7,font=\scriptsize]
\def \c {purple}
\def \cc {orange}

\draw[thick,red] (0,0.2)--++(-0.3,0);
\foreach \x in {0,1.4,2.8}{
  \draw[thick,red] (\x, 0.2)--++(0,-0.4);
  \draw[thick,red] (\x+1, 0.2)--++(0,-0.4);   
  \draw[thick,red] (\x+1, 0.2)--++(0.3,0);  
  \draw[thick,red] (\x,-0.26)--++(1,0); 
  \foreach \y in {(\x,0.2),(\x,-0.2),(\x+1,0.2),(\x+1,-0.2)} 
  \fill[white] \y circle (0.14);  
  \draw[thick,blue] (\x,0.7)--(\x,0.2)--(\x+1,0.2)--(\x+1,0.7); 
  \draw[thick,blue] (\x,-0.7)--(\x,-0.2)--(\x+1,-0.2)--(\x+1,-0.7);
  \foreach \y in {(\x,0.2),(\x,-0.2),(\x+1,0.2),(\x+1,-0.2)} 
  \draw[thick,red] \y circle (0.14);
  \foreach \y in {(\x,0.2),(\x,-0.2),(\x+1,0.2),(\x+1,-0.2)} 
  \filldraw \y circle (0.07);
  \draw[rounded corners=3pt,blue] (\x-0.1,0.7) rectangle (\x+1.1,0.9);
  \draw[rounded corners=3pt,blue] (\x-0.1,-0.7) rectangle (\x+1.1,-0.9);  
  \draw (\x+0.5,0.9)--++(0,0.2);
  \draw (\x+0.5,-0.9)--++(0,-0.2);
}
  \node at (4.4,0) {$\dots$};
  \node at (-0.6,0) {$\dots$};
\end{tikzpicture}
=
\begin{tikzpicture}[scale=0.7,font=\scriptsize]
\def \c {yellow}
\def \cc {orange}
\foreach \x in {0, 1.4, 2.8}{
  \draw[thick,blue] (\x,0.7)--(\x,0.2)--(\x+1,0.2)--(\x+1,0.7); 
  \draw[thick,blue] (\x,-0.7)--(\x,-0.2)--(\x+1,-0.2)--(\x+1,-0.7); 
  \foreach \y in {(\x,0.2),(\x,-0.2),(\x+1,0.2),(\x+1,-0.2)} \filldraw \y circle (0.1);
  \draw[rounded corners=3pt,blue] (\x-0.1,0.7) rectangle (\x+1.1,0.9);
  \draw[rounded corners=3pt,blue] (\x-0.1,-0.7) rectangle (\x+1.1,-0.9);  
  \draw[thick,blue] (\x+0.5,0.9)--++(0,0.6);      
  \draw[thick,blue] (\x+0.5,-0.9)--++(0,-0.6);
  \draw (\x-0.2,1.2)--(\x+1.2,1.2);      
  \draw (\x+0.2,-1.2)--(\x+0.8,-1.2) -- (\x+0.8,-1.4)--(\x+0.2,-1.4)--cycle; 
  \draw[fill=white] (\x+0.5,1.2) circle (0.1);
  \draw[fill=gray] (\x+0.5,-1.2) circle (0.1);
  
}
  \node at (4.4,0) {$\dots$};
  \node at (-0.6,0) {$\dots$};
\end{tikzpicture}.
\end{equation}
By construction, the red MPO appearing on the l.h.s.\ is $V_n(\tilde{O}_g)$. Therefore, if $V(g,h)$,$W(g,h)$ is a reduction from $V_n(\tilde{O}_g)V_n(\tilde{O}_h)$ to $V_n(\tilde{O}_{gh})$, then it is also a reduction from $V_n(X_g) V_n(X_h)$ to $V_n(X_{gh})$. As the third cohomology is assigned to the MPO representation with the help of these reductions, $V_n(O_g)$ is classified by the same third cohomology class as $V_n(X_g)$.
\end{proof}

The above proof can be repeated for the vertical boundary instead of the horizontal one. This means that the third cohomology label of the vertical boundary is the same as that of $V_n(\tilde{O}'_g)$, where $\tilde{O}'_g = D_g A_g B_g C_g $, if $O_g = A_g B_g C_g D_g$. \cref{prop:alternating_labeling} implies that the third cohomology labeling of $V_n(\tilde{O}'_g)$ and $V_n(\tilde{O}_g)$ coincide, therefore the third cohomology labeling of the horizontal and vertical boundary coincide.

\section{Conclusion}

In this work we introduced a new class of PEPS, semi-injective PEPS. We showed that semi-injective PEPS are a generalization of injective PEPS and that some important examples that are not known to have an injective PEPS description  naturally admit a semi-injective PEPS description. We showed that they are unique ground state of their parent Hamiltonian. We also derived a canonical form, i.e., a way to decide locally if two semi-injective PEPS are equal. One of the necessary conditions is that the boundaries of the two states are related by an invertible MPO. Using this result, the third cohomology labeling of SPT phases extends naturally to semi-injective PEPS, suggesting that these states are appropriate to capture the relevant physics of SPT phases. Using the canonical form, we have found that the third cohomology label of the SPT phase is not only encoded on the edge of the model, but also directly in the symmetry operators. 
 
\section{Acknowledgements}
 The authors would like to thank  Barbara Kraus, Frank Verstraete, David Pérez-García and Vedran Dunjko for helpful discussions.
This project has been supported by the European Union through the ERC
Starting Grant WASCOSYS (No.~636201).

\appendix
\crefalias{section}{appsec}

\section{Examples for canonical form}\label{app:canoniclaform_examples}

In the injective PEPS case, if two tensors generate the same state, then they are related by a product gauge transformation. In the case of semi-injective PEPS, this is no longer true as the following example shows.

Let $A$ be the following MPS tensor:
\begin{eqnarray}
A^0 & = & \left(\begin{matrix}1 & 0 \\ 0 & 2 \end{matrix}\right) \\
A^1 & = & \left(\begin{matrix}24 & -10 \\ 17 & -3 \end{matrix}\right). 
\end{eqnarray}
This tensor was constructed in such a way that it is $Z$ symmetric for size $4$, but not for longer chains: for the tensor $B$ with $B^0 = A^0$ and $B^1 = -A^1$, $V_4(B) = V_4(A)$, but $V_5(A)\neq V_5(B)$. The tensors $A$ and $B$ are also normal, after blocking two tensors they become injective. \Cref{prop:can_form_normal_mps} also implies that $A$ and $B$ are not related by a gauge transform. There is also no gauge relating the tensors after blocking four of them: $\nexists X: \ XBBBBX^{-1}=AAAA$.

Consider two semi-injective PEPS. Let $\Psi_A$ be defined by $\phi_A=V_4(A)$ and $\id$, $\Psi_B$ by $\phi_B=V_4(B)$ and $\id$. By construction, $\Psi_A=\Psi_B$. We will show, however, that the PEPS tensors defined by grouping four MPS tensors:
\begin{equation}
\begin{tikzpicture}
\truncstates[blue] (0,0) (1,1);
\end{tikzpicture} \quad \text{and} \quad 
\begin{tikzpicture}
\truncstates[green] (0,0) (1,1);
\end{tikzpicture}
\end{equation}
are not related by a gauge, where the blue tensors are $A$ and the green ones are $B$. We prove that by contradiction. Suppose there are such gauges, $X$ and $Y$:
\begin{equation}
\begin{tikzpicture}[font=\tiny]
  \scope
    \clip (0.3,0.3) rectangle (1.7,1.7);
    \foreach \x in {0,1} \foreach \y in {0,1} \pic[blue] at (\x,\y) {state};
    \foreach \ang in {0,90,180,270} 
    \draw[rounded corners=2pt,fill=white, rotate around={\ang:(1,1)}] (0.5,0.7) rectangle (0.6,1.3);
  \endscope
  \node at (1.5,0.3) {$Y^{-1}$};
  \node at (1.4,1.6) {$Y$};
  \node at (1.85,1.35) {$X^{-1}$};
  \node at (0.3,1.3) {$X$};
\end{tikzpicture} \ = \  
\begin{tikzpicture}
\truncstates[green] (0,0) (1,1);
\end{tikzpicture}\ .
\end{equation}
Inverting $Y$ and $Y^{-1}$, we get that 
\begin{equation}
\begin{tikzpicture}[font=\tiny]
  \scope
  \clip (0.3,0.3) rectangle (1.7,1.7);
  \foreach \x in {0,1} \foreach \y in {0,1} \pic[blue] at (\x,\y) {state};
  \foreach \ang in {0,180} 
  \draw[rounded corners=2pt,fill=white, rotate around={\ang:(1,1)}] (0.5,0.7) rectangle (0.6,1.3);
  \endscope
  \node at (1.85,1.35) {$X^{-1}$};
  \node at (0.3,1.3) {$X$};
\end{tikzpicture} \ = \  
\begin{tikzpicture}[font=\tiny]
  \scope
    \clip (0.3,0.3) rectangle (1.7,1.7);
    \foreach \x in {0,1} \foreach \y in {0,1} \pic[green] at (\x,\y) {state};
    \foreach \ang in {90,270} 
    \draw[rounded corners=2pt,fill=white, rotate around={\ang:(1,1)}] (0.5,0.7) rectangle (0.6,1.3);
  \endscope
  \node at (1.4,0.3) {$Y$};
  \node at (1.6,1.6) {$Y^{-1}$};
\end{tikzpicture}\ .
\end{equation}
Notice that the l.h.s.\ is product w.r.t.\ the vertical cut, whereas the r.h.s.\ is product w.r.t.\ the horizontal cut. As $A$ and $B$ become injective after blocking two tensors, both $X$ and $Y$ have to be product operators, and thus
\begin{equation}
\begin{tikzpicture}
\clip (0.3,0.3) rectangle (1.7,1.7);
\foreach \x in {0,1} \foreach \y in {0,1} \pic[blue] at (\x,\y) {state};
\foreach \ang in {0,90,180,270}{
\draw[fill=white, rotate around={\ang:(1,1)}] (0.55,1.15) circle (0.08cm);
\draw[fill=white, rotate around={\ang:(1,1)}] (0.55,0.85) circle (0.08cm);
} 
\end{tikzpicture} \ = \  
\begin{tikzpicture}
\truncstates[green] (0,0) (1,1);
\end{tikzpicture}\ .
\end{equation}
But this would mean that after blocking four tensors, $BBBB=XAAAAX^{-1}$ for some gauge $X$. As this is a contradiction, the two given PEPS tensors generating the same semi-injective PEPS are not related by a gauge.

\section{MPS reductions}\label{app:mps_reductions}

In this Section, we present the proofs of the theorems about reductions of MPS used in \cref{sec:spt}. 

\reductionexist*
 
\begin{proof}
  Suppose the injectivity length (see \cref{prop:normal2injective}) of $A$ is $L$. Let $\tilde{A}$ and $\tilde{B}$ denote the tensors obtained from  $A$ and $B$ by blocking them $L$ times, respectively. Then  $\tilde{A}$ has a left inverse, $\tilde{A}^{-1}$. Take the Jordan decomposition of the following matrix: 
  \begin{equation}\label{eq:BA_jordan_decomp}
    \begin{tikzpicture}[baseline=0.5cm]
      \pic[pic text=$\tilde{B}$] at (0,0) {mps};
      \pic[pic text =$\tilde{A}^{-1}$] at (0,1) {mps=180};
    \end{tikzpicture}
    = 
    \begin{tikzpicture}
      \draw (-0.5,0)--(0.9,0);
      \draw (-0.5,1)--(0.9,1);
      \draw[fill=white] (0,-0.2) rectangle (0.4,1.2);
      \node[font=\tiny] at (0.2,0.5) {$S$};
    \end{tikzpicture}
     + 
    \begin{tikzpicture}
      \draw (-0.5,0)--(0.9,0);
      \draw (-0.5,1)--(0.9,1);
      \draw[fill=white] (0,-0.2) rectangle (0.4,1.2);
      \node[font=\tiny] at (0.2,0.5) {$N$};
    \end{tikzpicture}
  \end{equation}
  where $S$ is semi-simple (diagonalizable), $N$ is nilpotent (upper triangular in the basis in which $S$ is diagonal) and $[S,N]=0$. $B$ and $A$ generate the same state, thus
  \begin{equation}
    \begin{tikzpicture}[baseline=0.5cm]
      \draw[decorate, decoration={brace,amplitude=5pt}] (-0.5,2)--(3.5,2) node[midway,anchor=south, yshift=3pt] {$n$}; 
      \node at (2,0) {$\dots$};
      \node at (2,1) {$\dots$};
      \foreach \x in {0,1,3}{
        \pic[pic text=$\tilde{B}$] at (\x, 0) {mps};
        \pic[pic text=$\tilde{A}^{-1}$] at (\x, 1) {mps=180};
      }
      \draw (-0.5,0)--(-0.5,-0.5)--(3.5,-0.5)--(3.5,0);
      \draw (-0.5,1)--(-0.5,1.75)--(3.5,1.75)--(3.5,1);
    \end{tikzpicture}
    =
    \begin{tikzpicture}[baseline=0.5cm]
      \draw[decorate, decoration={brace,amplitude=5pt}] (-0.5,2)--(3.5,2) node[midway,anchor=south, yshift=3pt] {$n$}; 
      \node at (2,0) {$\dots$};
      \node at (2,1) {$\dots$};
      \foreach \x in {0,1,3}{
        \pic[pic text=$\tilde{A}$] at (\x, 0) {mps};
        \pic[pic text=$\tilde{A}^{-1}$] at (\x, 1) {mps=180};
      }
      \draw (-0.5,0)--(-0.5,-0.5)--(3.5,-0.5)--(3.5,0);
      \draw (-0.5,1)--(-0.5,1.75)--(3.5,1.75)--(3.5,1);
    \end{tikzpicture} \ .
  \end{equation}
  The r.h.s\ is $D^n$, where $D$ is the bond dimension of $A$, as it is $n$ times the trace of $\id$. Using the Jordan decomposition \cref{eq:BA_jordan_decomp}, the l.h.s.\ is $\tr (S+N)^n = \tr S^n$. Therefore $\tr S^n=D^n$, thus \cref{prop:sum_power} implies that the rank of $S$ is 1. $[S,N]=0$ therefore implies that $SN=NS=0$. Thus, $(S+N)^n = S^n +N^n = S^n$ if $n$ is larger than the nilpotency rank of $N$. Then,  as $A$ and $B$ generate the same state, for all $n$ and $m$, 
  \begin{equation}\label{eq:exist_reduction_10}
    \begin{tikzpicture}[baseline=0]
      \draw[decorate, decoration={brace,amplitude=5pt}] (-0.5,1.6)--(2.45,1.6) node[midway,anchor=south, yshift=3pt] {$n$}; 
      \draw[decorate, decoration={brace,amplitude=5pt}] (2.55,1.6)--(4.5,1.6) node[midway,anchor=south, yshift=3pt] {$m$};    \draw (-0.5,-0.75) rectangle (4.5,0);
      \foreach \x in {3,4}{
        \pic[pic text=$B$] at (\x, 0) {mps};
      }
      \foreach \x in {0,1,2}{
        \pic[pic text=$\tilde{B}$] at (\x, 0) {mps};
        \pic[pic text=$\tilde{A}^{-1}$] at (\x, 1) {mps=180};
      }
    \end{tikzpicture}\
    =
    D^{n-1}\cdot \ 
    \begin{tikzpicture}[baseline=0]
      \draw[decorate, decoration={brace,amplitude=5pt}] (-0.5,1.6)--(1.5,1.6) node[midway,anchor=south, yshift=3pt] {$m$};   
      \pic[pic text=$A$] at (0,0) {mps};
      \pic[pic text=$A$] at (1,0) {mps};
      \draw (1.5,0)--(1.5,-0.75)--(-1.5,-0.75)--(-1.5,1)--(-2,1);
      \draw (-0.5,0)--(-1,0)--(-1,1)--(-0.5,1);
    \end{tikzpicture}\ ,
  \end{equation}
  where we have used $n-1$ times that $\tr \id =D$. As $S$ is rank one, there are matrices $V$ and $W$ such that $S$ can be written as
  \begin{equation}
    \begin{tikzpicture}
      \draw (-0.5,0)--(0.9,0);
      \draw (-0.5,1)--(0.9,1);
      \draw[fill=white] (0,-0.2) rectangle (0.4,1.2);
      \node[font=\tiny] at (0.2,0.5) {$S$};
    \end{tikzpicture} \ =\ 
    \begin{tikzpicture}
      \draw (-0.5,0)--(0,0);
      \draw (0.5,0)--(1,0);
      \draw (-0.5,1)--(0,1);
      \draw (0.5,1)--(1,1);
      \pic[rotate=90] at (0,0.5) {mpsgauge};
      \pic[rotate=90] at (0.5,0.5) {mpsgauge};
      \node[font=\tiny] at (0,0.5) {$W$};
      \node[font=\tiny] at (0.5,0.5) {$V$};
    \end{tikzpicture}\ .
  \end{equation}
  Therefore, as $N^n=0$, the l.h.s.\ can be rewritten as
  \begin{equation}
    \begin{tikzpicture}[baseline=0]
      \draw[decorate, decoration={brace,amplitude=5pt}] (-0.5,1.6)--(2.45,1.6) node[midway,anchor=south, yshift=3pt] {$n$}; 
      \draw[decorate, decoration={brace,amplitude=5pt}] (2.55,1.6)--(4.5,1.6) node[midway,anchor=south, yshift=3pt] {$m$};    \draw (-0.5,-0.75) rectangle (4.5,0);
      \foreach \x in {3,4}{
        \pic[pic text=$B$] at (\x, 0) {mps};
      }
      \foreach \x in {0,1,2}{
        \pic[pic text=$\tilde{B}$] at (\x, 0) {mps};
        \pic[pic text=$\tilde{A}^{-1}$] at (\x, 1) {mps=180};
      }
    \end{tikzpicture}\
    =
    D^{n-1}\cdot \ 
    \begin{tikzpicture}[baseline=0]
      \draw[decorate, decoration={brace,amplitude=5pt}] (-0.5,1.6)--(1.5,1.6) node[midway,anchor=south, yshift=3pt] {$m$};   
      \pic[pic text=$B$] at (0,0) {mps};
      \pic[pic text=$B$] at (1,0) {mps};
      \draw (1.5,0)--(1.5,-0.75)--(-1.5,-0.75)--(-1.5,1)--(-2,1);
      \draw (-0.5,0)--(-1,0)--(-1,1)--(-0.5,1);
      \pic[rotate=90] at (-1,0.5) {mpsgauge};
      \pic[rotate=90] at (-1.5,0.5) {mpsgauge};
      \node[font=\tiny] at (-1,0.5) {$V$};
      \node[font=\tiny] at (-1.5,0.5) {$W$};
    \end{tikzpicture}\ ,
  \end{equation}
  Therefore, comparing this with the r.h.s.\ of \cref{eq:exist_reduction_10}, for all $m$,
  \begin{equation}
    V B^{i_1} \dots B^{i_ m} W = A^{i_1} \dots A^{i_m}.
  \end{equation}
  For $m=0$, $VW=\id$. 
\end{proof}

\reductionnilpotent*

Before proceeding to the proof, we need the following simple statement:
\begin{lemma}\label{lem:VNW=0}
  Let $V,W$ be a reduction from $B$ to an injectie MPS tensor $A$, $N^i = B^i - W A^i V $. Then for any $m>0$,
  \begin{equation}
    V N^{i_1} N^{i_2} \dots N^{i_m} W = 0 .
  \end{equation}
\end{lemma}
 
\begin{proof}
  We prove this by induction on $m$. For $m=1$, 
  \begin{equation}
    V N^i W  = V B^i W - VWA^i VW = A^i - A^i = 0.
  \end{equation}
  Suppose the statement is true for all $n<m$. Then, writing $N^{i_1} = B^{i_1} - W A^{i_1} V$ and using the induction hypothesis,
  \begin{equation}
    V N^{i_1} N^{i_2} \dots  N^{i_{m}} W =     V B^{i_1} N^{i_2} \dots  N^{i_{m}} W 
  \end{equation}
  Similarly, $N^{i_2}, \dots, N^{i_{m-1}}$ can be changed to $B^{i_2}, \dots, B^{i_{m-1}}$:
  \begin{equation}
    V N^{i_1} N^{i_2} \dots  N^{i_{m}} W =     V B^{i_1} \dots B^{i_{m-1}}   N^{i_{m}} W. 
  \end{equation}
  Writing now $N^{i_{m}} = B^{i_{m}} - W A^{i_{m}} V$, we arrive to 
  \begin{equation}
    V N^{i_1}  \dots  N^{i_{m}} W =  V B^{i_1} \dots   B^{i_{m}} W -    V B^{i_1} \dots B^{i_{m-1}}   W A^{i_{m}} VW =0. 
  \end{equation}  
\end{proof}
 
\begin{proof}[Proof of \cref{prop:reduction_nilpotent}]
  $B$ and $A$ generate the same state:
  \begin{equation}
     \tr\left\{ B^{i_1} B^{i_2} \dots B^{i_{n}}\right\} = \tr\left\{ A^{i_1} A^{i_2} \dots A^{i_{n}}\right\} .
  \end{equation}
  Write $B^i = W A^i V + N^i$ and expand the product on the l.h.s. As $V$ and $W$ form a  reduction, $V N^{i_1} \dots N^{i_m} W =0$ for any $m>0$ and all $i_1,\dots i_m$ by \cref{lem:VNW=0}, and thus all terms cancel except the products of $A$ and the products of $N$. Therefore
  \begin{equation}
    \tr\left\{ N^{i_1} N^{i_2} \dots N^{i_{n}}\right\} =0.
  \end{equation} 
  This means that $\tr\{Z\}=0$ for every element $Z$ in the algebra generated by $N^i$. Thus, in particular, for every $n>0$, $\tr\{Z^{n}\} =0$. Therefore the algebra generated by $N^i$ is a nil algebra, and thus nilpotent\cite{nagata1952}. That is,
  \begin{equation}
    N^{i_1} N^{i_2} \dots N^{i_{n}}=0
  \end{equation}
  for large enough $n$.
\end{proof}


\reductionrelate*

Before proceeding to the proof, we need the following calculation that we use repeatedly:
\begin{lemma}\label{lem:B_expand}
  Let $V,W$ be a reduction from $B$ to a normal tensor $A$, $N^i = B^i - WA^iV$. Let $N_0$ be the nilpotency length of the reduction. Then the following equations hold:
  \begin{align}
    B^{i_1} B^{i_2} \dots B^{i_n} &= \sum_{0\leq k\leq l \leq n} N^{i_1}\dots N^{i_k} W A^{i_{k+1}}\dots A^{i_{l}} V N^{i_{l+1}}\dots N^{i_n}\\
    V B^{i_1} B^{i_2} \dots B^{i_n} &= \sum_{\max(0,n-N_0)\leq l \leq n}  A^{i_{1}}\dots A^{i_{l}} V N^{i_{l+1}}\dots N^{i_n}\\
    B^{i_1} B^{i_2} \dots B^{i_n} W &= \sum_{0\leq k\leq \min(N_0,n)} N^{i_1}\dots N^{i_k} W A^{i_{k+1}}\dots A^{i_{n}}. 
  \end{align}
\end{lemma}

\begin{proof}
  Write $B^{i_j} = WA^{i_j} V + N^{i_j}$ for all $j$ in $B^{i_1} \dots B^{i_n}$ and expand the expression. Using \cref{lem:VNW=0} and the definition of the nilpotency length (\cref{def:nilpotencylength}), we arrive at the desired equations.
\end{proof}

\begin{proof}[Proof of \cref{thm:uniqueness}]
  Let $L$ be the injectivity length of $A$ and let $m=2N_0+L$. Consider
  \begin{equation}
  C^{i_1 \dots i_m} = V B^{i_1} B^{i_2} \dots B^{i_{m}} \tilde{W}.
  \end{equation}
  Using \cref{lem:B_expand} with  $B^i = W A^i V + N^i$, we have
  \begin{equation}
  C^{i_1 \dots i_m} = \sum_{k=N_0+L}^{m} A^{i_1} \dots A^{i_k} V N^{i_{k+1}}  \dots N^{i_m} \tilde{W},
  \end{equation}
  as $m-N_0 = N_0 +L$. Similarly, using \cref{lem:B_expand} with $B^i = \tilde{W} A^i \tilde{V} + \tilde{N}^i$, we get
  \begin{equation}
  C^{i_1 \dots i_m} = \sum_{k=0}^{N_0} V \tilde{N}^{i_{1}}  \dots \tilde{N}^{i_k} \tilde{W} A^{i_{k+1}} \dots A^{i_m}. 
  \end{equation}
  Note that the MPS tensor at position $k=N_0+1$ to $N_0+L$ is  $A^{i_k}$ in both expressions. By assumption, that block is injective. Applying its inverse and comparing the two expressions, we conclude that 
  \begin{align}\label{eq:uniqueness_key}
  \sum_{k=0}^{N_0} V \tilde{N}^{i_{1}}  \dots \tilde{N}^{i_k} \tilde{W} A^{i_{k+1}} \dots A^{i_{N_0}} &= \lambda A^{i_{1}} \dots A^{i_{N_0}}\\
  \sum_{k=0}^{N_0} A^{i_{1}} \dots A^{i_{k}} V N^{i_{k+1}}  \dots N^{i_{N_0}} \tilde{W}  &= \lambda^{-1} A^{i_{1}} \dots A^{i_{N_0}}   \label{eq:uniqueness_key2}
  \end{align}
  for some $\lambda\in\mathbb{C}$. But then, using \cref{lem:B_expand} for $V B^{i_1} B^{i_2} \dots B^{i_{n}}$ with $B^i = \tilde{W} A^i \tilde{V} + \tilde{N}^i$, we get
  \begin{equation}
  V B^{i_1} B^{i_2} \dots B^{i_{n}} = \sum_{k=0}^{N_0} \sum_{l=n-N_0}^{n} V \tilde{N}^{i_{1}}  \dots \tilde{N}^{i_k} \tilde{W} A^{i_{k+1}} \dots A^{i_{l}} \tilde{V} N^{i_{l+1}}  \dots N^{i_n}.
  \end{equation}
  If $n\geq 2N_0$, then $l\geq N_0$. Therefore the left part of the r.h.s.\ can be replaced  using \cref{eq:uniqueness_key}:
  \begin{equation}
  V B^{i_1} B^{i_2} \dots B^{i_{n}} = \lambda \sum_{l=n-N_0}^{n}  A^{i_{1}} \dots A^{i_{l}} \tilde{V} N^{i_{l+1}}  \dots N^{i_n} = \lambda \tilde{V} B^{i_1} B^{i_2} \dots B^{i_{n}},
  \end{equation}
  where the last equation holds by using \cref{lem:B_expand} for $\tilde{V} B^{i_1} B^{i_2} \dots B^{i_{n}}$ with $B^i = \tilde{W} A^i \tilde{V} + \tilde{N}^i$.
  \Cref{eq:BBBW} can be proven similarly using \cref{eq:uniqueness_key2}.
\end{proof}

We now consider MPSs that are translationally invariant after blocking two sites. 

\begin{prop}\label{prop:alternating_reduction}
  Let $A\in \mathbb{C}^{D_1} \otimes \mathbb{C}^{D_2} \otimes \mathbb{C}^{d_1}$ and $B\in \mathbb{C}^{D_2} \otimes \mathbb{C}^{D_1} \otimes \mathbb{C}^{d_2}$ be two tensors such that both $AB$ and $BA$ are normal MPS tensors. Let $C\in \mathbb{C}^{\tilde{D}_1} \otimes \mathbb{C}^{\tilde{D}_2} \otimes \mathbb{C}^{d_1}$ and $D\in \mathbb{C}^{\tilde{D}_2} \otimes \mathbb{C}^{\tilde{D}_1} \otimes \mathbb{C}^{d_1}$ be two tensors such that $V_n(CD) = V_n(AB)$. Then $V_n(BA)=V_n(DC)$ and if $V,W$ are reductions of $CD$ to $AB$ and $\tilde{V},\tilde{W}$ are reductions of $DC$ to $BA$, then for a sufficiently long chain,
  \begin{align}
    AB\dots BA &= \lambda V CD\dots DC \tilde{W} \\
    BA\dots AB &= \mu \tilde{V} DC \dots CD W
  \end{align}
\end{prop}

\begin{proof}
  First, notice that  $V_n(BA)=V_n(DC)$, as $V_n(BA)$ is $V_n(AB)$ shifted by half a lattice constant, while $V_n(DC)$ is $V_n(CD)$ shifted by half a lattice constant. 
   
  Next, using \cref{lem:B_expand} with $CD = W AB V + N_1N_2$, we  have
  \begin{equation}
    V\underbrace{CD\dots D }_{2k} C\tilde{W} = \sum_{i=0}^{M} \underbrace{AB \dots B}_{2k-2i} V \underbrace{N_1 N_2 \dots N_1 N_2}_{2i} C\tilde{W},  
  \end{equation}  
  where $M$ is the injectivity length of the reduction $V,W$. Similarly, using \cref{lem:B_expand} with $DC = \tilde{W} BA \tilde{V} + \tilde{N}_1\tilde{N}_2$, we have
  \begin{equation}
    VC\underbrace{D\dots DC }_{2k} \tilde{W} = \sum_{i=0}^{\tilde{M}} VC \underbrace{\tilde{N}_1 \tilde{N}_2 \dots \tilde{N}_1 \tilde{N}_2}_{2i} \tilde{W} \underbrace{BA \dots BA}_{2k-2i}, 
  \end{equation}
  Where $\tilde{M}$ is the injectivity length of $\tilde{V}, \tilde{W}$.  Therefore, if $2k>2M+2\tilde{M}+2L$, where $L$ is the injectivity length of $AB$, then 
  \begin{equation}
    V\underbrace{CD\dots C}_{2k+1}\tilde{W} = \underbrace{AB\dots A}_{2N+1} \underbrace{B\dots A}_{2k-1-2N} \underbrace{\tikz[baseline=0cm]{\draw (0,0) rectangle (1.5,0.3);}}_{2N+1} = \underbrace{\tikz[baseline=0cm]{\draw (0,0) rectangle (1.5,0.3);}}_{2N+1} \underbrace{B\dots A}_{2k-1-2N} \underbrace{B\dots A}_{2N+1}
  \end{equation}
  As the middle part is injective, the last equation can hold only if $AB\dots BA = \lambda V CD\dots DC \tilde{W}$. The other equation can be proven similarly.
\end{proof}

\section{The qubit case}\label{app:qubitclassify}

In this section, we characterize how two semi-injective PEPS defined by $(\ket{\phi_A}, O)$ and $(\ket{\phi_B},\id)$ with $\ket{\phi_A},\ket{\phi_B}\in \left(\mathbb{C}^2\right)^{\otimes 4}$ can generate the same state. We restrict ourselves to the case where $\ket{\phi_A}$ and $\ket{\phi_B}$ do not factorize in either direction. Using \cref{cor:qubit_slocc}, $\ket{\phi_A}$ and $\ket{\phi_B}$ are SLOCC equivalent, and thus we can suppose $\ket{\phi_A}=\ket{\phi_B}$ (by changing $O$). Notice that the state $\ket{\xi}$ appearing in the proof of \cref{thm:slocc} (see \cref{eq:schmidt_vert}) is also SLOCC equivalent with $\phi_A$. We can thus suppose that $\ket{\phi_A}=\ket{\phi_B}=\ket{\xi}$.  Therefore, given $\ket{\phi_A}$, we only need to characterize all pairs of two-body invertible operators such that \cref{eq:schmidt_vert} holds. 

Let us fix $\ket{\phi_A}=\ket{\phi_B}=\ket{\xi}$. We start the investigation with  a state such that in the horizontal cut it has Schmidt rank two. As the span of the Schmidt vectors contains a product state and we are only interested in $\ket{\phi_A}$ up to SLOCC equivalence, w.l.o.g.\ we can suppose that a basis of its reduced density on the upper two particles is  
\begin{align}
\ket{\Psi_1} &= \vert 00 \rangle\\
\ket{\Psi_2} &= a \vert 01 \rangle+b \vert 10 \rangle+c \vert 11 \rangle,
\end{align}
whereas a basis of its reduced density on the lower two particles is 
\begin{align}
\ket{\Phi_1} &= \vert 00 \rangle\\
\ket{\Phi_2} &= A \vert 01 \rangle+B \vert 10 \rangle+C \vert 11 \rangle\ .
\end{align}
In this setting, we are looking for invertible two body operators $O_1$ and $O_2$ such that 
  \begin{equation}
    \begin{tikzpicture}[baseline=-0.1cm]  
    \draw[thick,blue] (0,0.7)--(0,0.2)--(1,0.2)--(1,0.7); 
    \draw[thick,purple] (0,-0.7)--(0,-0.2)--(1,-0.2)--(1,-0.7); 
    \foreach \x in {(0,0.2),(0,-0.2),(1,0.2),(1,-0.2)} \filldraw \x circle (0.1);
    \draw[rounded corners=3pt,blue] (-0.2,0.7) rectangle (1.2,0.9);
    \draw[rounded corners=3pt,purple] (-0.2,-0.7) rectangle (1.2,-0.9);  
    \draw (0.5,0.9)--++(0,0.2);
    \draw (0.5,-0.9)--++(0,-0.2);
    \draw[red, thick] (-0.2,-0.4) rectangle (0.2,0.4);
    \draw[red, thick] (0.8,-0.4) rectangle (1.2,0.4);
    \end{tikzpicture}
    =
    \begin{tikzpicture}[baseline=-0.1cm]
    \def\c{blue}
    \draw[thick,\c] (0,0.7)--(0,0.2)--(1,0.2)--(1,0.7); 
    \draw[thick,purple] (0,-0.7)--(0,-0.2)--(1,-0.2)--(1,-0.7); 
    \foreach \x in {(0,0.2),(0,-0.2),(1,0.2),(1,-0.2)} \filldraw \x circle (0.1);
    \draw[rounded corners=3pt,\c] (-0.2,0.7) rectangle (1.2,0.9);
    \draw[rounded corners=3pt,purple] (-0.2,-0.7) rectangle (1.2,-0.9);  
    \draw[thick,\c] (0.5,0.9)--++(0,0.2);      
    \draw[thick,purple] (0.5,-0.9)--++(0,-0.2);
    \end{tikzpicture}\ ,
  \end{equation}
where the left red rectangle represents $O_1$, the right one $O_2$, while the blue Schmidt vectors are $\Psi_{1/2}$, the purple ones are $\Phi_{1/2}$. This gives four times sixteen equations on the matrix elements of $O_1$ and $O_2$. Checking these equations can be done in any CAS. The following cases can be distinguished.

\begin{itemize}
\item $C\neq 0$, $c\neq 0$. In this case, the operators are ($\alpha,\beta,\gamma$ are free parameters):
\begin{align}
O_1  =& \ket{00}\bra{00} + \alpha \ket{01}\bra{01} + \frac{b}{c}(\alpha-1)\ket{00}\bra{01} + \beta \ket{10}\bra{10}+ \frac{B}{C}(\beta-1)\ket{00}\bra{10} +  \nonumber\\
&\gamma \ket{11}\bra{11} + \frac{b}{c}(\gamma-\beta)\ket{10}\bra{11}+\frac{B}{C}(\gamma - \alpha) \ket{01}\bra{11}+ \frac{Bb}{Cc} (1+\gamma-\alpha-\beta)\ket{00}\bra{11}\\
O_2  =& \ket{00}\bra{00} + \frac{1}{\alpha} \ket{01}\bra{01} + \frac{a}{c}(\frac{1}{\alpha}-1)\ket{00}\bra{01} + \frac{1}{\beta} \ket{10}\bra{10}+ \frac{A}{C}(\frac{1}{\beta}-1)\ket{00}\bra{10} +  \nonumber\\ 
&\frac{1}{\gamma} \ket{11}\bra{11} + \frac{a}{c}(\frac{1}{\gamma}-\frac{1}{\beta})\ket{10}\bra{11}+\frac{A}{C}(\frac{1}{\gamma} - \frac{1}{\alpha}) \ket{01}\bra{11}+ \frac{aA}{Cc} (1+\frac{1}{\gamma}-\frac{1}{\alpha}-\frac{1}{\beta})\ket{00}\bra{11} 
\end{align}
\item $C\neq 0$, $c=0$. In this case $b\neq0$, otherwise the one particle reduced densities of the state is not full rank. The operators are ($\alpha,\beta,\gamma$ are free parameters):
\begin{align}
O_1  =& \ket{00}\bra{00} + \ket{01}\bra{01} + \alpha\ket{00}\bra{01} + \beta \ket{10}\bra{10}+ \frac{B}{C}(\beta-1)\ket{00}\bra{10} +  \nonumber\\
&\beta \ket{11}\bra{11} + \gamma\ket{10}\bra{11}+\frac{B(\beta - 1)}{C} \ket{01}\bra{11}+ \frac{B \gamma - B \alpha}{C}\ket{00}\bra{11}\\
O_2  =& \ket{00}\bra{00} + \ket{01}\bra{01} - \frac{a}{b}\alpha\ket{00}\bra{01} + \frac{1}{\beta} \ket{10}\bra{10}+ \frac{A}{C}(\frac{1}{\beta}-1)\ket{00}\bra{10} +  \nonumber\\ 
&\frac{1}{\beta} \ket{11}\bra{11} -\frac{a \gamma}{b \beta^2}\ket{10}\bra{11}+(-\frac{A}{C} + \frac{A}{C \beta}) \ket{01}\bra{11}+ \frac{A a \gamma - Aa \gamma/\beta^2}{C b}\ket{00}\bra{11} 
\end{align}
\item $c=0$, $C=0$. In this case $B\neq 0 , b\neq 0$, otherwise the one particle reduced densities of the state is not full rank. The operators are ($\alpha,\beta,\gamma$ are free parameters):
\begin{align}
O_1  =& \ket{00}\bra{00} + \ket{01}\bra{01} + \alpha\ket{00}\bra{01} + \ket{10}\bra{10}+ \beta\ket{00}\bra{10} +  \nonumber\\
& \ket{11}\bra{11}+\alpha\ket{10}\bra{11}+\beta \ket{01}\bra{11}+  \gamma \ket{00}\bra{11}\\
O_2  =& \ket{00}\bra{00} + \ket{01}\bra{01} - \frac{a}{b}\alpha\ket{00}\bra{01} + \ket{10}\bra{10}- \frac{B}{A}\beta\ket{00}\bra{10} +  \nonumber\\ 
& \ket{11}\bra{11} -\frac{a}{b} \alpha \ket{10}\bra{11}-\frac{A}{B} \beta\ket{01}\bra{11}+ \frac{2 A a \alpha \beta - Aa \gamma}{B b}\ket{00}\bra{11} 
\end{align}
\end{itemize}
Using this result, we have checked that if the state has Schmidt rank at least 3, then $O_1$ and $O_2$ can only be product operators. 

To find therefore all possible operators $O$ such that the semi-injective PEPS defined by $(\ket{\phi_A}, O)$ and by $(\ket{\phi_B},\id)$ are the same (supposing they have Schmidt rank at least two along both vertical and horizontal cut), one has to do the following steps:
\begin{enumerate}
\item Transform $\ket{\phi_A}$ with an invertible product  operator $O_1$ to have $\ket{00}$ in the span of its Schmidt vectors in both the upper and lower two particles
\item If the Schmidt rank of $\ket{\phi_A}$ is two along the horizontal cut, then take the two-body operators given above, $O_2\otimes O_3$. Otherwise take $O_2=O_3=\id$.
\item Repeat the previous two steps for the vertical cut, giving an invertible product operator $O_4$ and two-body invertible operators $O_5\otimes O_6$.
\item Find all invertible product transformation $O_7$ such that $\ket{\phi_B} = O_7 \ket{\phi_A}$.
\end{enumerate}
Then all possible operators $O$ are given by $\tilde{O}_1 (O_2\otimes O_3) \tilde{O}_1^{-1} \tilde{O}_4 (O_5\otimes O_6 ) \tilde{O}_4^{-1} \tilde{O}_7$, where $\tilde{O}_i = SO_iS$, and $S$ is the swap operator   defined before \cref{thm:canonicalform}.

\section{G-injective tensors}

In this section, we try to generalize semi-injective PEPS in a way that it also includes G-injective PEPS. We try the obvious generalization: if the semi-injective PEPS $(\ket{\phi},\id)$  has symmetries $O_g$ for $g\in G$ for some group $G$, then $\ket{\phi}$ and $\sum_g O_g$ defines a non-semi-injective PEPS that could be a candidate to include $G$-injective PEPS. We present here, however, an example for such a state, that behaves very different from $G$-injective PEPS.   

Consider the following state: 
  \begin{equation}
  \ket{\Psi} = 
  \begin{tikzpicture}
  \truncstateswo[green,red] (0,0) (4,3)
  \end{tikzpicture}\ ,
  \end{equation}
where the green rectangle is a four-partite GHZ state, and the red circle is $O = \id + Z^{\otimes 4}$. 

We will show that on an $n\times n$ torus, there are at least $2^n$ linearly independent states that are locally indistinguishable from this state. This means that given any local (frustration free) parent Hamiltonian, its ground space is at least $2^n$-fold degenerate. 

To see this, consider states on the torus that are constructed similar to $\ket{\Psi}$, except that some of the four-partite GHZ states $\ket{\phi^+} = 1/\sqrt{2}(\ket{0000}+\ket{1111})$ are changed to $\ket{\phi^{-}}=1/\sqrt{2}(\ket{0000}-\ket{1111})$. Such a state will be depicted schematically as a rectangular grid, with squares colored black at all occurrence of $\ket{\phi^{-}}$. For example, the figure below depicts such a state with one occurrence of $\ket{\phi^{-}}$:
\begin{equation}
\begin{tikzpicture}[baseline=0.5*1.8cm, scale = 0.5]
\draw (0,0) grid (4,4);
\draw[fill=black] (1,2) rectangle (2,3);
\end{tikzpicture}\ .
\end{equation}
We will see that these states are all locally indistinguishable from $\ket{\Psi}$ and that they span an at least $\exp\{n/2\}$-dimensional space. First notice that $\ket{\phi^{-}} = Z \ket{\phi^{+}}$, where $Z$ acts on one of the four particles (any one of them). Due to the special form of $O$, however, if in a $2\times 2$ block all $\ket{\phi^+}$ are changed to $\ket{\phi^-}$, it doesn't change the state: 
\begin{equation}
\ket{\Psi} = 
\begin{tikzpicture}[scale=0.5]
\draw (0,0) grid (4,4);
\end{tikzpicture} = 
\begin{tikzpicture}[scale=0.5]
\draw (0,0) grid (4,4);
\draw[fill=black] (1,1) rectangle (3,3);
\end{tikzpicture}\ .
\end{equation}
In fact, inverting the color of all rectangles in any $2\times 2$ rectangle doesn't change the state. For example,
\begin{equation}
\begin{tikzpicture}[scale = 0.5]
\draw (0,0) grid (4,4);
\draw[fill=black] (1,2) rectangle (2,3);
\end{tikzpicture}
=
\begin{tikzpicture}[scale = 0.5]
\draw (0,0) grid (4,4);
\draw[fill=black] (1,1) rectangle (2,2);
\draw[fill=black] (2,1) rectangle (3,2);
\draw[fill=black] (2,2) rectangle (3,3);
\end{tikzpicture} \ .
\end{equation}
A consequence of this is that a pair of black rectangles in the same column (row) can ``travel'' horizontally (vertically) no matter how far they are separated. As an illustration, let us show how to move two black rectangles in the same column separated by one to the neighboring column:
\begin{equation}
\begin{tikzpicture}[scale = 0.5]
\draw (0,0) grid (4,4);
\draw[fill=black] (1,2) rectangle (2,3);
\draw[fill=black] (1,0) rectangle (2,1);
\end{tikzpicture}
=
\begin{tikzpicture}[scale = 0.5]
\draw (0,0) grid (4,4);
\draw[fill=black] (1,2) rectangle (2,3);
\draw[fill=black] (1,1) rectangle (2,2);
\draw[fill=black] (2,0) rectangle (3,1);
\draw[fill=black] (2,1) rectangle (3,2);
\end{tikzpicture}
=
\begin{tikzpicture}[scale = 0.5]
\draw (0,0) grid (4,4);
\draw[fill=black] (2,2) rectangle (3,3);
\draw[fill=black] (2,0) rectangle (3,1);
\end{tikzpicture}  \ .
\end{equation}
This means that these states are indistinguishable from $\ket{\Psi}$ on any finite (system size independent) region. Inverting the color of all rectangles in any $2\times 2$ rectangle in fact defines an equivalence relation on the colorings of the grid: two colorings are equivalent if and only if they can be transformed to each other by repeatedly inverting the color of all rectangles in $2\times 2$ regions. Equivalent colorings correspond to the same state, whereas inequivalent colorings to perpendicular ones: such states all have the form $(1+Z^{\otimes 4})^{\otimes n} \ket{\phi^\pm}^{\otimes n}$. Expanding this expression, we get a sum of tensor products of $\ket{\phi^+}$ and $\ket{\phi^-}$. Starting from two equivalent colorings, the sum contains the same terms reordered. Starting from inequivalent colorings, all terms differ from each other and thus the states are perpendicular as $\scalprod{\phi^+}{\phi^-}=0$. To see that there are at least $2^n$ equivalence classes, notice that the parity of black rectangles in each column is an invariant. (And for each parity assignment there is a coloring with that parities.)

\bibliography{library}

\end{document}